\documentclass[11pt]{article}

\usepackage{amssymb, amsmath, verbatim, amsthm,url, multirow,fullpage,mathtools}
\usepackage{longtable, rotating,makecell,array}
\usepackage[aligntableaux=top]{ytableau}

\usepackage{soul}
\usepackage[colorinlistoftodos,textsize=footnotesize]{todonotes}

\setstcolor{red}

\usepackage[all]{xy}
\xyoption{poly}
\xyoption{arc}

\usetikzlibrary{calc} 
\usetikzlibrary{decorations.pathmorphing} 
\usetikzlibrary{bending} 
\usetikzlibrary{patterns} 

\newcommand{\leplus}{\Large $+$}
\newcommand{\lezero}{\Large $0$}

\newcommand{\leplusbold}[1][black]{%
  \tikz\draw[#1,line width=1.5pt,scale = 0.35,line cap = round] (0,0) -- (1,0)(0.5,0.5) -- (0.5,-0.5);
}

\definecolor{light-gray}{gray}{0.6}

\tikzstyle{propagator}=[decorate,decoration={snake,amplitude=0.8mm}]
\tikzstyle{smallpropagator}=[decorate,decoration={snake,segment length=3mm,amplitude=0.5mm}]

\tikzstyle{linehighlight}=[blue,line width = 3pt,line cap = round, draw opacity = 0.5]

\tikzstyle{firstdash}=[dashed,line cap=round, dash pattern=on 2pt off 1pt]
\tikzstyle{seconddash}=[dashed,line cap=round, dash pattern=on 0.5pt off 1pt]
\tikzstyle{smalldash}=[dashed,line cap=round, dash pattern=on 1.5pt off 2pt]

\pgfmathsetmacro{\arrowangle}{90}
\pgfmathsetmacro{\diagramscale}{1.3}
\tikzstyle{propassignment} = [->,shorten >=2pt,thick]

\newcommand{\drawWLD}[2]{

\pgfmathsetmacro{\n}{#1}
\pgfmathsetmacro{\radius}{#2}
\pgfmathsetmacro{\angle}{360/\n}
\draw (0,0) circle (\radius);
    \foreach \i in {1,2,...,\n} {
      \draw (\angle*\i:\radius) node {$\bullet$};
    }

}

\newcommand{\drawprop}[4]{
\pgfmathsetmacro{\r}{#1}
\pgfmathsetmacro{\bumpr}{#2}
\pgfmathsetmacro{\s}{#3}
\pgfmathsetmacro{\bumps}{#4}
\pgfmathsetmacro{\perturbe}{\angle/\n}
\begin{scope}
\draw[smallpropagator] (\angle*\r + \angle/2 + \bumpr*\perturbe:\radius) -- (\angle*\s + \angle/2 + \bumps*\perturbe:\radius);
\end{scope}
}

\newcommand{\drawlabeledprop}[5]{
\pgfmathsetmacro{\r}{#1}
\pgfmathsetmacro{\bumpr}{#2}
\pgfmathsetmacro{\s}{#3}
\pgfmathsetmacro{\bumps}{#4}
\pgfmathsetmacro{\perturbe}{\angle/\n}

\begin{scope}
\draw[smallpropagator] (\angle*\r + \angle/2 + \bumpr*\perturbe:\radius) -- (\angle*\s + \angle/2 + \bumps*\perturbe:\radius) node[midway, below] {#5};
\end{scope}
}

\newcommand{\drawnumbers}{
  \foreach \i in {1,2,...,\n} {
  \pgfmathsetmacro{\x}{\angle*\i}
  \draw (\x:\radius*1.25) node {\footnotesize \i};
}
}

\def\centerarc[#1](#2)(#3:#4:#5)
    { \draw[#1] ($(#2)+({#5*cos(#3)},{#5*sin(#3)})$) arc (#3:#4:#5); }

\def\clipcenterarc(#1)(#2:#3:#4)
    { \clip ($(#1)+({#4*cos(#2)},{#4*sin(#2)})$) arc (#2:#3:#4); }

\newcommand{\drawWLDfragment}[3][10]{
\pgfmathsetmacro{\n}{#1} 
\pgfmathsetmacro{\radius}{#2}
\pgfmathsetmacro{\fragment}{#3} 
\pgfmathsetmacro{\halfangle}{360*\fragment/2}
\pgfmathsetmacro{\startpoint}{270 - \halfangle}
\pgfmathsetmacro{\endpoint}{270 + \halfangle}
\pgfmathsetmacro{\step}{2*\halfangle/\n} 
\pgfmathsetmacro{\zero}{\startpoint-0.5*\step} 
\centerarc[black](0,0)(\startpoint:\endpoint:\radius)
}

\newcommand{\newnode}[3][left]{
  \node[label={[label distance=-1mm]#1:{\scriptsize $#3$}}] at (\zero + #2*\step:\radius) {\scriptsize $\bullet$};
}

\newcommand{\newbetternode}[3][{label distance=-1mm]left}]{
  \node[label={#1:{\scriptsize $#3$}}] at (\zero + #2*\step:\radius) {\scriptsize $\bullet$};
}

\newcommand{\newprop}[4][midway,below]{
\pgfmathsetmacro{\startnode}{#2}
\pgfmathsetmacro{\endnode}{#3}

\draw[smallpropagator] (\zero+\startnode*\step:\radius) -- (\zero + \endnode*\step:\radius) node[#1] {#4};
}

\newcommand{\newpropbend}[3]{
\draw[smallpropagator] (\zero+#1*\step:\radius*1.1) to[bend left = #3] (\zero + #2*\step:\radius*1.1);
}

\newcommand{\R}{\mathbb{R}}

\newcommand{\Gr}{\mathbb{G}_{\R, \geq 0}}

\def\ba #1\ea{\begin{align} #1 \end{align}}
\def\bas #1\eas{\begin{align*} #1 \end{align*}}
\def\bml #1\eml{\begin{multline} #1 \end{multline}}
\def\bmls #1\emls{\begin{multline*} #1 \end{multline*}}

\newcommand{\cP}{\mathcal{P}}

\newcommand{\cB}{\mathcal{B}}

\newcommand{\Prop}{\textrm{Prop}}

\newcommand{\cZ}{\mathcal{Z}}

\newcommand{\gale}[1]{\preccurlyeq_{#1}}

\newtheorem{thm}{Theorem}[section]

\newtheorem{lem}[thm]{Lemma}
\newtheorem{cor}[thm]{Corollary}
\newtheorem{prop}[thm]{Proposition}
\newtheorem{algorithm}[thm]{Algorithm}

\theoremstyle{remark}
\newtheorem{eg}[thm]{Example}

\theoremstyle{definition}
\newtheorem{dfn}[thm]{Definition}
\newtheorem{rmk}[thm]{Remark}

\title{Combinatorics of the geometry of Wilson loop diagrams II: Grassmann necklaces, dimensions, and denominators}
\author{Susama Agarwala\thanks{SA was partially supported by an Office of Naval Research grant.}, Si\^an Fryer, and Karen Yeats\thanks{KY is supported by an NSERC Discovery grant, by the Canada Research Chair program, and also, through some the time this work was developed, by a Humboldt Fellowship from the Alexander von Humboldt foundation.}}
\date{}

\begin{document}
\maketitle

\begin{abstract}
  Wilson loop diagrams are an important tool in studying scattering amplitudes of SYM $N=4$ theory and are known by previous work to be associated to positroids. In this paper we study the structure of the associated positroids, as well as the structure of the denominator of the integrand defined by each diagram. We give an algorithm to derive the Grassmann necklace of the associated positroid directly from the Wilson loop diagram, and a recursive proof that the dimension of these cells is thrice the number of propagators in the diagram. We also show that the ideal generated by the denominator in the integrand is the radical of the ideal generated by the product of Grassmann necklace minors.
\end{abstract}
\section{Introduction}

This paper is the second in a two part series investigating the combinatorics and geometry underlying Wilson loop diagrams in SYM $N=4$ theory. This series lays out several results about the relationship between Wilson loop diagrams and the positroid cells that give a CW complex structure to the positive Grassmannian. This paper is concerned with identifying which positroid cells correspond to which Wilson loop diagrams, and how this translates to a geometric tiling of a subspace of the positive Grassmannians. 



In contrast, the first paper in this series focuses on combinatorial and matroidal properties of the Wilson loop diagrams. The main focus of that paper is on enumerating the positroid cells that are associated to Wilson loop diagrams with a fixed number of vertices and propagators, as well as counting the number of Wilson loop diagrams that map to the same positroid. The results of the first paper are not a prerequisite for this paper.

In recent years, there has been significant interest in understanding the geometry and combinatorics underlying the field theory SYM $N=4$ \cite{wilsonloop, Arkani-Hamed:2013jha, Amplituhedronsquared, galashinlam, AmplituhedronDecomposition}. This started with the observation that BCFW diagrams, which represent the on shell interactions of this theory, correspond to plabic graphs \cite{GrassmannAmplitudebook}. Plabic graphs were introduced by Postnikov in \cite[section 11]{Postnikov}, and provide one of the many equivalent ways to classify the positroid cells of the positive real Grassmannian $\Gr(k,n)$. This led to the representation of the on shell integrals associated to this theory in terms of a geometric space (the Amplituhedron) embedded in $\Gr(k,n)$ \cite{Arkani-Hamed:2013kca,Arkani-Hamed:2013jha}. Moreover, integrands associated to the BCFW diagrams define a volume form on the Amplituhedron. This gives a novel and elegant interpretation of scattering amplitudes for this theory, as well as a means of greatly simplifying the calculations involved in traditional Feynman diagram approaches. Since then, there has been significant work devoted to understanding the volume form and its physical implications \cite{UnwindingAmplituhedron,galashinlam}.

Meanwhile, a different body of work has studied SYM $N=4$ theory from the point of view of Feynman integrals in twistor space \cite{Adamo:2011pv,Adamo:2012xe,  Boels:2007qn, Bullimore:2010pj}. These integrals are calculated via holomorphic Wilson loops. As with the Amplituhedron, each Wilson loop represents the sum of a family of $N^kMHV$ Feynman diagrams. Each of the Feynman integrands in this sum associates a volume form to a convex space, parallel to the Amplituhedron \cite{CachazoSvrcekWitten, Hodges2013}. However, the geometric objects decomposing the Wilson loops are different from the geometric object decomposing the Amplituhedron.

Wilson loop diagrams, as we will study them, are objects given by a perturbative expansion of the holomorphic Wilson loop in twistor space.
Classically, the interactions of the theory were studied as maximal helicity violating diagrams ($MHV$ diagrams), next-to-maximal helicity violating diagrams, and so on, ($N^kMHV$ diagrams) in momentum space. Adamo and Mason, in \cite{Adamo:2011pr}, gave an elegant formulation for the $N^kMHV$ diagrams in twistor space. The Wilson loop diagram perspective of this paper is dual to this twistor representation of $N^kMHV$ diagrams.
The objects are closely related and carry the same information.  We choose to emphasize the dual perspective by our focus on the Wilson loop diagrams.


In \cite{wilsonloop}, Agarwala and Marin-Amat uncovered a connection between the Wilson loop diagrams and positroids, defined as matroids that can be realized as an element of $\Gr(k,n)$. Other work has shown a connection between the volume filled by Wilson loop diagrams and the Amplituhedron \cite{casestudy, Amplituhedronsquared}. Specifically, Marcott showed that the space parametrized by the Wilson loop diagrams is  exactly the correct positroid cell in $\Gr(k,n)$, up to a set of measure zero \cite{WLDdim}. 

In the above work and elsewhere, it is clear that while the space parametrized by the Wilson loop diagrams is related to the Amplituhedron, they are not the same space, nor is a direct connection between the geometry of the Amplituhedron and the geometry of the Wilson loop diagrams known. For instance, it is conjectured that the Amplituhedron is orientable, whereas Agarwala and Marcott showed in \cite{non-orientable} that the space parametrized by the Wilson loop diagrams as a subspace of $\Gr(k, n+1)$ can be seen as a non-orientable vector bundle over some submanifold of $\Gr(k,n)$. Furthermore, the cells associated to the Wilson loop diagrams are different from the cells associated to BCFW diagrams, which can be seen by comparing the results of \cite{AmplituhedronDecomposition} and \cite{casestudy}. Furthermore, the spaces in the two situations have different dimensions, as we will discuss further below, and examples show \cite{HeslopStewart} that they do not relate by a coordinate projection or boundary map or other straightforward function.  None the less, from their common physical origin, it is believed that once the integrals over these spaces are taken the results will parametrize related spaces \cite{wilsonloop, Amplituhedronsquared}. 

The Amplituhedron and the Wilson loop diagrams describe the scattering amplitudes from the same physical theory, but the Amplituhedron only captures the on shell portion while the Wilson loop diagrams capture arbitrary kinematics.   Both are related to positroids, and in this way they have similar combinatorics, but, these combinatorial objects manifest themselves differently in the two cases.  For instance, square moves and merge-expand moves on the plabic graph give different graphical representations, hence parametrizations, of the same positroid cell in the Amplituhedron case.  On the other hand, as we show in our companion paper \cite{generalcombinatoricsI}, retriangulations of polygon dissections associated to Wilson loop diagrams give all the Wilson loop diagrams associated to the same positroid cell.  Each distinct Wilson loop diagram associated to the same positroid cell gives a different parametrization \cite{WLDdim}. 
The deeper combinatorial and geometric connection between the Amplituhedron and the Wilson loop diagrams, which physics tells us should exist, remains to be uncovered in future work.

This paper continues the program of establishing concrete connections between the combinatorics of the Wilson loop diagram as determined by the physics, and the combinatorics of the associated positroid cell as identified in \cite{wilsonloop}. In this paper, we give an algorithm to read the Grassmann necklace, one of the standard combinatorial characterizations of the positroid cell, directly off the Wilson loop diagram. We also show that the positroid cells associated to Wilson loop diagrams are all $3k$-dimensional, where $k$ is the number of propagators in the diagram; by contrast, the BCFW cells are all $4k$-dimensional \cite{Arkani-Hamed:2013jha}. In \cite{non-orientable} the authors noted that the Wilson loop diagrams only correspond to positroid cells when the gauge vector is ignored. Including the gauge vector defines a $4k$ dimensional subspace of $\mathbb{G}_\R(k,n+1)$.

As described above, each integrand associated to a Wilson loop diagram gives the volume of a convex polytope. However, each individual tree level integral has singularities. These singularities are of two types: physical and spurious.  The physical singularities come from special kinematical points and we will not consider them in this paper.  The spurious singularities come from zeros of the denominator of the integral, and hence are accessible to our combinatorial techniques.  SYM $N=4$ is a finite theory so we only expect to see the physical singularites in the overall amplitude.  Consequently, the spurious singularities are artifacts of how the integrals are represented and must cancel in the calculation of a tree level scattering amplitude, though the details of how the cancellations come about remain unclear. The cancellation has been conjectured in general \cite{Hodges2013} and verified explicitly in a few cases \cite{casestudy, HeslopStewart}, but it is hard to prove in general due to the fact that the relationship between Wilson loop diagrams and $3k$-dimensional positroid cells of $\Gr(k,n)$ is neither one-to-one nor onto \cite{generalcombinatoricsI}.


In this paper, we give a geometric characterization of the factors appearing in the denominator of these Wilson loop diagram integrals in terms of the Grassmann necklace of the associated positroid. Specifically, we show that the set of irreducible factors in the denominator, as defined by the physics \cite{HeslopStewart}, are exactly the irreducible factors of the minors defined by the Grassmann necklace evaluated on a physically significant realization of the associated positroid cell. In other words, the singularities arising from tree level particle interactions in SYM $N=4$ theory do not just come from the Feynman rules describing the theory: they can also be seen directly from the associated geometry, given the physically correct representation.  We hope these results will pave the way for a general study of these poles in future.

\subsection{Roadmap}

In Section~\ref{section general background}, we summarize the required background for Wilson loop diagrams (subsection~\ref{section WLD background}), the theory of positroids in $\Gr(k,n)$ and their characterization in terms of Grassmann necklaces and Le diagrams (subsection~\ref{sec:positroid background}), and the link between Wilson loop diagrams and positroids (subsection~\ref{sec:WLD as positroids}). 

In Section~\ref{sec GN algorithm} we state an algorithm for constructing the Grassmann necklace of the positroid associated to a Wilson loop diagram (Algorithm~\ref{alg:put GN on WLD}) and prove its correctness (Theorem~\ref{res:alg gives GN}, Theorem~\ref{res alg gives correct GN}). While it was already possible to construct the Grassmann necklace of a Wilson loop diagram via existing bijections between some of the various combinatorial objects that index the positroid cells, this process was convoluted and involved multiple steps; our algorithm allows us to track directly how each propagator in the Wilson loop diagram contributes to each term in the Grassmann necklace. This makes it much easier to track how properties move through the bijection, as we do throughout the later sections. Furthermore, the above algorithm allows us to better understand the positroids associated to Wilson loop diagrams. In particular, Corollary \ref{no coloops} shows that these positroids have no coloops, and gives an exact characterization of their loops. Many of the technical lemmas required to prove Theorem~\ref{res:alg gives GN} and Theorem~\ref{res alg gives correct GN} continue to be important in subsequent sections.  In particular we have Lemma~\ref{lem sian}, which identifies certain configurations of propagators that must appear in any admissible Wilson loop diagram, and Lemma~\ref{vertex cyclic int lem}, which shows that the pattern of values that a propagator contributes to each term of the Grassmann necklace is both simple and predictable.

In Section~\ref{sec dim} we examine the dimension of positroids associated to Wilson loop diagrams, and show that the dimension is always equal to three times the number of propagators in the diagram (Theorem~\ref{thm dim}).  While this has been proved by Marcott in \cite{WLDdim} using matroidal arguments, our proof is constructive and explicitly relates the position of the plusses in the Le diagram to the propagators in the Wilson loop diagram. This proof and the methods used therein show how the positroid cell associated to a Wilson loop diagram changes as one adds propagators in specific positions. This is a significant step towards understanding the broader question about how adding propagators to a Wilson loop diagram changes the associated geometry.

Finally, in Section~\ref{sec poles} we characterize the denominator of the integrand of a Wilson loop diagram in terms of its associated matrix and Grassmann necklace. We show that given the representation of the positroid cell associated to an admissible Wilson loop diagram, the irreducible factors of the denominator are exactly the set of irreducible factors of the minors of said representation given by the Grassmann necklace. This characterization is given explicitly by Algorithm~\ref{alg WLD to denom via GN} and Theorem~\ref{thm denom}.

\section{Background}\label{section general background}

\subsection{Wilson loop diagrams}\label{section WLD background}

\begin{dfn}\label{WLdfn}
A {\bf Wilson loop diagram} $W = (\cP,V)$ consists of a cyclically ordered set $V$ with one distinguished element (giving a compatible linear ordering on $V$ with the distinguished element as the first element), and a set $\cP$ of unordered pairs of elements of $V$. The elements of $V$ are the {\bf vertices} of $W$, and elements of $\cP$ are the {\bf propagators} of $W$. \end{dfn}

A Wilson loop diagram is depicted as a circle, with the vertices $V$ arranged along the edge of the circle and listed counterclockwise in their cyclic order. The arc between two consecutive vertices is referred to as an {\bf edge} of the diagram, and the $i$th edge of $W$ is the edge from vertex $i$ to its successor. Each propagator $(i,j) \in \cP$ is denoted by a wavy line joining edge $i$ and edge $j$ inside the circle.

\begin{figure}[h]
\[W\ =\ \begin{tikzpicture}[rotate=67.5,baseline=(current bounding box.east)]
	\begin{scope}
	\drawWLD{8}{2}
	\drawnumbers
	\drawprop{1}{0}{4}{0}
	\drawprop{2}{0}{4}{-1}
    \drawprop{5}{0}{8}{0}
    \drawprop{5}{1}{7}{0}
		\end{scope}
	\end{tikzpicture}\]

\caption{The Wilson loop diagram $W = \big(\{(1,4), (2,4), (5,7), (5,8)\},[8]\big)$ which has four propagators and eight vertices.}
\label{fig:ex of WLD}
\end{figure}
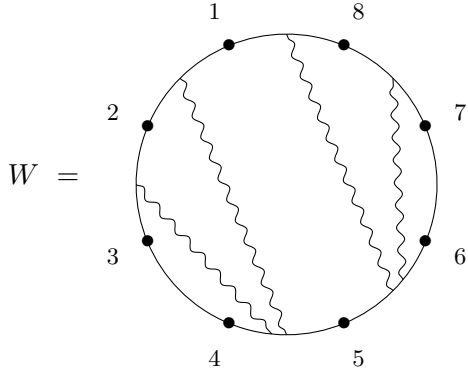

Let $[n]$ denote the set of integers $\{1,\dots,n\}$ with the obvious cyclic and linear orders.

Note that the propagators in a Wilson loop diagram are undirected, that is, $\cP$ defines a symmetric binary relation on $V$.  However, in some of the later arguments we need to impose a direction on a propagator $p = (i,j)$, for example in order to consider the region ``inside'' or ``outside'' of $p$. In this case, we write $(p,<_i)$ to indicate that $p$ is directed from edge $i$ to edge $j$, and $(p,<_j)$ for the opposite direction.

We take the convention that all vertices should always be interpreted in terms of the cyclic order on $V$.  So, for any $i\in V$ we will use $i+1$ to denote the successor of $i$ in the cyclic order on $V$.   In particular, if $V = [n]$ then all vertices are interpreted modulo $n$ and $(n,n+3)$, for example, means the propagator connecting edges $3$ and $n$.

On both of the last two points, we differ from the standard convention in the physics literature, where propagators are written as ordered pairs $(i,j)$ with $1 \leq i<j \leq n$.  Additionally, note that in this paper, we allow for the propagator set $\cP$ to be empty. This is for mathematical convenience, even though these diagrams do not correspond to tree level scattering amplitudes.  Additionally, since $\cP$ is defined to be a set of pairs, parallel propagators (propagators starting and ending on the same pair of edges) are already excluded from this formulation, but nothing of physical importance is lost as the admissibility restriction, Definition~\ref{admisdfn}, would have excluded them on physical grounds.

\begin{dfn} \label{VPropdfn}
Let $W = (\cP, V)$ be a Wilson loop diagram.
\begin{enumerate}
\item For $p = (i,j) \in \cP$, let $V(p) = \{i, i+1, j, j+1\}$ denote the set of {\bf vertices supporting} $p$. For a set of propagators $P \subseteq \cP$, define 
\[V(P) = \bigcup_{p \in P} V(p)\]
to be the {\bf vertex support} of $P$.
\item For $U \subseteq V$, the set of {\bf propagators supported on $U$} is denoted by
\[\Prop(U) = \{ p \in \cP\ |\ V(p) \cap U \neq \emptyset \}.\]
Note that a propagator in $\Prop(U)$ does not necessarily have its entire support contained in $U$. 
\item Vertices of $W$ that do not support any propagators are called {\bf non-supporting}.
\end{enumerate}
\end{dfn}

We will see in later sections that all four vertices supporting a propagator are important in understanding the contribution of that propagator to the objects represented by the diagram. This motivates the following definition of a subdiagram, which can consist of any subset $Q$ of the propagators and any subset $U$ of the vertices of $W$ but $U$ {\em must contain the support of $Q$}.

\begin{dfn} \label{subdiagramdfn}
Let $W = (\cP, V)$ be a Wilson loop diagram. Then $W' = (Q,U)$ is a {\bf subdiagram} of $W$ if $Q \subseteq \cP$ and $V(Q) \subseteq U \subseteq V$.
\end{dfn}

With these definitions in hand, we can impose some conditions on the density and behaviour of the propagators of Wilson loop diagrams.

\begin{dfn}\label{admisdfn}
A Wilson loop diagram $W = (\cP, V)$ is called {\bf admissible} if the following conditions hold.
\begin{enumerate}
\item $|V| \geq |\cP| + 4$.
\item For each non-empty subset of propagators, $Q \subseteq \cP$, we have $|V(Q)|\geq |Q| + 3$.
\item There are no crossing propagators. That is, there does not exist ${\{(i,j),(k,l)\} \subseteq \cP}$ such that $i < k < j < l$.
\end{enumerate}
 \end{dfn}

The first condition bounds the total number of propagators in the diagram, while the second limits how  densely  propagators  can  be  fitted  into  any  part  of  the  diagram.  The third condition simply requires that propagators are non-crossing.

All three of these conditions arise from the physical interpretation of Wilson loop diagrams.  The first condition comes from the fact that these diagrams represent $N^k$MHV interactions, where $k +4 <n$, see for instance the explanation in the introduction of \cite{Adamo:2011pr}.  The second condition can be derived from the integrals associated to the diagrams, as in \cite[Theorem 1.13]{wilsonloop}. Note that the second condition prohibits propagators that connect two adjacent edges. In our definition, $\cP$ is a set so we may not have two propagators starting and ending on the same pair of edges, but we have lost nothing of physical value by not allowing multiple propagators ending on the same pair of edges since these would also fail the second condition. The third condition comes from the fact that these Wilson loop diagrams are a representation of planar Feynman diagrams of $N=4$ SYM theory. For a description of the relationship between planar Feynman diagram and planar Wilson loops, see the discussion in \cite[section 3]{LipsteinMason}.

Note that a subdiagram $(Q,U)$ of an admissible Wilson loop diagram need not be admissible itself.  It inherits conditions 2 and 3 automatically, but we could have $|U| = |Q| + 3$. If a Wilson loop diagram satisfies conditions 2 and 3 of Definition~\ref{admisdfn}, we call it {\bf weakly admissible}.

These diagrams correspond to (tree level) amplitudes of SYM $N=4$ theory in twistor space. As such, a diagram associates a real value to $n$ particles, one for each vertex. Each particle is represented as a section of a $|\cP|$-vector bundle over twistor space, projected onto a real subspace in a process called bosonization \cite[section 2]{Amplituhedronsquared}. In other words, each Wilson loop diagram defines a functional on external particles. This is represented as an integral $I(W)$, which is the equivalent of a Feynman integral in this setting. This integral is built from two main components: an interaction matrix $C(W)$, and a volume form $\Omega(W)$. These are described next.

A $|\cP| \times n$ matrix $C(W)$ represents the (tree level) particle interactions. Each row of $C(W)$ corresponds to an internal propagator, and each column to an external particle (i.e. a vertex). The matrix $C(W)$ is defined as follows:

\begin{equation} C(W)_{p,q} = \begin{cases} x_{p,q} & \textrm{ if } q \in V(p) \\
0  & \textrm{ if } q \not \in V(p)  \end{cases}
\;. \label{C(W) dfn}\end{equation} 

The entries $x_{p,q}$ are nonzero real variables. Thus the zero entries in row $p$ correspond to the external particles that the propagator $p$ does not interact with. See Example~\ref{integraldetails} for an example of a diagram $W$ and its associated matrix.

The matrix $C(W)$ defines a positroid cell associated to $W$ in the positive Grassmannian $\Gr(n, |\cP|)$, which we denote by $\Sigma(W)$ \cite[Theorem 3.41]{wilsonloop}. Marcott \cite[Theorem 8.4]{WLDdim} showed that  the matrix $C(W)$ parametrizes a space whose closure is exactly the closure of $\Sigma(W)$. That is, the space defined by $C(W)$ and the cell $\Sigma(W)$ agree up to a set of measure $0$. Most of the work in this paper and in the previous paper in the series \cite{generalcombinatoricsI} focuses on characterizing the positroid cell associated to each Wilson loop diagram $W$. 

Note that since there is no ordering on the propagators of $W$, $C(W)$ is only defined up to rearrangement of the rows. When needed, we will impose various different orders on these rows to suit the task at hand. The choice of ordering does not affect $\Sigma(W)$.


One may also interpret the functional $I(W)$ as associating a volume to each cell $\Sigma(W)$ \cite{wilsonloop, Amplituhedronsquared, HeslopStewart}, where the volume is determined by the external particles involved. Specifically, the integrand is written in terms of the elements of $C(W)$ as  \begin{equation} \label{eq:omega(W) def}\Omega(W) = \frac{\prod_{r=1}^{|\cP|} \prod_{v \in V(p_r)} \textrm{d}x_{p_r}}{R(W)} \;, \end{equation} multiplied by a $\delta$ function that evaluates the matrix $C(W)$ on the external particles. The denominator $R(W)$ is given in Definition~\ref{def R(W)} below.

The differential form $\Omega(W)$ is playing a similar role to that of volume forms in the Amplituhedron literature.
For the purposes of this paper, we restrict our attention to $\Omega(W)$ only, and ignore the $\delta$ function in the integrand of $I(W)$. The interested reader is referred to \cite{Adamo:2012xe,HeslopStewart,LipsteinMason} for more information about the integrand of $I(W)$. 

For a fixed number of propagators and particles, the tree level amplitude of the Wilson loop diagrams is computed as the sum of the $I(W)$ for all $W$ of the correct size. By {\cite{Adamo:2011pv,Britto:2005fq,Arkani-Hamed:2013jha}} this sum is finite for all sets of external particles, i.e.
\bas \mathcal{A}_{k, n}^{tree} = \sum_{\substack{W = (\cP, [n]), \\ |\cP| = k}} I(W) <\infty. \eas

This would certainly hold if the singularities in the integrands $\Omega(W)$ cancel out in the sum.  The singularities coming from the zeros of $R(W)$ are the spurious poles of the theory, as mentioned in the introduction. The physical theory tells us that the spurious poles should cancel. 
Specifically, it is conjectured that these poles cancel on the boundaries of the positroid cells $\Sigma(W)$, analogous to the analysis in \cite{Arkani-Hamed:2013jha}. This has been checked for particular classes of examples (e.g. \cite{casestudy, Amplituhedronsquared, HeslopStewart}) but has not been proven in general.

With this motivation in mind, in Section~\ref{sec poles} of this paper we study the denominator $R(W)$ of $\Omega(W)$. The definition of $R(W)$ is as follows:

\begin{dfn}\label{def R(W)}
Let $W = (\cP,V)$ be an admissible Wilson loop diagram, and fix an edge $e$. Let $\{p_1 \ldots p_r\}$ be the propagators with one end lying on edge $e$ under a counterclockwise ordering. Define the polynomial
\[ R_e =  x_{p_1,e+1} \prod_{j= 1}^{r-1} \left((x_{p_j,e} x_{p_{j+1},e+1} - x_{p_{j+1},e} x_{p_{j},e+1} ) \right) x_{p_r,e}\;\]
to be the component of $R(W)$ associated to edge $e$. Note that if $r = 1$ (i.e. there is exactly one propagator lying on edge $e$) then this expression simplifies to $R_e = x_{p,e} x_{p,e+1}$. If $r=0$, set $R_e = 1$. The denominator $R(W)$ is defined to be the product of all of the $R_e$, i.e. 
\[R(W) = \prod_{e \in V} R_e.\]
\end{dfn}

This formulation for the denominator comes from the physics literature, see equation (2) in \cite{HeslopStewart}. In Section \ref{sec poles}, we give a geometric meaning to the factors $\{R_e\ |\ e \in V\}$ of $R(W)$.  Specifically, given a Grassmann necklace representation of the positroid cell $\Sigma(W)$ (see Definition~\ref{def:grassmann necklace} and Algorithm \ref{alg:put GN on WLD}), the $R_e$ are exactly the irreducible factors of the minors of $C(W)$ defined by the Grassmann necklace. In particular, the question of \emph{which} minors of $C(W)$ contribute to $R(W)$ is geometrically determined. 
 
We end this section by computing $C(W)$ and $R(W)$ for a particular Wilson loop diagram $W$.

\begin{eg} \label{integraldetails}
Consider the admissible Wilson loop diagram 
\[W = \big(\{(1,4), (2,4), (5,7), (5,8)\},[8]\big),\]
from Figure~\ref{fig:ex of WLD}. Ordering the propagators as listed above, we obtain the associated matrix
\[ C(W) = \left(
\begin{array}{cccccccc}
x_{1,1} & x_{1,2} & 0 & x_{1,4} & x_{1,5} & 0 & 0 & 0 \\
0 & x_{2,2} & x_{2,3} & x_{2,4} & x_{2,5} & 0 & 0 & 0 \\
0 & 0 & 0 & 0 & x_{3,5} & x_{3,6} & x_{3,7} & x_{3,8} \\
x_{4,1} & 0 & 0 & 0 & x_{4,5} & x_{4,6} & 0 & x_{4,8}  \\
\end{array}
\right) \;.\]

To calculate $R(W)$, we first compute the polynomials associated to each edge of $W$: \bmls R_1 = x_{1,1} x_{1, 2}\quad R_2 = x_{2,2} x_{2,3} \quad R_4 = x_{2,5} (x_{2,4}x_{1,5} - x_{2,5}x_{1,4})x_{1,4} \\ R_5 = x_{4,6} (x_{4,5}x_{3,6} - x_{3,5}x_{4,6})x_{3,5} \quad R_7 = x_{3,7} x_{3,8} \quad R_8 = x_{4,8} x_{4,1} \;.\emls All other $R_e$ are $1$.  Finally, the integrand associated to this Wilson loop diagram (recall equation~\eqref{eq:omega(W) def}) is \bas \Omega(W) = \frac{\textrm{d}x_{1,1}\textrm{d}x_{1,2}\textrm{d}x_{1,4}\textrm{d}x_{1,5}\textrm{d}x_{2,2}\textrm{d}x_{2,3}\textrm{d}x_{2,4}\textrm{d}x_{2,5}\textrm{d}x_{3,5}\textrm{d}x_{3,6}\textrm{d}x_{3,7}\textrm{d}x_{3,8}\textrm{d}x_{4,5}\textrm{d}x_{4,6}\textrm{d}x_{4,8}\textrm{d}x_{4,1}}{x_{1,1} x_{1, 2}x_{2,2} x_{2,3}x_{2,5} (x_{2,4}x_{1,5} - x_{2,5}x_{1,4})x_{1,4}x_{4,6} (x_{4,5}x_{3,6} - x_{3,5}x_{4,6})x_{3,5}x_{3,7} x_{3,8}x_{4,8} x_{4,1}} \;. \eas
\end{eg}

In order to compute the contribution of this diagram to the tree level scattering amplitude one needs only to localize the differential form $\Omega(W)$ with respect to a delta function that evaluates the matrix $C(W)$ at the external momenta.

A reader interested in studying these integrals further is invited to read \cite{casestudy, HeslopStewart, Amplituhedronsquared} in order to understand the connection to the scattering amplitude of $N=4$ SYM theory, the cancellation of spurious poles of specific instances of these integrals as well as explicit amplitude calculations.

\subsection{Positroids, Grassmann necklaces, and Le diagrams}\label{sec:positroid background}

We summarize here only the portion of positroid theory that we require in this paper; the interested reader is referred to \cite[sections 6, 16]{Postnikov} for more details. 

Let $\binom{[n]}{k}$ denote the set of all $k$-subsets of $[n]$.  

For our purposes, a {\bf positroid} is\footnote{Note that matroid isomorphism allows arbitrary permutations of the ground set but positroid isomorphism allows only cyclic permutations, so the property of being a positroid is not in general preserved by matroid isomorphism.} a matroid $M$ (with ground set $[n]$ and set of bases $\cB \subseteq \binom{[n]}{k}$) which can be represented by an element of the positive Grassmannian $\Gr(k,n)$. In other words, there exists a full-rank $k\times n$ real matrix whose $k\times k$ minors are all nonnegative and such that the minor comprised of the columns indexed by elements of $J$ 
is positive if and only if $J \in \cB$.

Postnikov showed in \cite{Postnikov} that the positroids of $\Gr(k,n)$ are indexed by many different collections of objects, each with their own advantages and disadvantages. The two most suited to our purposes are Grassmann necklaces \cite[section 16]{Postnikov} and Le diagrams \cite[section 6]{Postnikov}, which we discuss. In order to do so, we first need some preliminary definitions.

For each $j \in [n]$, we can define a total order $<_j$ on the interval $[n]$ by
\[ j <_j j+1 <_j \dots <_j n <_j 1 \dots <_j j-1\;.\]
This in turn induces a total order on $\binom{[n]}{k}$, namely the lexicographic order with respect to $<_j$.  It also induces a separate partial order $\gale{j}$ on $\binom{[n]}{k}$ (the {\bf Gale order} \cite{Gale}), which is defined as follows: for 
\[A = \{a_1 <_j a_2 <_j \dots <_j a_k\} \text{ and } B = \{b_1 <_j b_2 <_j \dots <_j b_k\} \in \binom{[n]}{k},\] we define
\[A \gale{j} B \text{ if and only if } a_r \leq_j b_r \text{ for all }1 \leq r \leq k.\]
For example, in $\binom{[6]}{3}$ we have $\{2,5,6\}\gale{2} \{2,6,1\}$ but $\{2,5,6\}\not\gale{2}\{3,4,6\}$.

\begin{dfn}\label{def:grassmann necklace}
A {\bf Grassmann necklace} of type $(k,n)$ is a sequence $(I_1, \dots, I_n)$ of $n$ sets $I_i \in \binom{[n]}{k}$ such that for each $i \in [n]$ the following statements are true.
\begin{itemize}
\item If $i \in I_i$, then $I_{i+1} = \big(I_i \backslash \{i\}\big) \cup \{j\}$ for some $j \in[n]$.
\item If $i \not\in I_i$, then $I_{i+1} = I_i$.
\end{itemize}
By convention, we set $I_{n+1} = I_1$.
\end{dfn}

By \cite[Theorem 17.1]{Postnikov}, the Grassmann necklaces of type $(k,n)$ are in 1-1 correspondence with the positroid cells in $\Gr(k,n)$. Each term $I_i$ is simply the minimal (with respect to the $<_i$-lex order) basis of the positroid. A characterization of all bases of the positroid in terms of the Grassmann necklace and the Gale order was given by Oh in \cite[Theorem 8]{Oh}: if $(I_1, \dots, I_n)$ is the Grassmann necklace associated to a positroid $M = ([n],\cB)$, then the bases of $M$ are exactly
\ba \cB = \left\{J \in \binom{[n]}{k}\ \Big|\ I_i \gale{i} J \ \ \forall i \in [n]\right\}. \label{basesofmatroids}\ea
Thus the Grassmann necklace is well suited to testing whether a given $k$-set is a basis for $M$, and for generating a list of all bases of $M$.

\begin{dfn}\label{def:le diagram}
A {\bf Le diagram} is a Young diagram in which every square contains either a $+$ or a $0$, subject to the rule that if a square contains a $0$ then either all squares to its left (in the same row) must also contain a $0$, or all squares above it (in the same column) must also contain a $0$, or both.
\end{dfn}

By \cite[Theorem 6.5]{Postnikov}, the set of all Le diagrams that fit within a $k\times(n-k)$ rectangle is in 1-1 correspondence with the positroid cells of $\Gr(k,n)$. Le diagrams are particularly useful for comparing dimensions of positroids, since the dimension of a positroid is equal to the number of $+$ squares in its Le diagram \cite[Theorem 6.5]{Postnikov}.

The rows and columns of a Le diagram are labelled as follows:  given a Le diagram fitting inside a $k\times (n-k)$ box, arrange the numbers $1,2, \dots, n$ along its southeast border, starting from the top-right corner. See Figure \ref{fig:row column numbering} for examples.

\begin{figure}
\[\begin{tikzpicture}
\node at (0,3.5) {}; 
\node at (0,-0.5) {}; 
\draw (1,0) grid (2,3);
\draw (2,1) grid (4,3);
\draw (4,2) grid (5,3);
\draw[dotted] (2,0) -- (5,0) -- (5,2);

\node at (5.15,2.5) {\footnotesize$1$};
\node at (4.5,1.85) {\footnotesize$2$};
\node at (4.15,1.5) {\footnotesize$3$};
\node at (3.5,0.85) {\footnotesize$4$};
\node at (2.5,0.85) {\footnotesize$5$};
\node at (2.15,0.5) {\footnotesize$6$};
\node at (1.5,-0.15) {\footnotesize$7$};

\begin{scope}[shift = {(7,0)}]
\draw (1,0) grid (2,3);
\draw (2,1) grid (4,3);
\draw (4,2) grid (5,3);
\draw[dotted] (2,0) -- (6,0) -- (6,3) -- (5,3);

\node at (5.5,2.85) {\footnotesize$1$};
\node at (5.15,2.5) {\footnotesize$2$};
\node at (4.5,1.85) {\footnotesize$3$};
\node at (4.15,1.5) {\footnotesize$4$};
\node at (3.5,0.85) {\footnotesize$5$};
\node at (2.5,0.85) {\footnotesize$6$};
\node at (2.15,0.5) {\footnotesize$7$};
\node at (1.5,-0.15) {\footnotesize$8$};
\end{scope}
\end{tikzpicture}
\]
\caption{Row and column numbering for a Young diagram with $k = 3$, $n = 7$ (left) and $k = 3$, $n = 8$ (right). The top left box has coordinates $(1,7)$ in the left diagram, and $(2,8)$ in the right diagram.}
\label{fig:row column numbering}
\end{figure}
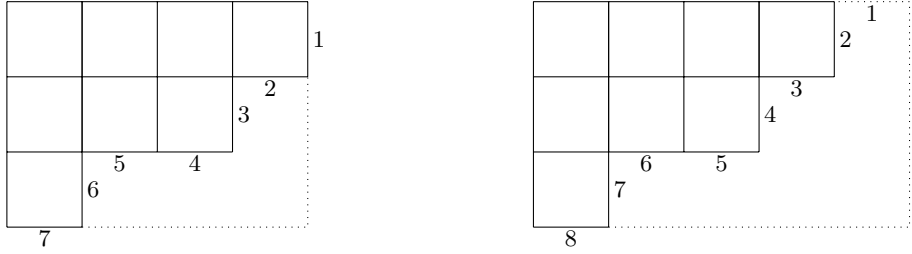

Specifying a $k$-subset $J \subseteq [n]$ therefore uniquely determines the shape of the Le diagram, by taking the elements of $J$ to be the row indices of the diagram.

An algorithm for constructing the Le diagram associated to a Grassmann necklace was given by Agarwala and Fryer in \cite{reversingOh}. Since we will make use of this algorithm in Section~\ref{sec dim}, we summarize the process here.
\begin{algorithm}\label{alg:GN to Le} \ \cite[Algorithm 2]{reversingOh}
Let $(I_1,\dots,I_n)$ be a Grassmann necklace of type $(k,n)$. Within a $k \times(n-k)$ square, draw the Young diagram whose rows are labelled by $I_1$ (as per the convention above).

For each $i$, $2 \leq i \leq n$ perform the following steps.
\begin{itemize}
\item Write \[I_1 \setminus I_i = \{a_1 > a_2 > \dots > a_r\}, \quad I_i \setminus I_1 = \{b_1 < b_2 < \dots < b_r\},\]
where the inequalities denote the $<_1$ order (subscripts suppressed for clarity).
\item For $1 \leq j \leq r$ , place a $+$ in square $(a_j,b_j)$ of the diagram. (We will refer to this $+$ as being {\em in the $a_j \rightarrow b_j$ position}.)
\end{itemize}
After performing the above for $2 \leq i \leq n$, place a $0$ in any remaining unfilled boxes.
\end{algorithm}

An algorithm for constructing the Grassmann necklace of a Le diagram also exists; this was given by Oh in \cite[section 4]{Oh}. A method for using the Le diagram to test whether a given $k$-subset is a basis for the corresponding positroid or not was given by Casteels in \cite{CasteelsPaths}.

\subsection{Wilson loop diagrams as positroids}\label{sec:WLD as positroids}

Let $W = (\cP,[n])$ be an admissible Wilson loop diagram with $k$ propagators, and $C(W)$ the associated matrix defined in equation~\eqref{C(W) dfn}. Let $M(W)$ be the matroid realized by $C(W)$, i.e. the matroid with ground set $[n]$ whose independent sets are exactly those sets $V \subseteq [n]$ such that the columns of $C(W)$ indexed by $V$ are linearly independent.

In \cite{wilsonloop}, Agarwala and Marin-Amat showed that $M(W)$ can also be read directly from the Wilson loop diagram itself.

\begin{thm} \label{thm WLD defines matroid} \cite[Theorem 3.6]{wilsonloop} The independent sets of the matroid $M(W)$ associated to an admissible Wilson loop diagram $W = (\cP,[n])$ and realized by $C(W)$ are exactly those subsets $V \subseteq [n]$ such that there is no subset $U \subseteq V$ satisfying $|\Prop(U)| < |U|$. \end{thm}
In other words, the independent sets of $M(W)$ correspond to sets of vertices in $W$ in which no subset supports fewer propagators than the vertices it contains.

One useful corollary of this, which we will want to keep in mind later, is the following.
\begin{cor}\label{lem basis as perm}
Let $W = (\cP,[n])$ be an admissible Wilson loop diagram and let $M(W)$ be its associated matroid. A subset $J \subseteq [n]$ is an independent set of $M(W)$ if and only if there exists an injective set map $f : J \rightarrow \cP$ with the property that for each $j\in J$ we have $j \in V(f(j))$.
\end{cor}

\begin{proof}
  This follows as a corollary of the previous theorem by induction in one direction and the pigeonhole principle in the other direction.  

  Alternatively we can simply think about the condition of linear independence in the matrix $C(W)$, and so the corollary can also be proved directly by linear algebra and the definition of $C(W)$ as follows.

  Because the nonzero entries of $C(W)$ are independent indeterminates, $J$ is an independent set if and only if there is some choice of $|J|$ nonzero entries of $C(W)$ one in each row associated to an element of $J$ and each in different columns. Each entry in $C(W)$ identifies a propagator by the row of the entry and a vertex by the column of the entry.  The entry is nonzero if and only if the propagator is supported on that vertex.  Consequently, take a choice of $|J|$ nonzero entries of $C(W)$, one in each row associated to an element of $J$ and each in different columns.  Such a choice is equivalent to an assignment of the propagators of $J$ to supporting vertices so that no two are assigned to the same vertex.  Such an assignment of the propagators of $J$ to supporting vertices is exactly a map $f$ as described in the statement, hence proving the result.
\end{proof}

In particular, this corollary says that a subset $J$ of $[n]$ is a basis of $M(W)$ if and only if there is a set bijection between $J$ and $\cP$ with the property that for each $j\in J$ the propagator associated to $j$ under the bijection is supported on vertex $j$.

By relating the behaviour of the propagators in $W$ to the matroidal properties of $M(W)$, Agarwala and Marin-Amat also showed that $M(W)$ is a positroid whenever $W$ is admissible \cite[Corollary 3.39]{wilsonloop}. 

However, to go from the Wilson loop diagram $W$ to the positroid cell, one had to go through the matroid explicitly realized via $C(W)$, then make a list of all its nonzero $k\times k$ minors and extract the Le diagram or Grassmann necklace via the methods described in Section~\ref{sec:positroid background}. In \cite[section 4]{casestudy}, Agarwala and Fryer applied this process in the smallest non-trivial case: admissible Wilson loop diagrams with two propagators on six vertices. Even in this simple case, we see that the mapping from admissible Wilson loop diagrams to positroid cells is neither one-to-one nor onto, and the process described above makes it almost impossible to track the relationship between the original Wilson loop diagram and the resulting positroid.

Our first goal is therefore to find a better method of obtaining the positroid associated to a given Wilson loop diagram; this is the focus of the next section.

\begin{rmk}
While this paper primarily approaches proofs from a combinatorial point of view, readers with a background in matroid theory are encouraged to keep the matroidal viewpoint in mind throughout; we expect that the two approaches will be complementary and will build upon each other. See the first paper in this series \cite[section 3]{generalcombinatoricsI} for some related results concerning the positroids of Wilson loop diagrams obtained by viewing them primarily as matroids.
\end{rmk}

\section{Wilson loop diagrams and their Grassmann necklaces}\label{sec GN algorithm}

In this section we will give an algorithm (Algorithm~\ref{alg:put GN on WLD}) to go directly from a Wilson loop diagram to the Grassmann necklace of its positroid. 

This is a useful result in and of itself, as it greatly simplifies the process of identifying the positroid associated to a given Wilson loop diagram. We also use it to characterize loops and coloops in the positroid associated to a Wilson loop diagram (Corollary \ref{no coloops}).  The characterization of the coloops, in particular, was not at all clear from the definition of Wilson loop diagrams but is an easy corollary of the algorithm.

Furthermore, the algorithm allows us to see how the Le diagram associated to a Wilson loop diagram changes as one adds propagators in certain positions (Theorem \ref{thm dim}). Indeed, the techniques we develop in this section to verify that our algorithm produces the correct Grassmann necklace will also be used repeatedly in later sections to calculate the dimension of the associated positroid cell (Section \ref{sec dim}) and to study the structure of the poles of the associated integrand $\Omega(W)$ (Section \ref{sec poles}).

One of the key insights is that each element of the Grassmann necklace can be viewed as a function from the propagators of the diagram to the vertices in the Grassmann necklace element.  This is captured in Definition~\ref{def I_i as a function}, and used throughout the paper.

The flavour of the calculations is similar to \cite{CGdyck} as in both our paper and theirs\footnote{Thanks to a referee for pointing this paper out.} the bijections are very explicit with a similar iterative flavour using Le diagrams and Grassmann necklaces related to another class of objects, for us Wilson loop diagrams and for them unit interval orders.  For them, as us, the Le diagrams are particularly useful for understanding the dimensions of positroid cells since the dimension becomes the count of the number of plusses.  The ubiquitous Catalan numbers arise for them, as for us in \cite{generalcombinatoricsI}.  The specific subclasses of positroids we work with are, however, different.  This basic methodology of gaining insight from bijections is very common in combinatorics.



Throughout this section, $W = (\cP,[n])$ is an admissible Wilson loop diagram with $k$ propagators.

\subsection{Propagator configurations in admissible Wilson loop diagrams}\label{sec:propagator configs}

Before we can describe the algorithm for extracting the Grassmann necklace of $M(W)$ from the Wilson loop diagram $W$, we require some initial results about the behaviour of propagators in admissible Wilson loop diagrams.

\begin{dfn}\label{props inside p}
Let $W = (\cP, [n])$ be a weakly admissible Wilson loop diagram, with $p \in \cP$ supported on edges $i$ and $j$. Write $(p, <_i)$ to represent the same propagator directed from $i$ to $j$. Then the sets of propagators {\bf inside} and {\bf outside}  $(p,<_i)$ are defined to be
\begin{gather*}\cP_{in}(p,<_i) = \{ (k,l) \in \cP \ |\ i \leq_i k <_i l \leq_i j \}, \\
\cP_{out}(p,<_i) = \cP\setminus \cP_{in}(p,<_i),
\end{gather*}
respectively.
\end{dfn}
Note that we have $p\in \cP_{in}(p, <_i)$ and so $\cP_{in}(p, <_i) = \cP_{out}(p, <_j)\cup\{p\}$.

\begin{eg}
Let $W$ be the Wilson loop diagram from Figure \ref{fig:ex of WLD}, reproduced here for convenience.  \[W\ =\ 
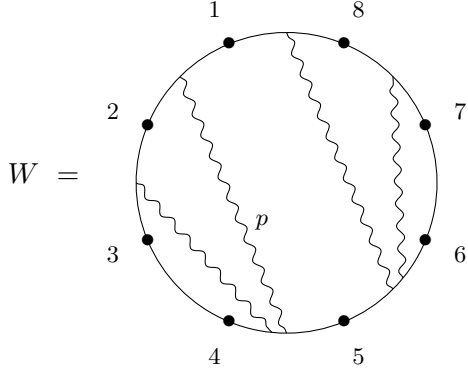
\begin{tikzpicture}[rotate=67.5,baseline=(current bounding box.east)]
	\begin{scope}
	\drawWLD{8}{2}
	\drawnumbers
	\drawlabeledprop{1}{0}{4}{0}{\footnotesize \qquad \ $p$}
	\drawprop{2}{0}{4}{-1}
    \drawprop{5}{0}{8}{0}
    \drawprop{5}{1}{7}{0}
		\end{scope}
	\end{tikzpicture}\]

Let $p$ be the propagator $p = (1, 4)$. Then $(p,<_4)$ indicates that we are viewing $p$ as being directed from 4 to 1, and we have
\begin{gather*}\cP_{in}(p, <_4) = \{(1,4), (5, 8), (5, 7) \};\\
\cP_{out}(p, <_4) = \{(2,4) \}. \end{gather*}
\end{eg}

\begin{dfn}
Let $p = (i,j) \in \cP$ be a propagator in $W$.  Define the {\bf length} of $p$ to be 
\[\ell(p) \  =\ \min\big\{|[i+1,j]|,|[j+1,i]|\big\}.\]
\end{dfn}
In other words, $\ell(p)$ is the size of the smaller of the sets of vertices on either side of the propagator $p$.

\begin{rmk}\label{rem:props of length 2 and 3} 
  The following observations about configurations of propagators of short length in a weakly admissibly Wilson loop diagram $W$ are easily verified.
  \vspace{-0.5em}
 \begin{enumerate}
\item If $p = (i,i+3)$ is a propagator of length 3, then the middle vertex $i+2$ supports at most one propagator.\item If every vertex in $W$ supports at least one propagator, then $W$ admits at least one propagator of length 2.  By the same reasoning, if $p=(i,j)$ is a propagator of $W$ and every vertex in the interval $[i,j]$ supports at least one propagator, then there is at least one propagator of length $2$ in $\mathcal{P}_{in}(p, <_i)$.
\end{enumerate}
\end{rmk}

\begin{figure}
\[
\begin{tikzpicture}[scale = 0.8,baseline=(current bounding box.north)]
  \begin{scope}
  \drawWLDfragment[7]{3}{0.4} 
  \centerarc[linehighlight](0,0)(\zero + 1.48*\step:\zero+6.64*\step:\radius)
  \newnode{1}{}
  \newnode{3.55}{}
  \newnode{4.45}{}
  \newnode{7}{}
    \begin{scope}
    \clipcenterarc(0,0)(\startpoint:\endpoint:\radius)
  \newpropbend{1.3}{6.7}{20}
  \newpropbend{2.9}{5.1}{45}
  \end{scope}
  \node[align=center] at (\zero + 4*\step:\radius*1.43) {\footnotesize \em $\leq 6$ vertices in the highlighted \\ \footnotesize \em region and no other propagators\\ \footnotesize \em in the highlighted region};
\end{scope}

  \begin{scope}[shift = {(9,0)}]
  \drawWLDfragment[7]{3}{0.4} 
  \centerarc[linehighlight](0,0)(\zero + 3.67*\step:\zero+4.47*\step:\radius)
  \newnode{1}{i}
  \newnode[below]{2}{i+1\ \ \ }
  \newnode[below]{3}{i+2}
  \newnode[below]{5}{j+1}
  \newnode[below]{6}{\quad j+2}
  \newnode[right]{7}{j+3}
      \begin{scope}
    \clipcenterarc(0,0)(\startpoint:\endpoint:\radius)
  \newpropbend{1.35}{3.6}{55}
  \newpropbend{4.4}{6.6}{55}
  \end{scope}
  \node[align=center] at (\zero + 4*\step:\radius*1.5) {\footnotesize \em $\leq 2$ vertices in the highlighted \\ \footnotesize \em region, and  no other propagators \\ \em \footnotesize end in the highlighted region};
\end{scope}
\end{tikzpicture}
\]
\caption{The two configurations described in Lemma~\ref{lem sian}; every weakly admissible Wilson loop diagram with no non-supporting vertices admits at least one example of at least one of these configurations.}\label{fig: small configs}
\end{figure}
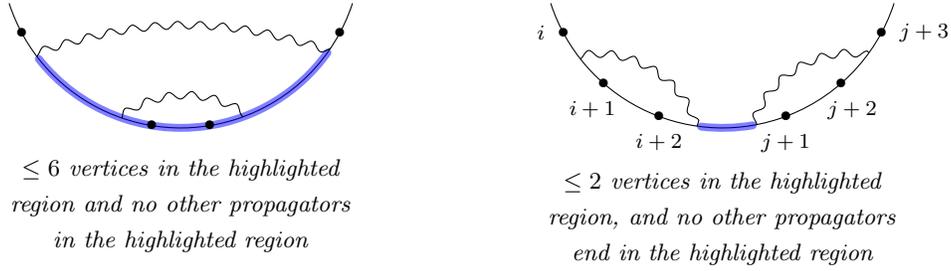

The following lemma establishes certain configurations of propagators that must exist in any admissible diagram with no non-supporting vertices. We make use of this result in several induction proofs below. 

\begin{lem}\label{lem sian}
  Let $W$ be a weakly admissible Wilson loop diagram with at least 5 vertices and in which each vertex supports at least one propagator.  Then at least one of the following two conditions holds.
  \vspace{-0.5em}
  \begin{enumerate}
    \item The diagram $W$ has a propagator of length at most 6 with a propagator of length 2 on one side of it and nothing else on that side.\label{item big and 2}
    \item There exists a pair of propagators of length $2$ with the property that the first propagator is $(i, i+2)$, the second is $(j, j+2)$, no other propagator has an end between vertices $i+2$ and $j+1$, and $j\in\{i+2, i+3, i+4\}$.\label{item pair of 2s}
  \end{enumerate}
See Figure~\ref{fig: small configs} for an illustration of the two configurations.
\end{lem}

\begin{proof}
Suppose first that $W$ has a propagator of length $3$, say $p=(i, i+3)$.  By Remark~\ref{rem:props of length 2 and 3} and the fact that every vertex of $W$ supports at least one propagator, we have that $i+2$ supports exactly one propagator and this propagator must have length 2 by noncrossingness.  This gives us an instance of configuration~\ref{item big and 2} from the statement.

Now suppose $W$ has no propagators of length $3$.
  
We will inductively construct sequences of propagators $(p_r)$ and $(q_r)$, with $\ell(p_r) = 2$ for each $r$ and such that either of the following are true:
\begin{itemize}
\item the propagator $p_r$ forms part of configuration \ref{item big and 2} or \ref{item pair of 2s} from the statement, at which point the induction terminates, 
{\em or} 
\item there is a propagator $q_r$ satisfying $\ell(q_r) \geq 4$, and $\{p_1, \ldots, p_r, q_1, \ldots, q_{r-1}\}$ are all on the same side of $q_r$.  
\end{itemize}

We can choose the orientation of $q_r$ so that the previous $p_i$ and $q_i$ are all on the outside of $q_r$, and restrict our attention to the subdiagram $(\cP_{in}(q_r),[n])$. By the finiteness of $W$ this must eventually terminate in one of the desired configurations. 

Start by choosing a propagator $p_1 = (j_1,j_1+2)$ of length 2 in $W$ (which exists by Remark~\ref{rem:props of length 2 and 3}).  If it is part of one of the configurations we are looking for we are done, 
so suppose otherwise. 

From our existing assumptions, we know the following facts: $p_1$ is not in configuration \ref{item pair of 2s}, there are no propagators of length $3$, and every vertex supports at least one propagator.  Therefore, there must be a propagator of length at least 4 with one end on edge $j_1$ or on edge $j_1-1$ or on edge $j_1-2$.  Let $q_1 = (i_1, k_1)$ be this propagator, oriented such that $p_1 \in \cP_{out}(q_1,<_{i_1})$; see Figure~\ref{fig base case lem sian}. This completes the base case.

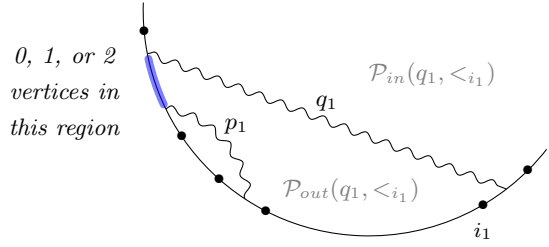
\begin{figure}[h]
\[
\begin{tikzpicture}[baseline=(current bounding box.east),rotate=-20]
  \begin{scope}
  \drawWLDfragment[10]{3}{0.4} 
  \centerarc[linehighlight](0,0)(\zero + 1.5*\step:\zero + 2.35*\step:\radius)
  \newnode{1}{}
  \newnode{3}{}
  \newnode{4}{}
  \newnode{5}{}
  \newnode[below]{9}{i_1}
  \newnode{10}{}
  \newprop[midway,above]{1.4}{9.5}{\footnotesize $q_1$}
  \begin{scope}
    \clipcenterarc(0,0)(\startpoint:\endpoint:\radius)
    \newpropbend{2.3}{4.7}{60}
    \node at (\zero + 3.5*\step:\radius*0.78) {\footnotesize $p_1$};
  \end{scope}
\node[align = center] at (\zero +1.7*\step:\radius*1.4) {\footnotesize \em 0, 1, or 2 \\ \footnotesize \em vertices in \\ \footnotesize \em this region};
\node[gray] at (\zero + 10*\step:\radius*0.4) {\footnotesize$\cP_{in}(q_1,<_{i_1})$};
\node[gray] at (\zero + 6.5*\step:\radius*0.8) {\footnotesize $\cP_{out}(q_1,<_{i_1})$};
\end{scope}
\end{tikzpicture}
\]
\caption{The base case of the induction in Lemma~\ref{lem sian}.}\label{fig base case lem sian} \end{figure}

Now suppose $q_{r} = (i_r, k_r)$ exists by the induction hypothesis and is oriented $(q_r, <_{i_r})$ so that the previous $p_i$ and $q_i$ are on the outside. For the rest of this proof, we assume this orientation and drop the $<_{i_r}$ from the notation.

Let $W_r := (\cP_{in}(q_r),[n])$ be the subdiagram of $W$ consisting of those propagators inside $q_r$ (including $q_r$ itself); in particular $W_r$ contains none of the previous $p_i$ or $q_i$. By the original hypotheses on $W$, every vertex in $[i_r,k_{r+1}]$ (of which there are at least 4, since $\ell(q_r) \geq 4$) must support at least one propagator. Therefore, $W_r$ admits at least one propagator of length 2 (by Remark~\ref{rem:props of length 2 and 3}) and no propagators of length $3$ (by our assumption on $W$).

Let $p_{r+1} = (j_{r+1}, j_{r+1}+2)$ be a propagator of length 2 in $W_r$.  If it forms part of configuration \ref{item big and 2} or \ref{item pair of 2s}
we are done, so assume otherwise. 

Note that we could replace $q_r$ by another propagator $q'_{r}$ in $W_r = (\cP_{in}(q_r), [n])$, as long as $p_{r+1}$ and $q_r$ are on opposite sides of $q_r'$.  Such a new $q'_{r}$ would still satisfy all of the induction hypotheses. (See Figure~\ref{fig smallest q in lem sian}.) Therefore, without loss of generality we may assume that $\cP_{in}(q_r)$ does not contain any propagators (other than $q_r$ itself) which satisfy the induction hypotheses, i.e. we assume that there are no other propagators with $p_{r+1}$ on one side and $q_r$ on the other.

\begin{figure}[h]
\vspace{1cm}
\[
\begin{tikzpicture}[baseline=(current bounding box.east)]
	\begin{scope}
	\drawWLDfragment[15]{4}{0.4} 
	\newnode[left]{1}{i_r}
	\newnode[left]{2}{i_r+1}
	\newnode{5}{}
	\newnode{6}{}
	\newnode{7}{}
	\newnode{8}{}
	\newnode[right]{14}{k_r}
	\newnode[right]{15}{k_r+1}
	\newnode[below]{5}{j_{r+1}}
	\newprop[midway,above]{1.5}{14.5}{\footnotesize $q_r$}
	\draw[smallpropagator,gray] (\zero+2.5*\step:\radius) -- (\zero + 12.5*\step:\radius) node[midway,above,gray] {\footnotesize $q_{r}'$};
	\begin{scope}
		\clipcenterarc(0,0)(\startpoint:\endpoint:\radius)
		\newpropbend{5.2}{7.7}{80}
		\node at (\zero + 6.5*\step:\radius*0.85) {\footnotesize $p_{r+1}$};
	\end{scope}
\end{scope}
\end{tikzpicture}
\]
\caption{We can replace $q_r$ with $q'_r$ as long as $q_r$ and $p_{r+1}$ are on opposite sides of $q'_r$.}
\label{fig smallest q in lem sian}
\end{figure}
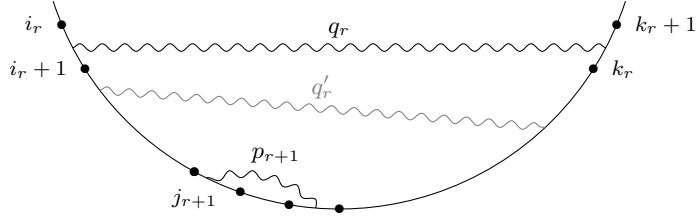

There are two cases to consider.

The first case is that $i_r \leq j_{r+1} \leq i_r+2$ and $k_r-2 \leq j_{r+1}+2 \leq k_{r}$ (so $p_{r+1}$ has one end before $i_{r} + 3$ and the other after $k_{r} - 2$, i.e. there are at most two vertices between an endpoint of $p_{r+1}$ and an endpoint of $q_r$).  Since $p_{r+1}$ has length 2, it must be that $(i_r+3)+1 \geq k_r-2$ and so $q_r$ has length at most 6.  By the minimality assumption on $q_r$, no propagator in $W_r$ has $q_r$ on one side and $p_{r+1}$ on the other side. Therefore, any propagator in $W_r$ other than $q_r$ and $p_{r+1}$ must be of the form $(i,j)$ with $i_r\leq i < j\leq i_r+2$ or $k_r-2 \leq i < j\leq k_r$. Both of these require that $(i,j)$ has length 2, and so we have configuration \ref{item pair of 2s} which we have already assumed does not occur.  Consequently, $W_r$ contains only the two propagators $q_r$ and $p_{r+1}$, yielding configuration \ref{item big and 2} from the lemma statement and again contradicting our assumption. Therefore, this first case cannot occur.

The second case is that there are at least $3$ vertices between one endpoint of $p_{r+1}$ and the nearest endpoint of $q_r$. In other words, at least one of the following four statements are true:
\begin{itemize} 
\item $i_r+3 \leq j_{r+1}\leq k_r-3$ (one endpoint of $p_{r+1}$ has at least 3 vertices between it and the nearest endpoint of $q_r$), or
\item $i_r+3 \leq j_{r+1}+2\leq k_r-3$ (the other endpoint of $p_{r+1}$ has at least 3 vertices between it and the nearest endpoint of $q_r$), or
\item $i_r = j_{r+1}$ ($p_{r+1}$ and $q_r$ both have an end on edge $i_r$), or 
\item $k_r = j_{r+1} +2$  ($p_{r+1}$ and $q_r$ both have an end on edge $k_r$).
\end{itemize}
In other words, at least one end of $p_{r+1}$ is supported entirely by vertices in the interval $[i_r+3,k_r+2]$ (the first two bullet points) or both ends lie in $[i_r, i_r+3]$ (third point) or both ends lie in $[k_r-2, k_r+1]$ (fourth point). These situations all behave similarly, and by symmetry it suffices to only consider the second and third possibilities. 

We have $j_{r+1}+4 \leq k_r-1$ in both of these cases.  In the situation of the second bullet point this follows directly from the inequality, while in the third case it follows from the fact that $\ell(p_{r+1}) = 2$ and $\ell(q_r) \geq 4$. Therefore, the vertex $j_{r+1}+4$ is in the interval $[i_r, k_r+1]$ but does not support either $q_r$ or $p_{r+1}$, and must therefore support some other propagator $t$.  Since $p_{r+1}$ is not part of configuration \ref{item pair of 2s} and $W_r$ has no propagators of length 3, it follows that $\ell(t) \geq 4$.  If $q_r$ and $p_{r+1}$ were on different sides of $t$, this would contradict the minimality of $q_r$.  Therefore, we can orient $t$ so that all previous $q_i$ and $p_i$ belong to $\cP_{out}(t)$, and set $q_{r+1} = t$ to continue the induction.

The overall result follows by induction.
\end{proof}

\begin{rmk}
In the case that all vertices of an admissible Wilson loop diagram $W$ support at least two propagators, Lemma~\ref{lem sian} substantially simplifies.  By Remark~\ref{rem:props of length 2 and 3}, $W$ has no propagators of length $3$.  Configuration \ref{item big and 2} necessarily entails vertices with support $1$ as does configuration \ref{item pair of 2s} unless $j=i+2$.  So in the case that $W$ has all vertices with support at least two, $W$ must contain a pair of propagators of length $2$  with the property that the first propagator is $(i, i+2)$, the second is $(i+2, i+4)$ and no other propagator ends on the edge $i+2$.
\end{rmk}

\subsection{From Wilson loop diagrams to Grassmann necklaces}\label{sec:GN alg}

In this section, we give an algorithm for passing directly from a Wilson loop diagram to its Grassmann necklace. This not only greatly simplifies the previously known process to obtain a positroid cell from a Wilson loop diagram, but also allows us to relate the behaviour of the positroid $M(W)$ directly to the configuration of propagators in $W$. In Section~\ref{sec poles} we use these results to show that the denominator $R(W)$ of the integrand associated to a Wilson loop diagram can be written in terms of the Grassmann necklace of said diagram. 

The fact that Algorithm~\ref{alg:put GN on WLD} does construct the required Grassmann necklace is proved in Theorem~\ref{res:alg gives GN} and Theorem~\ref{res alg gives correct GN}. A worked example of Algorithm~\ref{alg:put GN on WLD} is given in Example~\ref{eg:apply GN alg}.

\begin{algorithm}\label{alg:put GN on WLD}
Let $W = (\cP, [n])$ be an admissible Wilson loop diagram. In order to construct the set $I_i^W$ for $i \in [n]$, perform the following steps.

\begin{enumerate}
\item Fix a vertex $i \in [n]$, the {\bf starting vertex}. Set $j:=i$ and $I_i^W = \emptyset$.
\item While $\cP \neq \emptyset$, perform the following steps.
\begin{enumerate}
\item {\em Step $j$ for starting vertex $i$}: If $\Prop(j) = \emptyset$, do nothing. Else, let $p \in \Prop(j)$ be the clockwise-most propagator supported on $j$. Write $I_i^W = I_i^W\cup \{j\}$, and delete propagator $p$ from the diagram.
\item Increment $j$ by 1 and repeat from (a).
\end{enumerate}
\end{enumerate}
\end{algorithm}

Informally, one starts at $i$ and moves counterclockwise around the vertices of $W$, at each step removing the clockwise-most propagator supported on that vertex (if it exists). The set $I_i^W$ lists the vertices at which a propagator was deleted.

If the algorithm assigns vertex $j$ to propagator $p$ from starting vertex $i$, we say that $p$ \textbf{contributes} $j$ to $I_i^W$. Notationally, we represent this by allowing the $I_i^W$ symbol to represent a function as well as a set, as shown defined below.
\begin{dfn}\label{def I_i as a function}
Let $W = (\cP, [n])$ be an admissible Wilson loop diagram. For each $i \in [n]$, define a function $I_i^W : \cP \longrightarrow [n]$ by
\[I_i^W(p) := \text{the vertex label that $p$ contributes to $I_i$ in Algorithm~\ref{alg:put GN on WLD},}\]
for each $p \in \cP$. 
\end{dfn}

For both $I_i^W$ the set and the function, when the diagram $W$ is clear from context we drop the superscript and simply write $I_i$.

\begin{rmk}
Given $W=(\cP, [n])$ and $i\in [n]$, the order in which the propagators contribute to the algorithm during the construction of $I_i$ defines an order on $\cP$.

Since each propagator is removed by the algorithm after it contributes, each propagator contributes exactly once to $I_i$, making the above ordering well defined.
Note that different $I_*$ impose different orderings on the propagators. Furthermore, we will see below that this ordering respects the ordering of vertices in $I_i$, i.e. if $p_j$ occurs before $p_l$ in the ordering imposed by $I_i$, then we also have $I_i(p_j) <_i I_i(p_l)$. This will follow directly from Corollary \ref{GN alg well defined}.
\end{rmk}

\begin{eg}\label{eg:apply GN alg}Consider the admissible Wilson loop diagram
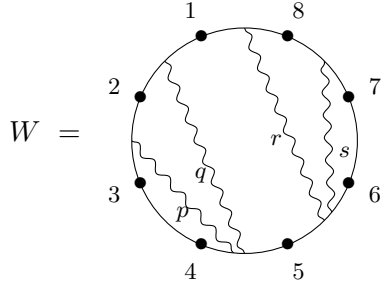
\[W\ =\ \begin{tikzpicture}[rotate=67.5,baseline=(current bounding box.east)]
	\begin{scope}
	\drawWLD{8}{1.5}
	\drawnumbers
	\drawlabeledprop{1}{0}{4}{0}{\footnotesize $q$ \quad}
	\drawlabeledprop{2}{0}{4}{-1}{\footnotesize  $p$}
    \drawlabeledprop{5}{0}{8}{0}{\footnotesize $r$ \ \ }
    \drawlabeledprop{5}{1}{7}{0}{\footnotesize \quad \ $s$}
		\end{scope}
	\end{tikzpicture}\]

To construct $I_1$ we start at vertex 1, so we set $i=1$  and $I_1 = \emptyset$.
\begin{itemize}
\item Begin at $j = 1$. Since $\Prop(1) = \{q,r\}$ and $r$ is the clockwise-most of these propagators, $r$ contributes at this step. We let $I_1 = \{1\}$ and remove propagator $r$ from the diagram. 
\item Set $j = 2$. $\Prop(2) = \{p,q\}$, with $q$ being the clockwise-most of these propagators, so $q$ contributes at this step. We now have $I_1 = \{1,2\}$ and propagator $q$ is removed.
\item Set $j = 3$. $\Prop(3) = \{p\}$, so we have $I_1 = \{1,2,3\}$ and we remove propagator $p$. The only remaining propagator in the diagram is $s$.
\item Set $j = 4$. Since $p$ and $q$ were removed in earlier steps, we now have $\Prop(4) = \emptyset$.
\item Set $j = 5$. $\Prop(5) = \{s\}$, and we have $I_1 = \{1,2,3,5\}$.
\end{itemize}
There are no propagators left in the diagram, so the algorithm terminates and we have ${I_1 = \{1,2,3,5\}}$, which we write more compactly as $1235$. 
$I_1$ viewed as a function is given by
\[I_1: \cP \longrightarrow [n] : \quad I_1(r) = 1,\quad I_1(q) = 2, \quad I_1(p) = 3, \quad I_1(s) = 5.\]
Applying the algorithm for all 8 starting vertices, we obtain the sets
\[1235, 2356, 3456, 4567, 5671, 6712, 7812, 8123.\]
The reader can easily verify that this sequence of $k$-sets satisfies the definition of a Grassmann necklace, and can (with significantly more work) also verify that this Grassmann necklace defines the positroid $M(W)$ associated to the Wilson loop diagram $W$ above. We will prove this in general shortly.
\end{eg}

%


\begin{rmk} \label{clockwise ordering rem}
Suppose that we have $I_i(p) = a$ for some propagator $p$ and vertex $a$. Then $p$ must be the clockwise-most propagator supported on $a$ that has not yet contributed to $I_i$ when the algorithm reaches vertex $a$. In particular, any propagator supported on $a$ that appears later in the ordering given by $I_i$ must be inside $(p, <_{a-1})$.
\end{rmk} 

\begin{rmk}\label{rmk algorithm locally same}
The following observation about the local behaviour of the algorithm will be very useful in many of the arguments below.

Suppose we are going through the algorithm to construct $I_i^W$ for some $i$, and we are currently at step $j$ of the loop. Consider the diagram $W'$ formed by removing from $W$ all propagators that have already contributed to $I_i^W$. (This replicates the behaviour of the algorithm, which removes each propagator once it contributes to $I_i^W$.)  Now suppose for some other diagram $V$ on the same vertex set we are constructing $I_m^V$ for some $m$, and we are at step $j$ of the loop, for the same $j$ as above.  Again obtain $V'$ from $V$ by removing all propagators assigned so far.

With this setup, if we can find a vertex $\ell$ occurring after both $i$ and $m$ such that the \emph{configuration} of propagators ending in the interval $[j-1, \ell+1]$ is identical in both $W'$ and $V'$, then the vertices contributed to $I_i^W$ and $I_m^V$ are the same from step $j$ to step $\ell$ (inclusive). See Figure~\ref{fig:locally identical} for an illustration of this.

Furthermore, if the propagators of $W'$ and $V'$ ending in the interval $[j-1, \ell+1]$ are also identically labelled, then $I_i^W$ and $I_m^V$ are the same as functions when restricted to those propagators.

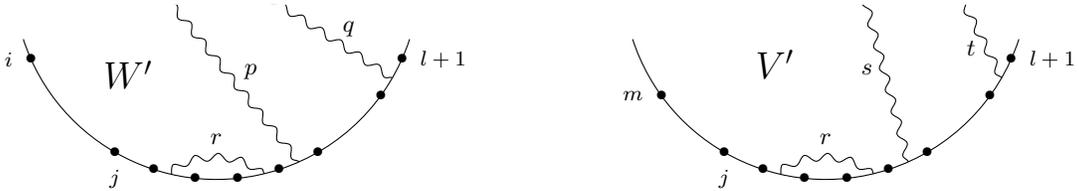
\begin{figure}[h]
\[
\begin{tikzpicture}[scale = 0.9,baseline=(current bounding box.east)]
  \begin{scope}
  \drawWLDfragment[12]{3}{0.4} 
  \node at (0:0) {}; 
  \newnode[left]{1}{i}
  \newnode[below]{4}{j}
  \newnode{5}{}
  \newnode{6}{}
  \newnode{7}{}
  \newnode{8}{}
  \newnode{9}{}
  \newnode{11}{}
  \newnode[right]{12}{l+1}
  \node at (\zero + 3*\step:\radius*0.65) {\Large $W'$};
  \begin{scope}
    \clipcenterarc(0,0)(\startpoint-10:\endpoint+10:\radius)
    \newpropbend{5.3}{7.6}{80}
    \node at (\zero + 6.5*\step:\radius*0.8) {\footnotesize $r$};
    \newprop[midway,left]{8.5}{-4}{}
    \node at (\zero + 8.15*\step:\radius*0.5) {\footnotesize $p$};
    \newprop[midway]{11.5}{-4}{}
    \node at (\zero + 12.2*\step:\radius*0.7) {\footnotesize $q$};
  \end{scope}
\end{scope}

  \begin{scope}[shift = {(9,0)}]
  \drawWLDfragment[12]{3}{0.4} 
  \newnode[left]{2}{m}
  \newnode[below]{4}{j}
  \newnode{5}{}
  \newnode{6}{}
  \newnode{7}{}
  \newnode{8}{}
  \newnode{9}{}
  \newnode{11}{}
  \newnode[right]{12}{l+1}
  \node at (\zero + 4*\step:\radius*0.5) {\Large $V'$};
  \begin{scope}
    \clipcenterarc(0,0)(\startpoint-10:\endpoint+10:\radius)
    \newpropbend{5.3}{7.6}{80}
    \node at (\zero + 6.5*\step:\radius*0.8) {\footnotesize $r$};
    \newprop[midway,left]{8.5}{22}{}
    \node at (\zero + 8.5*\step:\radius*0.5) {\footnotesize $s$};
    \newprop[midway]{11.5}{21}{}
    \node at (\zero + 11.8*\step:\radius*0.8) {\footnotesize $t$};
  \end{scope}
\end{scope}
\end{tikzpicture}
\]
\caption{$W'$ and $V'$ are locally the same between $j$ and $l+1$, so we will have $I_i^W\cap [j,l] = I_m^V\cap[j,l]$ even though $p \neq s$ and $q \neq t$. Note that propagators such as $r$ which have {\em both} ends inside the region of interest must be identical in both diagrams.}\label{fig:locally identical}
\end{figure}

\end{rmk}

With that in mind, we are ready to begin proving properties of the algorithm. We first show that the algorithm for $I_i$ can never reach the final vertex in a propagator's support (in the $<_i$ order) without having already assigned that propagator; in particular, this guarantees that the construction of $I_i$ always terminates in fewer than $n$ steps.

\begin{lem}\label{lem no fourth vertex}
Let $W$ be an admissible Wilson loop diagram containing at least one propagator. For any $i \in [n]$ and for any $p \in \cP$, $I_i(p)$ is not maximal (with respect to $<_i$) amongst the vertices that support $p$.
\end{lem} 
\begin{proof}
Suppose, by way of contradiction that we have $i \in [n]$ and $p \in \cP$ such that $I_i(p)$ is $<_i$-maximal in $V(p)$. There are two possibilities, which are illustrated in Figure~\ref{fig:no fourth vertex}: either $p = (i-1,b)$ for some $b$ (so $i$ and $I_i(p) = i-1$ are the two ends of a single edge), or $p = (a,b)$ with $i\leq_i a <_i b$ and $I_i(p) = b+1$. 

When the first possibility occurs, observe that we must have another propagator $q = (c,b)$ with $c<_i b$ and $I_i(q) = b+1$, because otherwise $I_i$ would assign $p$ to $b+1$. Now $q$ also contributes the $<_i$-maximal vertex in its support to $I_i$ and is an instance of the second possibility. Therefore, it suffices to show that no admissible Wilson loop diagram admits the configuration illustrated in Case 2 of Figure~\ref{fig:no fourth vertex}. 

So let $p = (a,b)$ with $a<_ib$ and $I_i(p) = b+1$, and choose $p$ such that $\big|[a+1,b]\big|$ is minimal amongst propagators of $W$ which contribute their last vertex and are in Case 2.

Since $I_i(p) \neq b$, there must exist a propagator $q$ inside $(p,<_i)$ with $I_i(q) = b$. The propagator $q$ cannot end on the edge $b-1$, as this would contradict the minimality of $p$, so we have $q = (d,b)$ with $a <_i d <_i b$, and $I_i(q) = b$. 

In order for $q$ to remain unassigned until the algorithm reaches vertex $b$, there must be another propagator $r$ with an end on edge $d$ and $I_i(r) = d+1$; the only way this can occur is if $r$ is outside $(q,<_i)$ but inside $(p,<_i)$. But now $r$ contributes its fourth vertex to $I_i$, again contradicting the minimality of $p$.
\end{proof}

\begin{figure}[h]
\[
\begin{tikzpicture}[baseline=(current bounding box.east),rotate = -18]
  \begin{scope}
  \drawWLDfragment[20]{2.2}{1}
  \node at (108:\radius*1.2) {}; 
  \newnode[left]{4}{I_i(p) = i-1}
  \newnode[left]{5}{i}
  \newnode[left]{7}{c}
  \newnode[left]{8}{}
  \newnode[right]{15}{b}
  \newnode[right]{16}{b+1}
  \newprop[midway,above]{4.7}{15.7}{\scriptsize $p$}
  \newprop[midway,below]{7.5}{15.3}{\scriptsize $q$}
  \node at (270+18:\radius*1.5) {\em Case 1};
    \end{scope}
  \end{tikzpicture}
\qquad \qquad 
  \begin{tikzpicture}[baseline=(current bounding box.east),rotate = 36]
  \begin{scope}
  \drawWLDfragment[20]{2.2}{1}
  \newnode[left]{4}{a}
  \newnode[left]{5}{}
  \newnode[left]{7}{}
  \newnode[left]{8}{}
  \newnode[right]{15}{b}
  \newnode[right]{16}{b+1 = I_i(p)}
  \newnode[right]{11}{d}
  \newnode[below]{12}{}
  \newnode[left]{2}{i}
  \newprop[midway,above]{4.5}{15.7}{\scriptsize $p$}
    \node at (270-36:\radius*1.5) {\em Case 2};
  \begin{scope}
  \clip (0,0) circle (\radius);
  \draw[smallpropagator] (\zero+11.5*\step:\radius*1.2) to [bend left = 45] (\zero + 15.5*\step:\radius*1.2); 
  \node at (\zero + 13.5*\step:\radius*0.55) {\scriptsize $q$};
  \draw[smallpropagator] (\zero+7.5*\step:\radius*1.2) to [bend left = 45] (\zero + 11.4*\step:\radius*1.2);
  \node at (\zero + 9.5*\step:\radius*0.55) {\scriptsize $r$};

  \end{scope}
    \end{scope}
  \end{tikzpicture}
\] 
\caption{The two types of configuration that would (in theory) allow the propagator $p$ to contribute its $<_i$-maximal support vertex to $I_i$.}
\label{fig:no fourth vertex}
\end{figure}
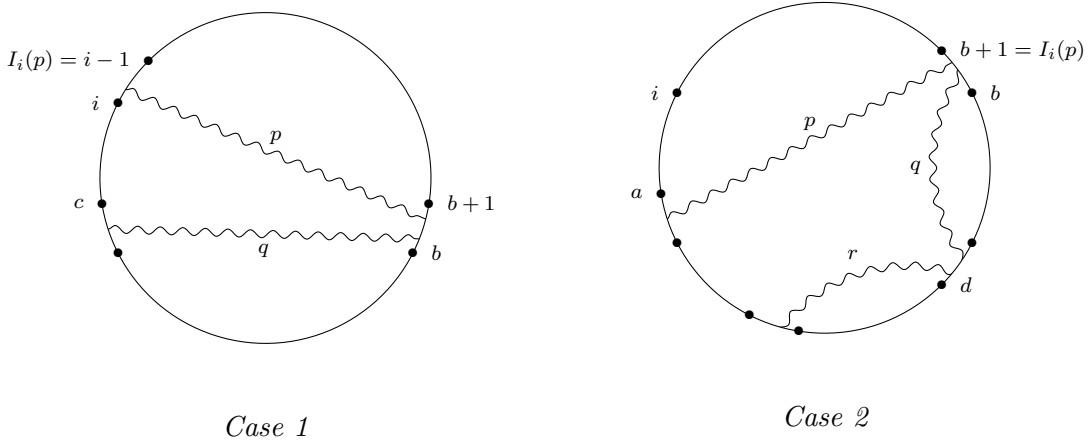

\begin{cor}\label{GN alg well defined}
If $W$ is an admissible Wilson loop diagram with $k$ propagators and $n$ vertices, then Algorithm~\ref{alg:put GN on WLD} assigns exactly $k$ distinct vertices to each $I_i$ in at most $n$ steps.
\end{cor} 
\begin{proof}
Fix $i \in [n]$. It follows from the proof of Lemma~\ref{lem no fourth vertex} that in applying Algorithm~\ref{alg:put GN on WLD} to construct the set $I_i$, the algorithm can never reach the fourth vertex of a propagator's support (ordered with respect to $<_i$) without having already assigned that propagator. Therefore, if the algorithm starts at vertex $i$, all propagators must have been eliminated by the time it gets to vertex $i-1$, ensuring that the algorithm is completed in at most one full circumnavigation of the diagram and that $I_i$ contains exactly $k$ distinct elements.
\end{proof}

We now reverse our perspective and ask: given a propagator $p$ and a vertex $i$ in the support of $p$, for which set of starting vertices does the algorithm assign $p$ to $i$?

Specifically, for each $i$ in the support of $p$ we define the set
\[J_p^{W}(i) = \{m \in [n] \ | \ I^{W}_m(p) = i \},\]
i.e. the set of starting vertices $m$ such that $I_m^{W}$ assigns $p$ to $i$. 

The following lemma establishes that these sets behave in a simple and predictable manner, a fact which we will repeatedly use in subsequent proofs.

\begingroup
\allowdisplaybreaks

\begin{lem} \label{vertex cyclic int lem}
  
Let $W = (\cP,[n])$ be an admissible Wilson loop diagram.  For every $p = (i,j) \in \cP$, the sets $J_p^{W}(i)$, $J_p^{W}(i+1)$, $J_p^{W}(j)$, and $J_p^{W}(j+1)$ are each non-empty cyclic intervals which partition $[n]$ and occur in the given cyclic order.
\end{lem}

\begin{rmk}\label{rmk cyclic}
  Before the proof, let us observe that one of the main ways that the above lemma will be used is in observations like the following.  Suppose $p=(i,j)$ is a propagator and $I_{j+1}(p) = j+1$.  Consider $I_{j+2}(p)$.  Since $j+1$ is the last vertex in the $j+2$ order, $p$ cannot contribute $j+1$ to $I_{j+2}$, so the fact that the intervals in the statement of the lemma are non-empty and occur in their cyclic order implies that $I_{j+2}(p)=i$.  More generally, we know that whenever a propagator stops contributing a particular vertex, its next contribution must be its cyclically next vertex.  This simple observation is surprisingly useful, to the point that the details which appear in the proof of Lemma~\ref{vertex cyclic int lem} are never needed again.  The mere fact of non-empty intervals in the correct order suffices.
\end{rmk}
  
Now we proceed to the proof of Lemma~\ref{vertex cyclic int lem}.

\begin{proof}
We will prove the result by induction on the number of propagators.  If $W$ has one propagator then the result is immediate.  Now suppose $W$ has more than one propagator.  Since non-supporting vertices have no effect on Algorithm~\ref{alg:put GN on WLD}, it suffices to prove the result for $W$  (possibly weakly admissible) with no non-supporting vertices.  Then by Lemma~\ref{lem sian}, $W$ has at least one of the four situations illustrated in Figure~\ref{fig 3 cases}.

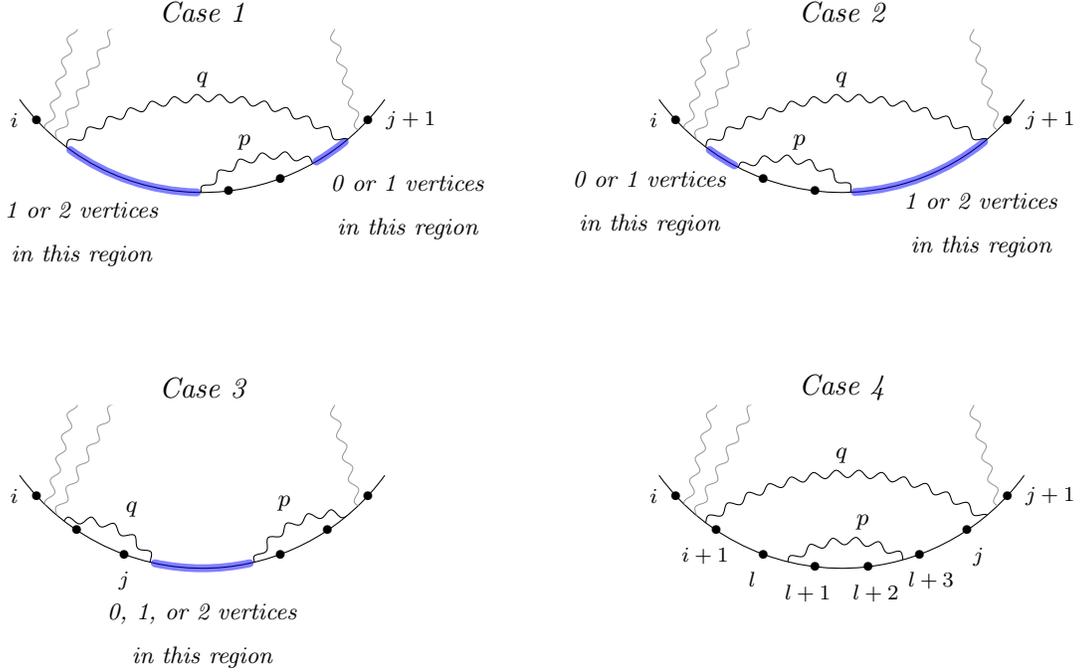
\begin{figure}
\[\begin{tikzpicture}[baseline=(current bounding box.east)]
	\begin{scope}
	\drawWLDfragment[8]{3}{0.3} 
	\centerarc[linehighlight](0,0)(\zero + 1.85*\step:\zero + 4.4*\step:\radius)
	\centerarc[linehighlight](0,0)(\zero + 6.75*\step:\zero + 7.4*\step:\radius)
	\node[align = center] at (\zero + 2.7*\step:\radius*1.3) {\em \footnotesize 1 or 2 vertices \\[3pt] \em \footnotesize in this region};
	\node[align = center] at (\zero + 7.5*\step:\radius*1.4) {\em \footnotesize 0 or 1 vertices \\[3pt]\em \footnotesize in this region};
	\begin{scope} 
		\clipcenterarc(0,0)(\startpoint-20:\endpoint+20:\radius)
		\draw[smallpropagator] (\zero+1.65*\step:\radius*1.1) to [bend left = 40] (\zero + 7.4*\step:\radius*1.1); 
		\draw[smallpropagator] (\zero+4.3*\step:\radius*1.1) to [bend left = 55] (\zero + 6.8*\step:\radius*1.1); 
		\draw[smallpropagator,light-gray] (\zero+1.5*\step:\radius*1.1) -- (180:\radius*0.5); 
		\draw[smallpropagator,light-gray] (\zero+1.6*\step:\radius*1.1) -- (180:\radius*0.25);
		\draw[smallpropagator,light-gray] (\zero+7.55*\step:\radius*1.1) -- (0:\radius*0.5); 
		\node at (\zero + 4.5*\step:\radius*0.5){\footnotesize $q$};
		\node at (\zero + 5.5*\step:\radius*0.8){\footnotesize $p$};
	\end{scope}
	\newnode[left]{1}{i}
	\newnode[right]{8}{j+1}
	\newnode[below]{5}{}
	\newnode[below]{6}{}
\node at (270:\radius*0.2) {\em Case 1};
\end{scope}

\begin{scope}[shift = {(8.5,0)}]
	\drawWLDfragment[8]{3}{0.3} 
	\centerarc[linehighlight](0,0)(\zero + 1.85*\step:\zero + 2.4*\step:\radius)
	\centerarc[linehighlight](0,0)(\zero + 4.75*\step:\zero + 7.4*\step:\radius)
	\node[align = center] at (\zero + 6.6*\step:\radius*1.3) {\em \footnotesize 1 or 2 vertices \\[3pt]\em \footnotesize in this region};
	\node[align = center] at (\zero + 1.6*\step:\radius*1.35) {\em \footnotesize 0 or 1 vertices \\[3pt]\em \footnotesize in this region};
	\begin{scope}
		\clipcenterarc(0,0)(\startpoint-20:\endpoint+20:\radius)
		\draw[smallpropagator] (\zero+1.65*\step:\radius*1.1) to [bend left = 40] (\zero + 7.4*\step:\radius*1.1); 
		\draw[smallpropagator] (\zero+2.3*\step:\radius*1.1) to [bend left = 55] (\zero + 4.8*\step:\radius*1.1); 
		\draw[smallpropagator,light-gray] (\zero+1.5*\step:\radius*1.1) -- (180:\radius*0.5); 
		\draw[smallpropagator,light-gray] (\zero+1.6*\step:\radius*1.1) -- (180:\radius*0.25);
		\draw[smallpropagator,light-gray] (\zero+7.55*\step:\radius*1.1) -- (0:\radius*0.5); 
		\node at (\zero + 4.5*\step:\radius*0.5){\footnotesize $q$};
		\node at (\zero + 3.5*\step:\radius*0.8){\footnotesize $p$};
	\end{scope}
	\newnode[left]{1}{i}
	\newnode[right]{8}{j+1}
	\newnode[below]{3}{}
	\newnode[below]{4}{}
	\node at (270:\radius*0.2) {\em Case 2};
\end{scope}

\begin{scope}[shift = {(0,-5)}]
	\drawWLDfragment[8]{3}{0.3} 
	\centerarc[linehighlight](0,0)(\zero + 3.6*\step:\zero + 5.4*\step:\radius)
	\node[align = center] at (\zero + 4.5*\step:\radius*1.3) {\em \footnotesize 0, 1, or 2 vertices \\[3pt]\em \footnotesize in this region};
	\begin{scope}
		\clipcenterarc(0,0)(\startpoint-20:\endpoint+20:\radius)
		\draw[smallpropagator] (\zero+1.5*\step:\radius*1.1) to [bend left = 55] (\zero + 3.7*\step:\radius*1.1); 
		\draw[smallpropagator] (\zero+5.3*\step:\radius*1.1) to [bend left = 55] (\zero + 7.6*\step:\radius*1.1); 
		\draw[smallpropagator,light-gray] (\zero+1.5*\step:\radius*1.1) -- (180:\radius*0.5); 
		\draw[smallpropagator,light-gray] (\zero+1.6*\step:\radius*1.1) -- (180:\radius*0.25);
		\draw[smallpropagator,light-gray] (\zero+7.55*\step:\radius*1.1) -- (0:\radius*0.5); 
		\node at (\zero + 2.8*\step:\radius*0.8){\footnotesize $q$};
		\node at (\zero + 6.5*\step:\radius*0.8){\footnotesize $p$};
	\end{scope}
	\newnode[left]{1}{i}
	\newnode[below]{2}{}
	\newnode[below]{3}{j}
	\newnode[below]{6}{}
	\newnode[below]{7}{}
	\newnode[right]{8}{}
	\node at (270:\radius*0.2) {\em Case 3};
\end{scope}

\begin{scope}[shift = {(8.5,-5)}]
	\drawWLDfragment[8]{3}{0.3} 
	\begin{scope}
		\clipcenterarc(0,0)(\startpoint-20:\endpoint+20:\radius)
		\draw[smallpropagator] (\zero+1.65*\step:\radius*1.1) to [bend left = 40] (\zero + 7.4*\step:\radius*1.1); 
		\draw[smallpropagator] (\zero+3.3*\step:\radius*1.1) to [bend left = 55] (\zero + 5.8*\step:\radius*1.1); 
		\draw[smallpropagator,light-gray] (\zero+1.5*\step:\radius*1.1) -- (180:\radius*0.5); 
		\draw[smallpropagator,light-gray] (\zero+1.6*\step:\radius*1.1) -- (180:\radius*0.25);
		\draw[smallpropagator,light-gray] (\zero+7.55*\step:\radius*1.1) -- (0:\radius*0.5); 
		\node at (\zero + 4.5*\step:\radius*0.5){\footnotesize $q$};
		\node at (\zero + 5*\step:\radius*0.8){\footnotesize $p$};
	\end{scope}
	\newnode[left]{1}{i}
	\newnode[below]{2}{i+1\quad }
	\newnode[below]{3}{l\quad}
	\newnode[below]{4}{l+1\ \ }
	\newnode[below]{5}{\ \ l+2}
	\newnode[below]{6}{\quad l+3}
	\newnode[below]{7}{\quad j}
	\newnode[right]{8}{j+1}
	\node at (270:\radius*0.2) {\em Case 4};
\end{scope}
\end{tikzpicture}\]
\caption{Four cases for admissible Wilson loop diagrams with no non-supporting vertices. The gray half-propagators illustrate where propagators may occur, but are not required to exist.}\label{fig 3 cases}
  \end{figure}

In each of the four cases, when we remove the propagator labelled $p$ we obtain a diagram which satisfies the statement of the theorem by the induction hypothesis, and contains a propagator ${q = (i,j)}$ with no other propagators inside it (with respect to the $<_i$ orientation). Note that this region may or may not contain non-supporting vertices, which we will call $l, l+1, \ldots$ as necessary. Let $V$ be the diagram obtained by removing the propagator $p$ from $W$ (Figure \ref{fig V diagram}). The diagram $V$ is guaranteed to be admissible, since it was formed by removing a propagator from a diagram that was at least weakly admissible. 

Consider Case 4 of Figure~\ref{fig 3 cases} first, as it is the easiest.  In this case, the vertices ${\{l,l+1,l+2,l+3\}}$ which support $p$ in $W$ are non-supporting in $V$, so for every propagator $r$ in $V$ (including $q$) we have $J_r^{V} = J_r^{W}$; these are non-empty cyclic intervals in the correct order by the induction hypothesis.  Additionally, it is clear from Figure~\ref{fig 3 cases} that $J_p^{W}(l+a) = \{l+a\}$ for $a\in\{1,2,3\}$ and $J_p^{W}(l) = [n] \backslash \{l+1, l+2, l+3\}$. The result therefore holds in this case.

Now we proceed to consider the first three cases of Figure~\ref{fig 3 cases}.
We first describe $J_q^{V}(*)$, which can be handled identically for all three cases.

In $V$ there are no propagators inside $q$ by construction, so we see from Figure \ref{fig V diagram} that
\[l, l+1, \ldots ,j \in J^{V}_q(j) \text{ (if $l$ exists) and }  j+1 \in J^{V}_q(j+1).\]
Note that $j+2 \not\in J^{V}_q(j+1)$ by Lemma~\ref{lem no fourth vertex}, so  $J^{V}_q(j+1) = \{j+1\}$ and by the induction hypothesis we must have $j+2 \in J^{V}_q(i)$. Thus there exist vertices $d,e \in [j+2,i+1]$ with $d <e$, such that
\ba J^{V}_q(i) = [j+2,d-1], \quad J^{V}_q(i+1) = [d,e-1], \quad J^{V}_q(j) = [e,j], \quad J^{V}_q(j+1) = \{j+1\}, \label{deeqs}\ea
and all intervals are non-empty.

We now consider what happens as we move from $V$ to each of the three remaining cases for $W$. We need to consider both $J^{W}_p$ and $J^{W}_r$ for $r \neq p$, since the addition of $p$ can have a knock-on effect on later steps in the algorithm.

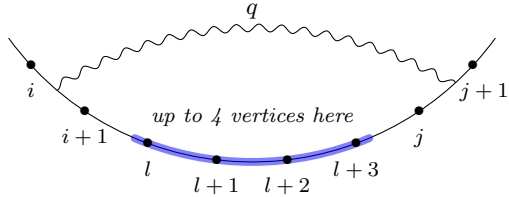
\begin{figure}
\[
\begin{tikzpicture}[baseline=(current bounding box.east)]
	\begin{scope}
	\drawWLDfragment[8]{4}{0.3}
	\centerarc[linehighlight](0,0)(\zero + 2.8*\step:\zero + 6.2*\step:\radius)
	\newnode[below]{1}{i}
	\newnode[below]{2}{i+1}
	\newnode[below]{3}{l}
	\newnode[below]{4}{l+1}
	\newnode[below]{5}{l+2}
	\newnode[below]{6}{l+3}
	\newnode[below]{7}{j}
	\newnode[below]{8}{\quad j+1}
	\node at (\zero + 4.5*\step:\radius*0.85) {\em \scriptsize up to 4 vertices here};
	\begin{scope}
	\clipcenterarc(0,0)(\startpoint-10:\endpoint+10:\radius)
		\draw[smallpropagator] (\zero+1.3*\step:\radius*1.2) to [bend left = 40] (\zero + 7.7*\step:\radius*1.2); 
	\node at (\zero + 4.5*\step:\radius*0.5){\footnotesize $q$};
\end{scope}
\end{scope}
\end{tikzpicture}
\]
\caption{Diagram $V$ is $W$ with $p$ removed; there are no propagators inside $q = (i,j)$, though there may be up to 4 non-supporting vertices labelled $l, l+1, \ldots l+3$.}
\label{fig V diagram}
\end{figure}


\textbf{Cases 1 and 2}: These cases behave very similarly except when $j$ or $j+1$ are not in the support of $p$, which can occur in Case 2 only; see below. Consequently we handle the majority of the proof for these two cases simultaneously. Write $p=(m,m+2)$ where $m\in \{i, i+1, l\}$. Let $1\leq a\leq 3$ be the number of non-supporting vertices inside $q$ in $V$ (note that there is at least 1); so these vertices are $l, \ldots, l+a-1$.  Note that these vertices are non-supporting in $V$ but are all in the support of $p$ in $W$. 

We first calculate $I^{W}_w$ for a starting vertex ${w \in [n]\backslash \{l, l+1, \ldots,j,j+1\}}$. Note that $p$ has no effect on other propagators for starting vertices in this range (since $p$ will be the final propagator encountered by the algorithm) so if $r \neq p$, then $I_w^W(r) = I_w^V(r)$. Meanwhile, the value of $I^{W}_w(p)$ depends on which step in the algorithm the propagator $q$ contributes a vertex in the diagram $V$, i.e. on the value of $I^{V}_w(q)$. Thus, if $w\in J_q^{V}(i)$ then 
    \[
    I_w^{W}(r) =  \begin{cases}
        \max\{i+1, m\} & \text{if } r=p \\
        I_{w}^{V}(r) & \text{if } r\neq p\;,
      \end{cases} 
    \]
    while if $w\in J_q^{V}(i+1)$ or $w\in J_q^{V}(j)$ then
    \[
    I_w^{W}(r) =  \begin{cases}
        l & \text{if } r=p \\
        I_{w}^{V}(r) & \text{if } r\neq p \;.
      \end{cases} 
      \]
We also need to understand $I^{W}_w$ for $w \in \{l,l+1,\ldots,j,j+1\}$. For the majority 
of these vertices, we use the following observation: if $p$ is the first propagator to be assigned a value by $I_w^{W}$, then the remainder of $I_w^{W}$ proceeds identically to the assignments of $I^{V}_{w+1}$ (recall Remark~\ref{rmk algorithm locally same}).  Thus we have for any $0\leq b <a$
    \[
    I_{l+b}^{W}(r) = \begin{cases}
      l+b & \text{if } r=p\\
      I_{l+b+1}^{V}(r) & \text{if } r\neq p\;.
    \end{cases}
    \]
Similarly, if $j$ is in the support of $p$, then we have
     \[
       I_j^{W}(r) = \begin{cases}
         j & \text{if $r=p$ and $j$ is in the support of $p$}\\
         I_{j+1}^{V}(r) & \text{if $r\neq p$ and $j$ is in the support of $p$} \;.
       \end{cases} 
       \]
If $j$ is not in the support of $p$, then we must be in Case 2 of Figure~\ref{fig 3 cases} with two vertices in the right hand region.  In this case, if we start the algorithm at $j$ we need to know whether there will be any unassigned propagators other than $p$ when we reach vertex $i$, so as to know what $p$ contributes.

Consider the Wilson loop diagram $X$ formed from $V$ by replacing the propagator $q = (i,j)$ with $q = (i, j-1)$. The diagram  $X$ is still admissible since we have not decreased the support of any set of propagators. 
It follows from the position of $q$ in $X$ that $I_j^X(q) = j$, and from Lemma~\ref{lem no fourth vertex} that $I_{j+1}^X(q) \neq j$. By the induction hypothesis, we must therefore have $I_{j+1}^X(q) = i$.  In particular, this means that if we start at $j+1$ and assign propagators to vertices according to the algorithm, when we reach vertex $i$ in $X$ all propagators other than $q$ must have been assigned.

Therefore, if we start at $j$ in $W$, we first assign $q$ to $j$ then proceed to assign as in $X$ starting at $j+1$, and hence when we get to $i$ the only remaining unassigned propagator is $p$. We obtain
       \[
       I_j^{W}(r) = \begin{cases}

       m & \text{if $r =p$ and $j$ is not in the support of $p$} \\
       I_{j}^{V}(r) & \text{if $r\neq p$ and $j$ is not in the support of $p$}\; .\end{cases} 
       \]  
This completes the analysis of $I_j^W$. Finally, we consider what happens when we start the algorithm at vertex $j+1$. If $j+1$ is in the support of $p$ then we can argue as above to get
\[
       I_{j+1}^{W}(r)  = \begin{cases}
         j+1 & \text{if $r=p$ and $j+1$ is in the support of $p$}\\
         I_{j+2}^{V}(r) & \text{if $r\neq p$ and $j+1$ is in the support of $p$}\;.
       \end{cases}
       \]
Now suppose $j+1$ is not in the support of $p$. If we start at $j+1$ we need to know whether there are any unassigned propagators supported on edge $i$ when we reach vertex $i$. We already know that $J_q^{V}(j+1) = \{j+1\}$; in particular this means that $q$ contributes $i$ in $I_{j+2}^{V}$, by the induction hypothesis applied to $V$. However the construction of $I^{V}_{j+1}$ first assigns $q$ to $j+1$ and proceeds identically to $I^{V}_{j+2}$.  In particular if $i$ was assigned in $I^{V}_{j+1}$, then it would not be available to assign to $q$ in $I^{V}_{j+2}$ as all other propagators supported at $i$ in $V$ come before $q$. 

Therefore, $p$ is the only potentially unassigned propagator on edge $i$ when we reach vertex $i$ in the algorithm for $I_{j+1}^W$, and
    \[
    I_{j+1}^{W}(r)  = \begin{cases}
      m & \text{if  $r=p$ and $j+1$ is not in the support of $p$}\\
      I_{j+1}^{V}(r) & \text{if  $r\neq p$ and $j+1$ is not in the support of $p$} \;.
    \end{cases}
    \]
We can now describe the intervals $J^{W}_r(*)$ for Cases 1 and 2 of Figure~\ref{fig 3 cases}. For $r \neq p$ the intervals are clearly still cyclic and appear in the correct order, and we can assemble the intervals for the $J_p^{W}(*)$ as follows. Recall that $p = (m,m+2)$, $1 \leq a \leq 3$ is the number of non-supporting vertices inside $q$ in $V$, and these non-supporting vertices are denoted $\{l,l+1, \dots\}$. 
    \begin{itemize}
      \item 
    If $m=l$ then either $a=2$ (so $l+1=m+1$, $j=m+2$, and $j+1=m+3$) or $a=3$ (so $l+1=m+1$, $l+2=m+2$, $j=m+3$, and $j+1$ is not in the support of $p$).  In both cases
    \begin{gather*}
    J^{W}_p(m) = [m+4, m], \qquad  J^{W}_p(m+1) = \{m+1\}, \\ J^{W}_p(m+2) = \{m+2\}, \qquad  J^{W}_p(m+3) = \{m+3\},
    \end{gather*}
    which are non-empty and otherwise as required.
  \item
        If $m=i+1$ then checking each of the three different possibilities for $a$ we likewise get
    \begin{gather*}
    J^{W}_p(m) = [m+4, d-1], \qquad  J^{W}_p(m+1) = [d, m+1], \\  J^{W}_p(m+2) = \{m+2\}, \qquad  J^{W}_p(m+3) = \{m+3\},
    \end{gather*}
    which are non-empty and otherwise as required. Recall that $d$ was defined in equation \eqref{deeqs}.
  \item If $m=i$ then $a=1$ or $a=2$, in the former case $l=m+2$, $j=m+3$ and $j+1$ is not in the support of $p$ so
    \begin{gather*}
    J^{W}_p(m) = \{m+4\}, \qquad  J^{W}_p(m+1) = [m+5, d-1], \\  J^{W}_p(m+2) = [d, m+2], \qquad  J^{W}_p(m+3) = \{m+3\},
    \end{gather*}
    while in the latter $l=m+2$, $l+1=m+3$, and $j$ and $j+1$ are not in the support of $p$ so
    \begin{gather*}
    J^{W}_p(m) = [m+4, j+1], \qquad  J^{W}_p(m+1) = [j+2, d-1], \\  J^{W}_p(m+2) = [d, m+2], \qquad  J^{W}_p(m+3) = \{m+3\},
    \end{gather*}
    which are again as required.
    \end{itemize}

\textbf{Case 3}: In this case there are no non-supporting vertices $l, l+1, \ldots$ inside $q$.  Again write ${p=(m, m+2)}$ where $m\in \{j, j+1, j+2\}$. We proceed as in the previous cases, by computing $I^{W}_w$ for vertices $w$ in roughly increasing order of difficulty.

For $w \in [n]\backslash\{j+1,m,m+1,m+2,m+3\}$, if $w\in J_q^{V}(i)$ or $w\in J_q^{V}(i+1)$ then
    \[
    I_w^{W}(r) =  \begin{cases}
        m & \text{if } r=p \\
        I_{w}^{V}(r) & \text{if } r\neq p \;,
      \end{cases} 
    \]
    while if $w\in J_q^{V}(j)$ then
    \[
    I_w^{W}(r) =  \begin{cases}
        \max\{m, j+1\} & \text{if } r=p \\
        I_{w}^{V}(r) & \text{if } r\neq p
      \end{cases} \;.
    \]

Finally, for $j+1$ and the vertices in the support of $p$, we have
\begin{align*}
  I_{j+1}^{W}(r) &= \begin{cases}
    j+1 & \text{if } r=q\\
    j+2 & \text{if } r=p\\
    I_{j+3}^{V}(r) & \text{if } r\neq p,q
  \end{cases}\\
  I_m^{W}(r) &= \begin{cases}
    m & \text{if $r=p$ and $q$ not supported on $m$}\\
    m+1 & \text{if $r=p$ and $q$ supported on $m$}\\
    I_{m}^{V}(r) & \text{if } r\neq p
  \end{cases}\\
  I_{m+1}^{W}(r) & = \begin{cases}
    m+1 & \text{if $r=p$ and $q$ not supported on $m+1$} \\
    m+2 & \text{if $r=p$ and $q$ supported on $m+1$} \\
    I_{m+1}^{V}(r) & \text{if } r\neq p,q
  \end{cases}\\
  I_{m+2}^{W}(r) & = \begin{cases}
    m+2 & \text{if } r=p \\
    I_{m+3}^{V}(r) & \text{if } r\neq p
  \end{cases}\\
  I_{m+3}^{W}(r) & = \begin{cases}
    m+3 & \text{if } r=p \\
    I_{m+4}^{V}(r) & \text{if } r\neq p \; .
  \end{cases} 
\end{align*}
Note that $I_{m+2}^{V}(r) = I_{m+3}^{V}(r)$ for all propagators $r$ in $V$, and that if $j+1\not\in\{m, m+1, m+2, m+3\}$ then $j+2$ and $j+3$ are non-supporting vertices in $V$, so in that case $I_{j+2}^{V}(r) = I_{j+3}^{V}(r) = I_{j+4}^{V}(r)$ for $r$ in $V$. 

Therefore, once again we can see that the $J_r^{W}(*)$ are cyclic for all $r\neq p$ in $W$.  Assembling the intervals for $p$ we have the following statements.
\begin{itemize}
\item If $m=j$ then
\begin{gather*}J_p^{W}(m) = [m+4,e-1], \qquad  J_p^{W}(m+1) = [e,m], \\  J_p^{W}(m+2) = [m+1,m+2], \qquad  J_p^{W}(m+3) = \{m+3\}\;.\end{gather*}
Recall that $e$ was defined in equation \eqref{deeqs}.
\item If $m=j+1$ then
\begin{gather*}J_p^{W}(m) = [m+4,m-1], \qquad J_p^{W}(m+1) = [m,m+1], \\  J_p^{W}(m+2) = \{m+2\}, \qquad  J_p^{W}(m+3) = \{m+3\}\;.\end{gather*}
\item If $m=j+2$ then
\begin{gather*}J_p^{W}(m) = [m+4,m], \qquad  J_p^{W}(m+1) = \{m+1\}, \\  J_p^{W}(m+2) = \{m+2\}, \qquad  J_p^{W}(m+3) = \{m+3\}\;.\end{gather*}
\end{itemize}
The result now follows by induction.
\end{proof}
\endgroup 

We are now ready to prove that Algorithm~\ref{alg:put GN on WLD} gives the Grassmann necklace of the positroid associated to $W$. We do this in two parts: Theorem~\ref{res:alg gives GN} verifies that the sets $(I_1,\dots,I_n)$ obtained from the algorithm do form a Grassmann necklace, and Theorem~\ref{res alg gives correct GN} proves that this Grassmann necklace defines the positroid $M(W)$.

\begin{thm}\label{res:alg gives GN}
The sequence of $k$-subsets $(I_1,\dots,I_n)$ obtained by applying Algorithm~\ref{alg:put GN on WLD} to all vertices of an admissible diagram $W$ is a Grassmann necklace.
\end{thm}
\begin{proof}
For each $i \in [n]$, let $I_i$ be the set of vertices assigned to the propagators of $W$ by Algorithm~\ref{alg:put GN on WLD} with starting vertex $i$. By Lemma~\ref{GN alg well defined}, we know that $|I_i| = k$ for each $i \in [n]$. We prove that $I_{i+1} \supseteq I_i \backslash \{i\}$ for all $i \in [n]$; this is equivalent to the definition of Grassmann necklace given in Definition~\ref{def:grassmann necklace}.

Fix $n$, and suppose, by way of contradiction that there exists an admissible diagram for which there exists an $i$ with $m\in I_i\setminus \{i\}$ and $m \not\in I_{i+1}$.  Let $W$ be such a counterexample with the minimal number of propagators.

If $i \not\in I_i$, then there are no propagators supported on $i$ at all.  In this case it is clear that applying Algorithm~\ref{alg:put GN on WLD} at vertex $i$ and vertex $i+1$ produces exactly the same result, i.e. $I_{i+1} = I_i$, contradicting that $W$ is a counterexample.

Now suppose that $i \in I_i$.  Let $p$ be the propagator which contributes $i$ to $I_i$, that is $I^W_i(p) = i$; thus one end of $p$ must lie on either edge $i-1$ or edge $i$.  In both cases let $b$ denote the edge supporting the other end of $p$, i.e. $p = (i, b)$ or $p = (i-1, b)$.

\textbf{Case I}:  Suppose $p$ has one end on edge $i-1$. By Remark~\ref{rmk algorithm locally same} the construction of $I_i^W$ and $I_{i+1}^W$ is locally the same (indeed, locally {\em identical} in this case) from vertex $i+1$ to $b-1$, so we have $I_{i+1}^W \cap [i+1,b-1] = I_{i}^W \cap [i+1,b-1]$. Note that if $b = i+1$ then this interval is empty and the equality remains true.


By Lemma~\ref{vertex cyclic int lem} it must happen that we have $I_{i+1}^W(p) = b$, as otherwise $b$ would never be contributed by $p$.  Consequently, when the construction of $I_{i+1}^W$ reaches vertex $b$ there cannot be any remaining unassigned propagators that are supported on $b$ and clockwise of $p$ (Remark~\ref{clockwise ordering rem}). It follows that the same must be true when $I_i^W$ reaches vertex $b$, since $I_i^W$ and $I_{i+1}^W$ behave identically on the interval $[i+1,b-1]$ as noted in the previous paragraph. 

The previous two paragraphs also imply that $m \geq_i b+1$, since $I_{i+1}^W \cap [i+1,b-1] = I_{i}^W \cap [i+1,b-1]$ and $b \in I_{i+1}^W$.


Now let $W' = (\cP_{out}(p,<_i), [n])$, i.e. the diagram obtained from $W$ by removing both $p$ and all propagators inside $p$. The diagram $W'$ has strictly fewer propagators than $W$ did.

Observe that $I_i^W$ and $I_b^{W'}$ are locally the same from vertex $b$ onwards, since by the above observations the propagators removed from $W$ to create $W'$ are exactly those that come first in the ordering defined by $I_i^W$. Similarly, $I_{i+1}^W$ and $I_{b+1}^{W'}$ are locally the same from $b+1$ onwards. By Remark~\ref{rmk algorithm locally same}, this yields the equalities
\begin{align*}
  I_i^{W} \cap [b,i-1] & = I_b^{W'}, \\
  I_{i+1}^{W} \cap [b+1,i-1] & = I_{b+1}^{W'}.
\end{align*} 
Since we have already noted that $m \geq_i b+1$, this implies that $m\in I_b^{W'}\setminus\{b\}$ and $m\not\in I_{b+1}^{W'}$. This contradicts the minimality of $W$.

\textbf{Case II}: Suppose $p$ has one end on edge $i$.  Note that by assumption we have $I_i(p) = i$; this means that $p$ must be the clockwise-most propagator lying on edge $i$, and hence we must have $I_{i+1}(p) = i+1$ as well.  Observe also that $m\geq_i i+2$ since $i+1\in I_{i+1}$.

Let $W'$ be the diagram obtained from $W$ by removing the propagator $p$.  We have
\begin{align*}
  I_i^{W} \backslash \{i\} & = I_{i+1}^{W'} \\
  I_{i+1}^{W} \backslash \{i+1\} & = I_{i+2}^{W'}\;,
\end{align*}
since in both cases the algorithm in $W'$ proceeds identically to that in $W$ after assigning $p$. Since $m \neq i+1$, we have $m\in I_{i+1}^{W'}\setminus\{i+1\}$ but $m\not\in I_{i+2}^{W'}$, contradicting the minimality of $W$. This completes the proof.
\end{proof}

Thus we have shown that Algorithm \ref{alg:put GN on WLD} does produce a Grassmann necklace. Next we check that it produces the correct one, i.e. the Grassmann necklace that defines $M(W)$. 
\begin{thm}\label{res alg gives correct GN} 
The Grassmann necklace $(I_1 \ldots I_n)$ from Theorem~\ref{res:alg gives GN} is the Grassmann necklace of the positroid $M(W)$.
\end{thm}
\begin{proof}
We have shown in the previous theorem that $(I_1,\dots,I_n)$ is a Grassmann necklace; it remains to check that this Grassmann necklace corresponds to the positroid $M(W)$. Recall from the discussion following Definition~\ref{def:grassmann necklace}, and equation \eqref{basesofmatroids} in particular, that it suffices to show the following points.
\begin{itemize}
\item For each $i \in [n]$, $I_i$ is a basis for $M(W)$.
\item If $J$ is lexicographically smaller than $I_i$ with respect to $<_i$ (for any $i \in [n]$) then $J$ is not a basis for $M(W)$. 
\end{itemize}
The first point follows immediately.  The algorithm assigns each propagator to a vertex in the support of that propagator, so $I_i$ is a basis for $M(W)$ by Corollary~\ref{lem basis as perm}.

Now fix $i \in [n]$, and write $I_i = \{a_1 <_i a_2 <_i \dots <_i a_k\}$. Label the propagators of $W$ according to the ordering imposed by $I_i$, i.e. so that $I_i(p_{r}) = a_{r}$ for $1 \leq r \leq k$. Define $\cP_r: = \{p_1,p_2,\dots,p_r\}$ for each $r\geq 0$, with $\cP_0 = \emptyset$.


Let $J = \{j_1 <_i j_2<_i \dots <_i j_k\}$ be a $k$-set, and suppose that $J<_i I_i$ in the $<_i$-lexicographic order. Let $m$ be the minimal index such that $j_m <_i a_m$, i.e. the first entry at which $J$ and $I_i$ differ. We will show that $J$ must be a dependent set.

If $\Prop(j_m) = \emptyset$ then $J$ is automatically dependent since the corresponding column of $C(W)$ contains only zeros, so suppose henceforth that $\Prop(j_m) \neq \emptyset$. 

In order to exhibit a dependent subset of vertices in $J$, we first define a sequence of sets of propagators by
\[Q_0 = \Prop(j_m), \qquad Q_{l+1} = \Prop(I_i(Q_l))\quad \forall l \geq 0.\]
It is easily verified that this is an increasing chain of sets.  Indeed, if $q \in Q_l$ then $I_i(q) \in I_i(Q_l)$, and since $I_i(q)$ is one of the vertices that supports $q$ we must have $q \in \Prop(I_i(Q_l)) = Q_{l+1}$. Since $\cP$ is finite, the sequence stabilizes in finitely many steps. Let $Q := \bigcup_{l \geq 0} Q_l$ be this set of propagators.

We will show that $Q \subseteq \mathcal{P}_{m-1}$, which will imply the theorem by the following argument.  If $Q \subseteq \cP_{m-1}$ then we have $I_i(Q)\subseteq J$, and hence $I_i(Q) \cup \{j_m\} \subseteq J$ as well. Since $I_i(Q) \cup \{j_m\}$ only supports the propagators in $Q$ but has $|Q|+1$ vertices, it is a dependent set in $M(W)$. Thus $J$ is also a dependent set, and the Grassmann necklace $(I_1,\dots,I_n)$ constructed from $W$ does indeed define the positroid $M(W)$.

All that remains is to prove $Q\subseteq \mathcal{P}_{m-1}$, which we do by inductively proving $Q_l \subseteq \mathcal{P}_{m-1}$ for all $l \geq 0$. The base case $Q_0 \subseteq \cP_{m-1}$ is immediate; since $j_m \not\in I_i$ by construction, there cannot be any unassigned propagators supported on $j_m$ when the algorithm reaches that vertex, i.e. ${\Prop(j_m) \subseteq \cP_{m-1}}$.

Now assume that $Q_0,\dots,Q_{l-1} \subseteq \cP_{m-1}$, and consider $Q_{l}$. If $Q_{l} \setminus Q_{l-1} = \emptyset$ then we are done, with $Q = Q_l$.  Otherwise, let $q \in Q_{l} \setminus Q_{l-1}$. By construction there is some propagator $p_r \in Q_{l-1}$ such that $p_r$ and $q$ are both supported on $I_i(p_r) = a_r$. If $q$ is clockwise of $p_r$ on vertex $a_r$ then we immediately have $q \in \cP_{m-1}$ by Remark~\ref{clockwise ordering rem}.  It remains to consider the case that $q$ is counterclockwise of $p_r$ at $a_r$.

Note that we cannot have $a_r+ 1 = j_m$ when $q$ is counterclockwise of $p_r$, since $q \not\in \Prop(j_m)$ by assumption. If the other end of $p_r$ (i.e. the end not supported on $a_r$) lies on an edge between $a_r$ and $j_m$, then since $q$ is inside $p_r$ with respect to $<_{a_r-1}$ we must have $V(q) \subseteq [a_r-1,j_m]$, and hence $q \in \cP_{m-1}$ by Lemma~\ref{lem no fourth vertex}; this is illustrated in Figure~\ref{fig gn proof induction 1}. In particular, if $l=1$ then $p_r \in Q_{0} = \Prop(j_m)$ and we must be in this case. 

If not, then there must be some propagator $p_s \in Q_{l-2}$ such that $p_r$ and $p_s$ are both supported on $I_i(p_s) := a_s$; see Figure~\ref{fig gn proof induction 2}. Since $p_r$ is clockwise of $q$, and $q$ cannot be supported on $I_i(p_s)$ since $q \not\in Q_{l-1}$, it follows that $I_i(p_s) <_i I_i(p_r)$. Therefore, $q$ is inside $p_s$ with respect to $<_{a_s - 1}$ as well.

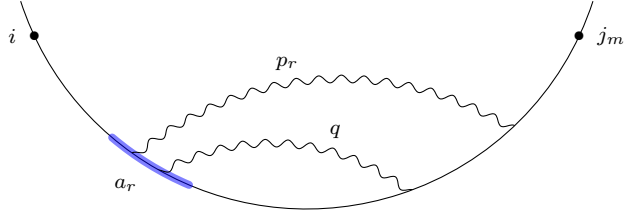
\begin{figure}
\[
\begin{tikzpicture}[baseline=(current bounding box.east)]
  \begin{scope}
  \drawWLDfragment[10]{4}{0.4} 
  \node at (270:\radius*0.2) {}; 
  \node at (270:\radius*1.2) {}; 
  \centerarc[linehighlight](0,0)(\zero + 2.7*\step:\zero + 3.9*\step:\radius)
  \newnode{1}{i}
  \newnode[right]{10}{j_m}
  \begin{scope}
\clipcenterarc(0,0)(\startpoint:\endpoint:\radius)
\newpropbend{3}{8.5}{40}
\newpropbend{3.3}{7}{40}
\end{scope}
\node at (\zero+5*\step:\radius*0.52) {\scriptsize $p_r$};
\node at (\zero + 6*\step:\radius*0.75) {\scriptsize $q$};
\node at (\zero + 3.2*\step:\radius*1.1) {\scriptsize $a_r$};
\end{scope}
\end{tikzpicture}
\]
\caption{If $q$ is counterclockwise of $p_r$ and the other end of $p_r$ lies on an edge between $a_r$ and $j_m$, then $V(q) \subseteq [a_{r-1},j_m]$. The vertex $a_r$ lies somewhere within the (blue) highlighted region.}\label{fig gn proof induction 1}
\end{figure}

\begin{figure}
\[
\begin{tikzpicture}[baseline=(current bounding box.east)]
  \begin{scope}
  \drawWLDfragment[10]{4}{0.35} 
  \node at (270:\radius*0.3) {}; 
  \node at (270:\radius*1.2) {}; 
  \centerarc[linehighlight](0,0)(\zero + 4.4*\step:\zero + 5.2*\step:\radius)
  \newnode{1}{i}
  \newnode[below]{6}{a_r}
  \newnode[right]{10}{j_m}
  \begin{scope}
\clipcenterarc(0,0)(\startpoint:\endpoint:\radius)
\clip (\zero + 5.5*\step:\radius*1.5) circle (\radius);
\newpropbend{5}{18}{30} 
\newpropbend{5.5}{15}{30} 
\newpropbend{6.3}{12}{30} 
\end{scope}
\node at (\zero+5.4*\step:\radius*0.65) {\scriptsize $p_r$};
\node at (\zero+3.7*\step:\radius*0.71) {\scriptsize $p_s$};
\node at (\zero + 7.3*\step:\radius*0.75) {\scriptsize $q$};
\node at (\zero + 4.8*\step:\radius*1.07) {\scriptsize $a_s$};
\end{scope}
  \begin{scope}[shift = {(8.5,0)}]
  \drawWLDfragment[10]{4}{0.35} 
  \centerarc[linehighlight](0,0)(\zero + 2.3*\step:\zero + 3.7*\step:\radius)
  \centerarc[linehighlight](0,0)(\zero + 6.1*\step:\zero + 7.5*\step:\radius)
  \newnode{1}{i}
  \newnode[right]{10}{j_m}
  \begin{scope}
\clipcenterarc(0,0)(\startpoint:\endpoint-\step:\radius)
\newpropbend{2.9}{15}{30} 
\newpropbend{3}{6.8}{45} 
\newpropbend{7}{15}{30} 
\end{scope}
\node at (\zero+5.3*\step:\radius*0.75) {\scriptsize $p_r$};
\node at (\zero+2.1*\step:\radius*0.8) {\scriptsize $p_s$};
\node at (\zero + 7.7*\step:\radius*0.8) {\scriptsize $q$};
\node at (\zero + 3*\step:\radius*1.07) {\scriptsize $a_s$};
\node at (\zero + 7*\step:\radius*1.07) {\scriptsize $a_r$};
\end{scope}
\end{tikzpicture}
\]
\caption{If the other end of $p_r$ does not lie on an edge between $a_r$ and $j_m$, then we must have a propagator $p_s \in Q_{l-2}$ in one of the depicted configurations. The (blue) highlighted regions indicate the approximate position of a vertex whose precise position is unknown. Cut-off propagators indicate that the location of the propagator's other end is unknown; in particular, no meaning should be ascribed to the angle at which these propagators are drawn.}\label{fig gn proof induction 2}
\end{figure}
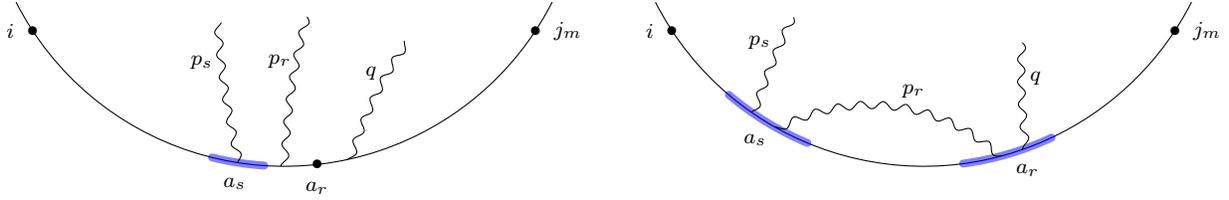
Now repeat the argument above with $p_s$ in place of $p_r$.  If the other end of $p_s$ lands on an edge between $q$ and $j_m$ then we have $q \in \cP_{m-1}$ as above; if not then there exists some $p_t \in Q_{l-3}$ with $p_s$ supported on $I_i(p_t)$ and $q$ on the inside of $p_t$ in the $a_t-1$ order. Therefore, $I_i(p_t) <_i I_i(p_r)$.
Since there are finitely many steps this process can take before reaching $Q_0$, and any $p_* \in Q_0  = \Prop(j_m)$ has $j_m$ in its support by definition, this argument eventually terminates with the construction of a propagator constraining the positioning of $q$ as required. This completes the induction step, and hence the proof.\end{proof}

We have now shown that Algorithm~\ref{alg:put GN on WLD} associates the correct Grassmann necklace to a Wilson loop diagram. Therefore, we refer to Algorithm~\ref{alg:put GN on WLD} as the Grassmann necklace algorithm. We close this section with a simple observation about the Grassmann necklaces of positroids associated to Wilson loop diagrams.

In an arbitrary Grassmann necklace, it is possible for an index $i$ to appear in no terms of the Grassmann necklace (a {\bf loop}) or in all terms of the necklace (a {\bf coloop}). Using Theorems~\ref{res:alg gives GN} and \ref{res alg gives correct GN}, a characterization of the loops and coloops of the Grassmann necklace associated to a Wilson loop diagram follows easily.

\begin{cor}\label{no coloops}
Grassmann necklaces coming from admissible Wilson loop diagrams have no coloops. A vertex $j$ is a loop if and only if $j$ supports no propagators. 
\end{cor}
\begin{proof}
For any $i \in [n]$, $i-1$ is maximal with respect to the $<_{i}$ order and so by Lemma~\ref{lem no fourth vertex} there can be no propagator $p$ with $I_{i}(p) = i-1$. Thus $i-1 \not\in I_{i}$, so $i-1$ is not a coloop. Since this holds for any $i \in [n]$, the Grassmann necklace admits no coloops.

If $j \in [n]$ is a loop then in particular we have $j \not\in I_j$, which can only happen if there are no propagators supported on vertex $j$. Conversely, if $j$ supports no propagators, then Algorithm~\ref{alg:put GN on WLD} never assigns a propagator to $j$ and hence $j \not\in I_i$ for all $i \in [n]$.
\end{proof}

\section{Dimension of the Wilson loop cells}\label{sec dim}

In continuing to characterize the positroid cells associated to Wilson loop diagrams, our next goal is to show that the dimension of the positroid cell defined by a Wilson loop diagram $(\cP, [n])$ is $3|\cP|$.  This is achieved in Theorem \ref{thm dim}.  

Note that the positroid cells associated to the Amplituhedron are of dimension $4|\cP|$. In \cite[section 2]{non-orientable}, Agarwala and Marcott associated a Deodhar component\footnote{Since Deodhar components play no role in this paper, the interested reader is referred to \cite{non-orientable} for more details on this.} to each Wilson loop diagram, and showed that the dimension of this Deodhar component is the sum of $|\cP|$ and the dimension of the associated positroid cell.
Combining this with Theorem~\ref{thm dim}, we see that the Deodhar component associated to a Wilson loop diagram is $4|\cP|$-dimensional, showing that the geometry underlying the Wilson loop diagrams is consistent with the geometry underlying the Amplituhedron. It is worth noting that the Amplituhedon is defined by a projection of a set of $4|\cP|$ dimensional positroid cells (i.e. subspaces of $\Gr(k,n)$), while the corresponding $4|\cP|$-dimensional cells of the Wilson loop diagram lie in $\mathbb{G}_\R(k,n+1)$.
While it is conjectured that the geometry of the Amplituhedron is related to the geometry of the Wilson loop diagrams, in part due to their common physical origin, the $3 |\cP|$ dimensional positroid cells defined by the Wilson loop diagrams do not seem to bear any geometric relationship to the $4|\cP|$ dimensional positroid cells parametrized by the Amplituhedron.

Marcott in \cite[Theorem 8.4]{WLDdim} also gave a proof of the $3|\cP|$-dimensionality of $M(W)$ which is much more geometric. However, it is not easy to track the effect of a particular propagator in that proof. By contrast, our approach is much more constructive, and elucidates how each additional propagator contributes to the dimension. By combining Algorithms~\ref{alg:put GN on WLD} and \ref{alg:GN to Le} (converting from Wilson loop diagram to Grassmann necklace, and Grassmann necklace to Le diagram respectively) we explicitly describe the effect of adding another propagator to a Wilson loop diagram in terms of the plusses of the associated Le diagrams. Recall that the dimension of a positroid cell is equal to the number of plusses in its associated Le diagram \cite[Theorem 6.5]{Postnikov}.  Thus understanding the effect of adding another propagator is sufficient to give a recursive proof of the $3|\cP|$-dimensionality of the cells.


In order to prove our main theorem of this section, we first require several lemmas. 

We begin by relating non-supporting vertices to columns of zeros in the Le diagram (Lemma \ref{lem uncovered}) and verifying that performing dihedral transformations on a Wilson loop diagram does not change the dimension of the associated Le diagram (Lemma \ref{lem dihedral}). This allows us to restrict our attention to a particular subset of admissible Wilson loop diagrams (Lemma~\ref{lem good p}) whose Grassmann necklaces can be more easily described. 

\vspace{0.5em} 
\begin{lem}\label{lem uncovered}
Let $W$ be an admissible Wilson loop diagram with $k$ propagators, and with a vertex $i$ that supports no propagators.  Let $V$ be $W$ with vertex $i$ removed.  Then the Le diagram of $W$ is obtained from the Le diagram of $V$ by inserting an extra column containing all $0$s in position $i$ (i.e. such that the new column has the label $i$).
\end{lem} 

\begin{proof}
  By Algorithm~\ref{alg:put GN on WLD} the Grassmann necklace of $W$ is obtained from the Grassmann necklace of $V$ by duplicating the $i$th element of the Grassmann necklace of $V$, and incrementing all indices greater than $i$ by $1$ in each Grassmann necklace element.  Formally, 
  \[
  I_j^{W} =
  \begin{cases}
    \{\ell \in I_j^{V} : \ell < i\} \cup \{\ell+1 \in I_j^{V} : \ell \geq i\} & \text{if $j\leq i$} \\
    \{\ell \in I_{j-1}^{V} : \ell < i\} \cup \{\ell+1 \in I_{j-1}^{V} : \ell \geq i\} & \text{if $j > i$.}
  \end{cases}
  \]

By Corollary~\ref{no coloops} we know that $i \not\in I_1^W$, and so $i$ must label a horizontal edge on the boundary of the Le diagram of $W$, i.e. it must be a column label. The shapes of the Le diagram of $V$ and $W$ are the same except for the insertion of this column, since $I_1$ is the same for $V$ and $W$ except for the incrementation of the indices  greater than or equal to $i$ in the transition from the necklace for $V$ to the necklace for $W$. 
\end{proof}
\vspace{0.5em}

\begin{lem}\label{lem dihedral}
If two Wilson loop diagrams differ by a dihedral transformation then their Le diagrams have same number of plusses.
\end{lem}
\begin{proof}
Recall that the number of plusses in a Le diagram is given by the dimension of the corresponding positroid. By \cite[Proposition 17.10]{Postnikov}, the dimension of a positroid $M$ is $k(n-k) - A(\pi_M)$, where $A(\pi_M)$ denotes the number of alignments of the decorated permutation $\pi_M$ of $M$. The alignments of a decorated permutation are defined in terms of its chord diagram representation, where an alignment is any pair of non-crossing chords that share the same orientation. See \cite[Figure 17.1]{Postnikov} and preceding discussion for definitions and more details.

With the help of Algorithm \ref{alg:put GN on WLD}, it is now easy to trace the effect of a dihedral transformation on a Wilson loop diagram through its Grassmann necklace to its decorated permutation to its chord diagram, and see that dihedral transformations of a Wilson loop diagram $W$ correspond to dihedral transformations of the chord diagram representation of $\pi_{M(W)}$. 

Since the number of alignments in a chord diagram is clearly preserved under dihedral transformations, the result follows.
\end{proof}

\vspace{0.5em}

\begin{lem}\label{lem good p}
  Let $W$ be an admissible Wilson loop diagram with $k \geq 1$ propagators.  Then there is some dihedral transformation of $W$ such that the resulting diagram $W'$ has a propagator $p$ with the following properties.
  \begin{itemize}
  \item The propagator $p = (i, n-1)$ for some $i$, and $p$ has no propagators inside it with respect to the $<_i$ ordering.
  \item Either the edge $i$ in $W'$ only supports $p$, or the edge $i$ in $W'$ supports exactly one other propagator $q = (j,i)$, with no other propagators inside $q$ in the $<_j$ ordering. 
  \end{itemize}
\end{lem}
Note that since the first condition tells us that $p$ has nothing inside it, if $q$ exists in the second condition then $j<_ji<_jn-1$. 

\begin{proof}
  Remove all vertices of $W$ which do not support any propagators to get a weakly admissible Wilson loop diagram $V$.  Lemma~\ref{lem sian} applied to $V$ gives a length 2 propagator $p$ in $V$ for which either no other propagator is supported on one of the supporting edges of $p$ or there is a second length 2 propagator which is the only other propagator supported on one of the supporting edges of $p$.  (Figure~\ref{fig 3 cases} shows the possible configurations arising from Lemma~\ref{lem sian}, and the reader can easily check that in each case $p$ must be in one of the two situations described above.)

  We can now make a dihedral transformation of $V$ to obtain a diagram satisfying the statement of the lemma with $p$ and $q$ both length 2. Restoring the vertices which do not support any propagators, we obtain a dihedral transformation $W'$ of $W$ as desired (with potentially longer lengths for $p$ and $q$).
\end{proof}

Combining Lemmas \ref{lem uncovered}, \ref{lem dihedral}, and \ref{lem good p}, it therefore suffices to study the Le diagrams of weakly admissible Wilson loop diagrams with no non-supporting vertices and admitting one of the configurations described in Lemma \ref{lem good p}, with propagators $p = (n-3, n-1)$ and $q= (n-5, n-3)$ (if $q$ exists) both of length 2. See Figure~\ref{fig special p} for an illustration of the two possibilities. 

With the above restrictions in place, the next few lemmas describe how the Grassmann necklaces (Lemma \ref{lem I}) and the Le diagrams (Lemmas \ref{lem shape} to Lemma \ref{lem other k}) change upon the removal of $p$ from the Wilson loop diagram.  This gives us the inductive lever we need to prove the main theorem (Theorem \ref{thm dim}).

\begin{figure} \label{fig:twocases}
\[
\begin{tikzpicture}[rotate = 150,baseline=(current bounding box.east),label distance = -2mm]
\begin{scope}
  \drawWLDfragment[6]{3}{0.3}
  \begin{scope}
    \clipcenterarc(0,0)(\startpoint:\endpoint:\radius)
    \draw[smallpropagator] (\zero+2.4*\step:\radius*1.1) to [bend left = 55] (\zero + 4.6*\step:\radius*1.1); 
    \node at (\zero + 3.5*\step:\radius*0.75){\footnotesize $p$};
  \end{scope}
  \centerarc[red,line width = 2pt,line cap = round](0,0)(\zero + 2*\step:\zero+4.45*\step:\radius)
  \newbetternode[{[label distance=-3mm]above right}]{1}{n-4}
  \newbetternode[{[label distance=-3mm]above right}]{2}{n-3}
  \newbetternode[{[label distance=-3mm]above right}]{3}{n-2}
  \newbetternode[{[label distance=-3mm]above right}]{4}{n-1}
  \newnode[above]{5}{n}
  \newnode[above]{6}{1}
  \node at (180-60:\radius*0.1){\em Case 1};
\end{scope}
\end{tikzpicture}
\qquad 
\qquad 
\begin{tikzpicture}[rotate=150, baseline=(current bounding box.east)]
\begin{scope}
  \drawWLDfragment[7]{3}{0.3}
  \begin{scope}
    \clipcenterarc(0,0)(\startpoint:\endpoint:\radius)
    \draw[smallpropagator] (\zero+1.4*\step:\radius*1.1) to [bend left = 60] (\zero + 3.5*\step:\radius*1.1); 
    \draw[smallpropagator] (\zero+3.5*\step:\radius*1.1) to [bend left = 60] (\zero + 5.6*\step:\radius*1.1); 
    \node at (\zero + 2.4*\step:\radius*0.8){\footnotesize $q$};
    \node at (\zero + 4.5*\step:\radius*0.8){\footnotesize $p$};
  \end{scope}
  \centerarc[red,line width = 2pt,line cap = round](0,0)(\zero + 1.6*\step:\zero+5.5*\step:\radius)
  \newbetternode[{[label distance=-3mm]above right}]{1}{n-5}
  \newbetternode[{[label distance=-3mm]above right}]{2}{n-4}
  \newbetternode[{[label distance=-3mm]above right}]{3}{n-3}
  \newbetternode[{[label distance=-3mm]above right}]{4}{n-2}
  \newbetternode[{[label distance=-3mm]above right}]{5}{n-1}
  \newnode[above]{6}{n}
  \newnode[above]{7}{1}
  \node at (180-60:\radius*0.1){\em Case 2};
\end{scope}
\end{tikzpicture}
\] 
  \caption{The two cases for $W$ and $p$ following Lemma~\ref{lem good p}.  No other propagators can end in the bolded (red) sections.  Other segments may have additional propagators ending in them.}\label{fig special p}

\end{figure}
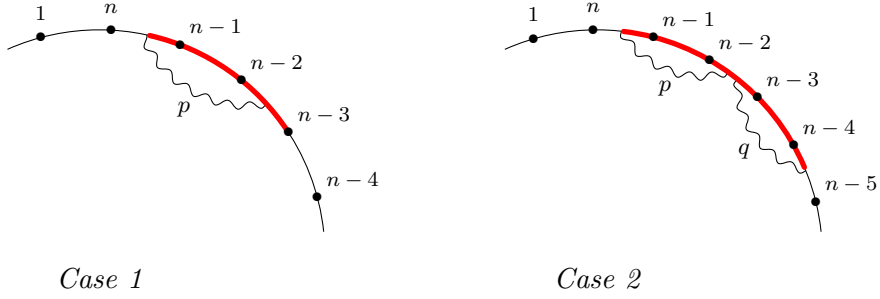

\begin{lem}\label{lem I}
Let $W$ be an admissible Wilson loop diagram with $k\geq 1$ propagators, and suppose $W$ admits one of the configurations described in Lemma~\ref{lem good p}, with $p$ and $q$ (if $q$ exists) both of length 2. Let $V$ be $W$ with $p$ removed.  We can express $I_*^W$ in terms of $I_*^V$ as indicated below. If $k \in \{1, n-1, n-2, n\}$, 

  \begin{align*}
    I_1^{W} & = I_1^{V} \cup \{n-3\} \\
    I_n^{W} & = I_1^{V} \cup \{n\} \\
    I_{n-1}^{W} & = I_n^{V} \cup \{n-1\} \\
    I_{n-2}^{W} & =
    \begin{cases}
      I_{n-2}^{V}\cup \{n-2\} & \text{if $n-2\not\in I_{n-2}^{V}$} \\
      I_{n-2}^{V}\cup \{n-1\} & \text{if $n-2\in I_{n-2}^{V}$, $n-1\not\in I_{n-2}^{V}$} \\
      (I_{n}^{V} \setminus \{n-5\})\cup \{n-1,n-2\} & \text{if $n-1, n-2\in I_{n-2}^{V}$}
    \end{cases} \;. \end{align*}
For all other $k$, i.e. $1<k<n-2$,
   \begin{align*}
    I_{k}^{W} & =
    \begin{cases}
      I_k^{V}\cup \{n-3\} & \qquad \qquad \qquad \qquad \text{if $n-3 \not\in I_k^{V}$}\\
      I_k^{V}\cup\{n-2\} & \qquad \qquad \qquad \qquad \text{if $n-3\in I_k^{V}$}
    \end{cases}
  \end{align*}
\end{lem}

\begin{proof}
The two possible cases for $W$ are illustrated in Figure~\ref{fig special p}. 

We begin with $I_1$, and first observe that by Remark~\ref{rmk cyclic} we must have $I_1^W(p) = n-3$. Since $p$ is the final propagator to be assigned by $I_1^W$, it follows that $I_1^W$ and $I_1^V$ behaved identically on all earlier propagators (recall Remark~\ref{rmk algorithm locally same} about locally identical propagator configurations), which implies in particular that $n-3 \not\in I_1^V$. Thus $I_1^W = I_1^V\cup\{n-3\}$.

We note for future reference that the arguments of the previous two paragraphs also imply that $n-2$, $n-1$, and $n$ are not in $I_1^V$.

Next consider $I_n^{W}$. From Figure \ref{fig special p} we see that in both cases we have $I_n^{W}(p) = n$. Thus from vertex 1 onwards the unassigned propagators are exactly those that appear in $V$. Therefore, the algorithm continues as in $I_1^{V}$, i.e. we have $I_n^{W} = I_1^{V} \cup \{n\}$. By a similar argument, we must have $I_{n-1}^{W} = I_n^{V} \cup \{n-1\}$.

Now consider $I_{n-2}^{W}$. If $n-2\not\in I_{n-2}^{V}$ then the vertex $n-2$ is non-supporting (see also Corollary~\ref{no coloops}).  Thus we are in Case 1, with ${I_{n-2}^W(p) = n-2}$ and this assignment of $p$ does not affect the rest of the construction of $I_{n-2}^{V}$, so we obtain ${I_{n-2}^{W} = I_{n-2}^{V}\cup \{n-2\}}$ as above.

On the other hand, if $n-2\in I_{n-2}^{V}$ then we must be in Case 2 of Figure~\ref{fig special p} and we must have ${I_{n-2}^V(q) = I_{n-2}^W(q) = n-2}$ since $q$ is always the clockwise-most propagator supported on vertex $n-2$. If $n-1 \not\in I_{n-2}^{V}$, then $I_{n-2}^W(p) = n-1$ and the algorithm proceeds identically to $I_{n-2}^{V}$ for the remainder of its steps; thus $I_{n-2}^{W} = I_{n-2}^{V}\cup \{n-1\}$ in this case.

Finally, if $n-2$ and $ n-1$ are in $I_{n-2}^{V}$ then $I_{n-2}^W(q) = n-2$ and $I_{n-2}^W(p) = n-1$ as above, but some other propagator must have contributed $n-1$ to $I_{n-2}^{V}$ so we cannot use the same argument as above.  When, in the algorithm, the construction of $I_{n-2}^W$ reaches vertex $n$, only the propagators $p$ and $q$ have contributed; thus we proceed as in the construction of $I_n^{V}$ but without propagator $q$ by Remark \ref{rmk algorithm locally same}. By Lemma~\ref{vertex cyclic int lem} we know that $I_n^V(q) = n-5$, and from Remark \ref{clockwise ordering rem} we see that the only way this could occur is if $q$ was the last to contribute to $I_{n-2}^V$. 
Therefore, $I_{n-2}^{W} = (I_{n}^{V} \setminus \{n-5\}) \cup \{n-1, n-2\}$ in this case. This completes all cases for $I_{n-2}^{W}$.

The arguments for $I_k^{W}$ ($1< k < n-2$) proceed analogously to those of $I_{n-2}^{W}$, with one simplification: we cannot have both $n-3$ and $n-2$ in $I_k^{V}$ since $q$ is the only propagator that could be assigned to either of them, and it cannot be assigned to both.

This covers all cases and hence completes the proof.
\end{proof}

The next several lemmas assume $V$ and $W$ to be as above. These lemmas discuss how the shape of the Le diagram changes from $V$ to $W$ and how the plusses in $V$ shift to become plusses in $W$. In particular, the Le diagram for $W$ has an extra row, labeled by $n-3$, with three columns (see Figure~\ref{fig Le} and Lemma~\ref{lem shape}). This last row contains three (new) plusses. All other plusses in the Le diagram of $V$ shift in a predictable manner to become plusses in the Le diagram of $W$; this is described in Lemma~\ref{lem n and n-1} to Theorem~\ref{thm dim}.

\begin{rmk}To avoid wordiness, we will refer to the row indexed by $i$ in the Le diagram as row $i$ or the $i$ row, and similarly for columns.  Note that the $i$ row in this sense is not the $i^{th}$ row since the indexing for Le diagrams does not label the rows sequentially starting at $1$, and so to avoid confusion we will not refer to the $i^{th}$ row. \end{rmk}

We proceed to address the cases laid out in Lemma~\ref{lem I}, getting the shape of the Le diagram of $W$ from the relation $I_1^W= I_1^V \cup \{n-3\}$. 

\begin{lem}\label{lem shape}
  Let $V$ and $W$ be as in Lemma~\ref{lem I}.
  The shape of the Le diagram of $V$ can be built from left to right of the following blocks: a rectangle which is 3 columns wide, one more column of the same height, and a partition shape with at most as many rows as the rectangle.
  The shape of the Le diagram of $W$ can be built from left to right of the following blocks: a rectangle with 3 columns and one more row than the first rectangle of $V$, and the same partition shape as in $V$.
\end{lem}

\begin{proof}
See Figure~\ref{fig Le} for an illustration of the shapes described in the statement of the lemma.

Recall that $I_1$ determines the shape of the Le diagram.  Specifically, the elements of $I_1$ label the rows of the diagram, working from top right to bottom left. By the arguments in Lemma {\ref{lem I}} we know that none of $n$, $n-1$, $n-2$, and $n-3$ are in $I_1^{V}$.  These are the leftmost four columns of the diagram for $V$.  From Lemma \ref{lem I} we know that that $I_1^{W} = I_1^{V}\cup \{n-3\}$. Therefore, the right hand boundary of the shape of $V$ is the same as the right hand boundary of the shape of $W$ except that $W$ has one additional row of 3 boxes while $V$ has an additional column in the $n-3$ position; that is, an extra column fourth from the left. 
\end{proof}
\begin{figure}
  \[\begin{tikzpicture}[scale = 1.5]
\node at (1.5,2.5) {\Large $V$};
\draw (0,0) rectangle (1,2);
\draw (1,0) rectangle (1.3,2);
\begin{scope}[shift = {(-2.7,-2)}]
\draw (4,2) .. controls (5.3,2.2) and (4.7,3) .. (5.5,3.2) .. controls (6.3,3.5) .. (6.5,4) -- (4,4);
\end{scope}
\node at (0.5,1) {\large$\mathcal{A}$};
\node at (1.8,1) {\large$\mathcal{B}$};
\begin{scope}[shift = {(5,0)}]
\node at (1.5,2.5) {\Large $W$};
\draw (0,-0.3) rectangle (1,0);
\draw (0,0) rectangle (1,2);
\begin{scope}[shift = {(-3,-2)}]
\draw (4,2) .. controls (5.3,2.2) and (4.7,3) .. (5.5,3.2) .. controls (6.3,3.5) .. (6.5,4) -- (4,4);
\end{scope}
\node at (0.5,1) {\large$\mathcal{A}$};
\node at (1.5,1) {\large$\mathcal{B}$};
\end{scope}
\end{tikzpicture}\]
  \caption{Le diagrams for $V$ (left) and $W$ (right).}\label{fig Le}
\end{figure}
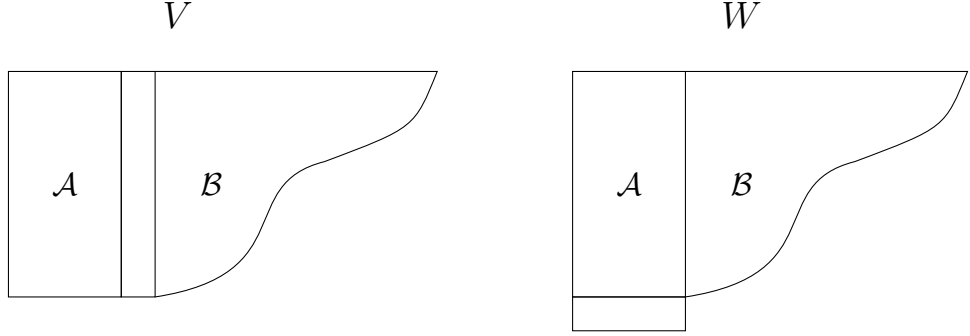

As illustrated in Figure~\ref{fig Le}, the pieces of the Le diagrams of $V$ and $W$ will be called $\mathcal{A}$ and $\mathcal{B}$ respectively, in what follows. Explicitly, $\mathcal{B}$ is the set of columns with labels $v < n-3$ and $\mathcal{A}$ is the set of columns labeled by $n$, $n-1$ and $n-2$ in $V$. In $W$ we see that $\mathcal{A}$ consists of these same three columns, but does not include the last row. Over the course of the next few lemmas we will prove that the plusses in the $\mathcal{B}$ parts of each diagram are identical, and that the relationship between the plusses in the two $\mathcal{A}$ regions can be described explicitly.

We do this by applying Algorithm~\ref{alg:GN to Le}, which constructs the Le diagram associated to the Grassmann necklaces of $V$ and $W$. As described in Section~\ref{sec:positroid background}, if Algorithm~\ref{alg:GN to Le} places a $+$ in the box with row index $i$ and column index $j$, we say that this plus is {\bf in the $i \rightarrow j$ position}, and refer to it as {\bf the plus defined by ``the (hook) path from $i$ to $j$''}. Note that the collection of paths contributed by a single Grassmann necklace term must be non-crossing.  This follows immediately from the Le condition.


\begin{lem}\label{lem n and n-1}
Let $V$ and $W$ be as in Lemma~\ref{lem I}.
Then
%
  $I_n^W$ yields an $(n-3)\rightarrow n$ plus, and $I_{n-1}^W$ yields an $(n-3)\rightarrow (n-1)$ plus along with any plusses yielded by $I_n^V$.
\end{lem}

\begin{proof}
By Lemma~\ref{lem I} we have $I_n^{W}= I_1^{V} \cup \{n\}$ and $I_1^{W} = I_1^{V} \cup \{n-3\}$. Thus 
\[I_1^{W} \setminus I_n^{W} = \{n-3\}, \qquad \qquad I_n^{W} \setminus I_1^{W} = \{n\},\]
and so by Algorithm~\ref{alg:GN to Le} we have a plus in the $(n-3) \rightarrow n$ position, i.e. in the leftmost box of the bottom row. 

Also by Lemma~\ref{lem I} we have $I_{n-1}^{W} = I_n^{V} \cup \{n-1\}$. Recall from the proof of Lemma~\ref{lem I} that none of $n-3,\;n-2,\;n-1$ or $n$ are in $I_1^V$; therefore,
\[
I_1^W\setminus I_{n-1}^W = (I_1^V \setminus I_n^V) \cup\{n-3\},
\]
and $n-3$ is maximal in this set. Similarly, 
\[I_{n-1}^W \setminus I_1^W = (I_n^V \setminus I_1^V) \cup \{n-1\}.
\]
Also $(I_n^V \setminus I_1^V) \cup \{n-1\}\subseteq \{n-1,n\}$ from the definition of Grassmann necklace, and so in particular $n-1$ is minimal in this set. Thus Algorithm~\ref{alg:GN to Le} yields a plus in the $(n-3) \rightarrow (n-1)$ position (i.e. the second box in the bottom row of the Le diagram of $W$) along with any plusses yielded by $I_n^V$.
\end{proof} 

We next consider the plusses contributed to the Le diagram of $W$ by $I_{n-2}^W$. Lemma~\ref{lem I} lists three possibilities for $I_{n-2}^W$, but the following observation shows that we need only consider two of these cases.

\begin{rmk}\label{rem n-2 case} Consider the case when $n-2 \in I_{n-2}^V$ but $n-1 \not\in I_{n-2}^V$. This case can only occur if our $W$ is in Case 2 of Figure~\ref{fig special p} and $n-1$ supports no propagators other than $p$. In this situation, apply the dihedral transformation that reflects $W$ along the line perpendicular to the $(n-2)^{th}$ edge, so that $p$ is unchanged but $q = (n-1,1)$ in the new diagram. The new diagram is an example of Case 1, i.e. it satisfies $n-2 \not \in I_{n-2}^V$. Therefore, we do not need to separately consider the case where $n-2\in I_{n-2}^V$ but $n-1\not\in I_{n-2}^V$, as it will be subsumed into Case 1 in the main theorem.
\end{rmk}

We can now describe the plusses contributed by $I_{n-2}^W$ in the remaining two cases: when $n-2 \not\in I_{n-2}^V$, and when both $n-2$ and $n-1$ are not in $I_{n-2}^V$.

\begin{lem}\label{lem n-2 good}
  Let $V$ and $W$ be as in Lemma~\ref{lem I}, and suppose that $n-2 \not\in I_{n-2}^{V}$. Then $I_{n-2}^{W}$ contributes all of the same plusses as $I_{n-1}^{V}$ (which is the same as $I_{n-2}^{V}$) along with a new ${(n-3)\rightarrow (n-2)}$ plus.
\end{lem}

\begin{proof}
If $n-2\not\in I_{n-2}^{V}$ then $I_{n-1}^{V}=I_{n-2}^{V}$ by definition, and by Lemma~\ref{lem I} we have the equality ${I_{n-2}^{W} = I_{n-1}^{V} \cup \{n-2\}}$.  Note that $n-3\not\in I_{n-2}^{V}$ by Lemma~\ref{lem no fourth vertex}.  Therefore, the plusses contributed by $I_{n-2}^{W}$ are exactly the  plusses from $I_{n-1}^{V}$ along with the $(n-3)\rightarrow (n-2)$ plus.  This gives the statement of the lemma. 
\end{proof}

Next we need an observation about how the Grassmann necklace for $V$ changes over the indices near $n$ under the hypotheses that $n-2, n-1\in I_{n-2}^V$.

Specifically with $V$ and $W$ as in Lemma~\ref{lem I} and with $n-2, n-1\in I_{n-2}^V$ (so we are necessarily in Case 2 and $n-1$ supports at least one additional propagator) we have
  \begin{equation}\label{eq necklace}
  I_{n-3}^{V} \xrightarrow[\substack{n-3\text{ out}\\n-2\text{ in}}]{} I_{n-2}^{V} \xrightarrow[\substack{n-2\text{ out}\\n-5\text{ in}}]{} I_{n-1}^{V} \xrightarrow[\substack{n-1\text{ out}\\ u \text{ in}}]{} I_{n}^{V}  \xrightarrow[\substack{n \text{ out}\\v \text{ in}}]{} I_1^{V}
  \end{equation}
where $u$ and $v$ are two vertices. We will derive some properties of $u$ and $v$ below.

The first two transitions follow by the cyclic intervals property highlighted in Remark~\ref{rmk cyclic} along with the fact that $n-3$ is a supporting vertex in $V$.  The third transition follows from the definition of Grassmann necklaces and the fact that $n-1\in I_{n-2}^V$ by hypothesis, while the final transition follows from the fact that $n-1$ supports an additional propagator so we are guaranteed $n \in I_n^V$.

Note that neither $u$ nor $v$ can be $n-1$, $n-2$, or $n-3$, because as noted in Lemma~\ref{lem I}, these elements do not belong to $I_1^V$ and $I_n^V = (I_1^V\setminus \{v\}) \cup \{n\}$. However, $u$ could be $n$.  

Also $u$ cannot be $n-4$, for the following reason. By Remark~\ref{rmk cyclic} we know that ${I_{n-1}^W(q) = n-5}$, leaving no propagators supported on $n-4$ for $I_{n-1}^W$ to assign. Thus $n-4 \not\in I_{n-1}^W$. Since ${I_{n-1}^W = I_n^V\cup\{n-1\}}$ by Lemma~\ref{lem I}, we cannot have $n-4$ in $I_n^V$ either.

However, it is important to observe that $v$ could be $n-4$, as we see in the next result.

  \begin{lem}\label{lem n-2 bad}
    Let $V$ and $W$ be as in Lemma~\ref{lem I}, and suppose that $n-2,n-1 \in I_{n-2}^{V}$. Let $c$ be the largest element of $I_1^V$.  Then the following statements are true.
    \begin{enumerate}
  \item If $n-4\in I_1^V\setminus I_n^V$ then $c=n-4$. Otherwise, $c=n-5$ otherwise. If $c = n-4$ then we have $v = n-4$ in equation \eqref{eq necklace}.
  \item The function $I_{n-3}^W$ contributes the same plusses as $I_{n-3}^V$, except that the plus in column $n-3$ (which necessarily exists) is shifted one square left into column $n-2$.
  \item In the Le diagram of $V$, $I_{n-2}^{V}$ contributes a $(c)\rightarrow (n-2)$ plus and no other term in the Grassmann necklace of $V$ gives a plus in this column. All other plusses from $I_{n-2}^{V}$ were already contributed by $I_{n-3}^V$.
  \item The function $I_{n-2}^W$ contributes an $(n-3) \rightarrow (n-2)$ plus, a $(c) \rightarrow (n-1)$ plus, and one other plus. If $c = n-5$ then this third plus is the one contributed by $I_n^V$, while if $c = n-4$ then the third plus is $(n-5) \rightarrow n$.
  \item If $c = n-4$ then $I_{n-1}^V$ contributes a $(n-4) \rightarrow (n-1)$ plus; all other plusses contributed by $I_{n-1}^V$ were already contributed by $I_{n-3}^V$, regardless of the value of $c$.
  \item The Grassmann necklace of $V$ contributes a $(c)\rightarrow (n-1)$ plus if and only if $c = n-4$. No term in the Grassmann necklace of $V$ contributes a $(n-5) \rightarrow n$ plus.
    \end{enumerate} 
\end{lem}

\begin{proof}
\begin{enumerate}  

\item As in Lemma~\ref{lem I} we have that none of $n,\; n-1, \; n-2$ and $n-3$ are in $I_1^V$, which implies the first part of the claim. The second part is a consequence of the discussion following equation \eqref{eq necklace}, which shows that $n-4 \not\in I_n^V$.

\item By Lemma~\ref{lem I} we have that $I_{n-3}^W = I_{n-3}^V \cup \{n-2\}$, so $I_{1}^W \setminus I_{n-3}^W = I_1^V \setminus I_{n-3}^V$ and $I_{n-3}^W \setminus I_1^W$ is obtained from $I_{n-3}^V \setminus I_{1}^V$ by replacing the $n-3$ with $n-2$.    

\item From equation \eqref{eq necklace} the pair of sets $(I_{n-2}^V\setminus I_1^V, I_1^V\setminus I_{n-2}^V)$ is either $(\{n-2, n-1, n\}, \{n-5, u, v\})$ (if $n \in I_{n-2}^V$) or $(\{n-2, n-1\}, \{n-5, v\})$ (if $n \not\in I_{n-2}^V$, i.e. $u = n$).

Recall that we cannot have $u = n-4$ (since $n-4 \notin I_n^V$) so either $u = n$ or $u < n-5$, and either way $n-2$ is minimal in $I_{n-2}^V\setminus I_1^V$ and $c$ is maximal in $I_1^V \setminus I_{n-2}^V$, yielding a $(c) \rightarrow (n-2)$ plus. 

Since $I_{n-2}^V$ is the only term in the Grassmann necklace for $V$ that contains $n-2$, no other term can contribute a plus in column $n-2$. Additionally, since $I_{n-2}^V$ only differs from $I_{n-3}^V$ in that the $n-3$ was replaced by an $n-2$, all other plusses must agree with those contributed by $I_{n-3}^V$. 

\item By Lemma~\ref{lem I} we have $I_1^W = I_1^V\cup \{n-3\}$ and $I_{n-2}^W = (I_n^V \setminus \{n-5\}) \cup \{n-1,n-2\}$. By equation \eqref{eq necklace} we have $n-5 \in I_1^V$, and by the following discussion we have that none of $n-1,\;n-2$ and $n-3$ are not in either $I_n^V$ or $I_1^V$. It follows that the symmetric difference of $I_1^W$ and $I_{n-2}^W$ is
    \begin{gather*}
    I_1^W \setminus I_{n-2}^W  = (I_1^V\setminus I_n^V) \cup \{n-3, n-5\}. \\
    I_{n-2}^W \setminus I_1^W  = (I_n^V\setminus I_1^V) \cup \{n-1, n-2\}.
\end{gather*}
As per the discussion following equation \eqref{eq necklace} above, we must also have $I_n^V\setminus I_1^V = \{n\}$ and ${I_1^V\setminus I_n^V = \{v\}}$. If $c = n-5$ then the path $v \rightarrow n$ from $I_n^V$ does not intersect with the paths from $\{n-3, n-5\}$ to $\{n-2, n-1\}$ and we obtain the plusses as stated.
    
Now suppose $c = v= n-4$; in this case $I_{n-2}^W$ contributes plusses defined by the noncrossing paths starting from $\{n-3, n-4, n-5\}$ and ending at $\{n, n-1, n-2\}$, namely the plusses in positions $(n-3)\rightarrow (n-2)$, $(n-4)\rightarrow (n-1)$ and $(n-5)\rightarrow (n)$.

\item Combining equation \eqref{eq necklace} with point 3 of this proof, the pair $(I_{n-1}^V \setminus I_1^V, I_1^V \setminus I_{n-1}^V)$ is either $(\{n-1, n\}, \{u, v\})$ or $(\{n-1\}, \{v\})$. Per point 1, we have $c = n-4$ if and only if $v = n -4$, in which case we get a $(n-4) \rightarrow (n-1)$ plus and the $u \rightarrow n$ plus (if it exists) has already been contributed by $I_{n-2}^V$ and $I_{n-3}^V$. If $c = n-5$ then $v < n-5$ and we are in the same situation as $I_{n-2}^V$. 

\item A Grassmann necklace term $I_i^V$ can only contribute a path beginning at $c$ when $c \not\in I_i^V$. Further, since $c$ will always be maximal in $I_1^V\setminus I_i^V$ by point 1, a path from $c$ must end at the minimal element of $I_i^V \setminus I_1^V$. From equation \eqref{eq necklace}, the only way we can have $c \not\in I_i^V$ and $n-1$ minimal in $I_i^V \setminus I_1^V$ is if $c = n-4$ and $i = n-1$.

Similarly, the only terms that can contribute an $(n-5) \rightarrow n$ plus are those with $n-5 \in I_i^V$, i.e. when $i \in \{n-2,n-3,n-4\}$. However, $n-5$ is always the largest or second-largest element in $I_1\setminus I_i^V$, while if $n \in I_i^V \setminus I_1^V$ for $i \in \{n-2,n-3,n-4\}$ then $|I_i^V \setminus I_1^V| \geq 3$ and so $n$ cannot be the smallest or second-smallest in this set. Thus there can be no path $(n-5) \rightarrow n$.
\end{enumerate}
\end{proof}

\begin{lem}\label{lem other k}
  Let $V$ and $W$ be as in Lemma~\ref{lem I}, and suppose that if $n-2\in I_{n-2}^{V}$ then $n-1\in I_{n-2}^{V}$ also. Then for each $k$ in the range $1<k<n-2$, $I_k^{W}$ contributes the same plusses as $I_{k}^{V}$, except that if $I_{k}^{V}$ contributed a plus in the $n-3$ column of the Le diagram of $V$ then this plus is shifted one square left in the Le diagram of $W$.
\end{lem}

\begin{proof}
Recall that $I_1^W = I_1^V\cup \{n-3\}$ and that $n-2$, $n-1$, and $n$ are not in $I_1^W$.

  If $n-3\not\in I_{k}^{V}$ then by Lemma~\ref{lem I} we have $I_{k}^{W} = I_k^{V}\cup \{n-3\}$.  Since $n-3$ is the largest element of $I_1^{W}$ this transformation leaves the disjoint paths unchanged and so the plusses carry over from $V$ to $W$ directly.

  If $n-3\in I_{k}^{V}$ then $I_k^V$ must contribute a plus in column $n-3$ of the Le diagram of $V$, and by Lemma~\ref{lem I} we have $I_{k}^{W} = I_k^{V}\cup \{n-2\}$. If $n-2$ supports no propagators in $V$ then certainly no plusses appear in column $n-2$ of the Le diagram of $V$. If $n-2$ supports at least one propagator in $V$ then $n-2 \in I_{n-2}^V$ and so by hypothesis $n-1 \in I_{n-2}^V$ as well. By Lemma~\ref{lem n-2 bad} item 3, the only necklace element of $V$ containing $n-2$ is $I_{n-2}^{V}$ and further with item 4, the plusses in column $n-2$ from $I_{n-2}^V$ and $I_{n-2}^W$ are different.

  Since $n-3 \in I_k^V \setminus I_1^V$, we must have a path in the Le diagram of $V$ from some vertical edge $i$ to the bottom edge $n-3$.  In the Le diagram of $W$, the index $n-3$ labels a vertical edge with no path from $I_k^W$ starting at it (since $n-3$ belongs to both $I_1^W$ and $I_k^W$), and there must be a path leading to $n-2$ since $n-2 \in I_k^W \setminus I_1^W$.  By the previous paragraph no other path from $I_k^V$ could end at $n-2$, and since the paths cannot cross we conclude that the $i\rightarrow (n-3)$ path in the Le diagram of $V$ must become a $i\rightarrow (n-2)$ path in the Le diagram of $W$. All other paths are unchanged.
\end{proof}

We are now ready to prove the main result of this section, showing that the dimension of the positroid cell associated to each Wilson loop diagram is $3|\cP|$. 

\begin{thm}\label{thm dim}
Let $W = (\cP, [n])$ be an admissible Wilson loop diagram.  The number of plusses in the Le diagram of $W$ is three times the number of propagators.
\end{thm}

\begin{proof}
  The proof is by induction on the number of propagators.

  First note that a Wilson loop diagram $W$ with one propagator supported on vertices \[{i<i+1<j<j+1}\] has a Le diagram consisting of a single row with $n-i$ boxes. The columns are labeled from left to right by $n, \ldots, i+1$. By Algorithm~\ref{alg:GN to Le} there are plusses in the $i+1$, $j$, and $j+1$ positions. 

  Now consider Wilson loop diagrams with $k>1$ propagators.  By Lemma~\ref{lem uncovered} it suffices to prove the result for weakly admissible Wilson loop diagrams with $k$ propagators and no non-supporting vertices.  By Lemma~\ref{lem dihedral} it suffices to prove the result for at least one Wilson loop diagram from each dihedral orbit. In particular, we can restrict our attention to diagrams $W$ which have a propagator $p$ with the properties in Lemma~\ref{lem good p}.  


Thus $W$ has one of the two configurations depicted in Figure~\ref{fig special p}, and by Remark~\ref{rem n-2 case} we can further assume that if $n-2 \in I_{n-2}^V$ then $n-1 \in I_{n-1}^V$ as well. Let $V$ be $W$ with $p$ removed. 

  From Lemma~\ref{lem shape} we know how the shapes of the Le diagrams of $V$ and $W$ are related; let $\mathcal{A}$ and $\mathcal{B}$ be as described in that lemma and subsequent discussion.  

  Lemmas~\ref{lem n and n-1}, \ref{lem n-2 good}, and \ref{lem n-2 bad} tell us that the three boxes of the bottom row of the Le diagram of $W$ each have a plus, i.e. there are plusses in the positions $(n-3) \rightarrow (n-2)$, $(n-3) \rightarrow (n-1)$ and $(n-3) \rightarrow n$. We will view these as the three ``new'' plusses required by the induction step and use Lemmas~\ref{lem n and n-1} through \ref{lem other k} to construct a pairing between the plusses of the Le diagram of $V$ and the plusses of the Le diagram of $W$ that are not in the bottom row. See Figure~\ref{fig:le diagrams bijection}.

 The pairing is defined as follows:
  \begin{itemize}
  \item The plusses in the $\mathcal{B}$ region of both Le diagrams are identical, so each plus in this region is paired with itself.
  \item Every square containing a plus in the leftmost two columns (columns $n$ and $(n-1)$) of $\mathcal{A}$ for $V$ also contains a plus in $W$, so these plusses are also paired with themselves.
  \item The plusses in column $n-3$ for $V$ are paired with the corresponding plusses in column $n-2$ for $W$.
  \item If there is a plus in column $n-2$ of $\mathcal{A}$ in $V$ then Lemma~\ref{lem n-2 bad} applies, so there is exactly one such plus in this column. If $n-4 \not\in I_1^V$ then pair this plus with the $(n-5)\rightarrow (n-1)$ plus in $W$, whereas if $n-4 \in I_1^V$ then pair it with the $(n-5) \rightarrow n$ plus in $W$.

  \end{itemize}
The first two points follow from Lemmas \ref{lem n and n-1} through \ref{lem n-2 bad}. The third point follows from Lemma~\ref{lem other k}, which tells us that column $n-3$ of $V$ is identical to column $n-2$ of $W$ (excluding the last row).

Finally, if $n-2 \not\in I_{n-2}^V$ by Lemma~\ref{lem n-2 good} there are no plusses in column $n-2$ of the diagram of $V$, and all plusses have been accounted for. On the other hand, if $n-2 \in I_{n-2}^V$ then we are in the setting of Lemma~\ref{lem n-2 bad}, and there is exactly one plus in column $n-2$ of the diagram of $V$. This plus is paired with the one plus in $W$ yet unaccounted for, namely the one in position $(n-5) \rightarrow n$ or $(n-5) \rightarrow (n-1)$, as described in points 4 and 6 of that lemma.

Therefore, the Le diagram of $W$ contains $3(k-1)$ plusses in bijection with the plusses from the Le diagram of $V$ and 3 new plusses in the bottom row, yielding $3k$ in total. Applying induction completes the proof.
\end{proof}

\begin{figure}
\[
\begin{tikzpicture}[scale = 0.85]
\node at (2.5,5.5) {\Large $V$};
\node at (12.5,5.5) {\Large $W$};
\node at (-2,2.5) {\large $n-4 \in I_1^V$};
\node at (-2,-4.5) {\large $n-4 \not\in I_1^V$};
\node at (7.5,2.5) {\Large $\longrightarrow$};
\node at (8,-4.5) {\Large $\longrightarrow$};
\begin{scope}
\fill[pattern color = gray,pattern = north east lines] (0,4.5) -- (0,2) -- (2,2) -- (2,4.5);
\fill[pattern color = gray,pattern = north west lines] (3,4.5) -- (3,0) -- (4,0) -- (4,4.5);
\draw (0,0) grid (4,2);
\foreach \x in {0,1,2,3,4} {
  \draw (\x,2) -- (\x,4);
  \draw[smalldash] (\x,4) -- (\x,4.5);
  };
\foreach \c in {(0.5,0.5),(1.5,0.5),(1.5,1.5)} {
  \node at \c {\leplus};
}
\foreach \c in {(0.5,1.5),(2.5,1.5),(2.5,2.5),(2.5,3.3)} {
  \node at \c {\lezero};
}
\node[red] at (2.5,0.5) {\leplusbold[red]};
\node at (2.5,4.2) {$\vdots$};
\begin{scope}
\fill[pattern color = gray,pattern = horizontal lines] (4,2) .. controls (5.3,2.2) and (4.7,3) .. (5.5,3.2) .. controls (6.3,3.5) .. (6.5,4) -- (6.6,4.5) -- (4,4.5) -- (4,2);
\draw (4,2) .. controls (5.3,2.2) and (4.7,3) .. (5.5,3.2) .. controls (6.3,3.5) .. (6.5,4);
\draw[smalldash] (6.5,4) -- (6.6,4.5);
\end{scope}
\node[below] at (0.5,0) {\scriptsize $n\vphantom{1}$};
\node[below] at (1.5,0) {\scriptsize $n-1$};
\node[below] at (2.5,0) {\scriptsize $n-2$};
\node[below] at (3.5,0) {\scriptsize $n-3$};
\node[right] at (4,0.5) {\scriptsize $n-4$};
\node[right] at (4,1.5) {\scriptsize $n-5$};
\end{scope}

\begin{scope}[shift = {(9,0)}]
\fill[pattern color = gray,pattern = north east lines] (0,4.5) -- (0,2) -- (2,2) -- (2,4.5);
\fill[pattern color = gray,pattern = north west lines] (2,4.5) -- (2,0) -- (3,0) -- (3,4.5);
\draw (0,-1) grid (3,2);
\foreach \x in {0,1,2,3} {
  \draw (\x,2) -- (\x,4);
  \draw[smalldash] (\x,4) -- (\x,4.5);
  };
\foreach \c in {(0.5,0.5),(1.5,0.5),(1.5,1.5)} {
  \node at \c {\leplus};
}
\node[red] at (0.5,1.5) {\leplusbold[red]};
\foreach \c in {(0.5,-0.5),(1.5,-0.5),(2.5,-0.5)} {
  \node[blue] at \c {\leplus};
}
\begin{scope}[shift = {(-1,0)}]
\fill[pattern color = gray,pattern = horizontal lines] (4,2) .. controls (5.3,2.2) and (4.7,3) .. (5.5,3.2) .. controls (6.3,3.5) .. (6.5,4) -- (6.6,4.5) -- (4,4.5) -- (4,2);
\draw (4,2) .. controls (5.3,2.2) and (4.7,3) .. (5.5,3.2) .. controls (6.3,3.5) .. (6.5,4);
\draw[smalldash] (6.5,4) -- (6.6,4.5);
\end{scope}
\node[below] at (0.5,-1) {\scriptsize $n\vphantom{1}$};
\node[below] at (1.5,-1) {\scriptsize $n-1$};
\node[below] at (2.5,-1) {\scriptsize $n-2$};
\node[right] at (3,-0.5) {\scriptsize $n-3$};
\node[right] at (3,0.5) {\scriptsize $n-4$};
\node[right] at (3,1.5) {\scriptsize $n-5$};
\end{scope}

\begin{scope}[shift = {(0,-6)}]
\fill[pattern color = gray,pattern = north east lines] (0,3.5) -- (0,1) -- (2,1) -- (2,3.5);
\fill[pattern color = gray,pattern = north west lines] (3,3.5) -- (3,0) -- (4,0) -- (4,3.5);
\draw (0,0) grid (4,1);
\foreach \x in {0,1,2,3,4} {
  \draw (\x,1) -- (\x,3);
    \draw[smalldash] (\x,3) -- (\x,3.5);
  };
\foreach \c in {(0.5,0.5),(1.5,0.5),(2.5,1.5),(2.5,2.3)} {
  \node at \c {\lezero};
}
\node[red] at (2.5,0.5) {\leplusbold[red]};
\node at (2.5,3.2) {$\vdots$};
\begin{scope}[shift = {(1,-1)}]
\fill[pattern color = gray,pattern = horizontal lines] (3,1) -- (4,1) -- (4,2) .. controls (5.3,2.2) and (4.7,3) .. (5.5,3.2) .. controls (6.3,3.5) .. (6.5,4)-- (6.6,4.5) -- (3,4.5) -- (3,1);
\draw (3,1) -- (4,1) -- (4,2) .. controls (5.3,2.2) and (4.7,3) .. (5.5,3.2) .. controls (6.3,3.5) .. (6.5,4);
\draw[smalldash] (6.5,4) -- (6.6,4.5);
\end{scope}

\node[below] at (0.5,0) {\scriptsize $n\vphantom{1}$};
\node[below] at (1.5,0) {\scriptsize $n-1$};
\node[below] at (2.5,0) {\scriptsize $n-2$};
\node[below] at (3.5,0) {\scriptsize $n-3$};
\node[below] at (4.5,0) {\scriptsize $n-4$};
\node[right] at (5,0.5) {\scriptsize $n-5$};
\end{scope}

\begin{scope}[shift = {(9,-6)}]
\fill[pattern color = gray,pattern = north east lines] (0,3.5) -- (0,1) -- (2,1) -- (2,3.5);
\fill[pattern color = gray,pattern = north west lines] (2,3.5) -- (2,0) -- (3,0) -- (3,3.5);
\draw (0,-1) grid (3,1);
\foreach \x in {0,1,2,3} {
  \draw (\x,1) -- (\x,3);
  \draw[smalldash] (\x,3) -- (\x,3.5);
  };
\node at (0.5,0.5) {\lezero};
\node[red] at (1.5,0.5) {\leplusbold[red]};
\foreach \c in {(0.5,-0.5),(1.5,-0.5),(2.5,-0.5)} {
  \node[blue] at \c {\leplus};
}
\begin{scope}[shift = {(0,-1)}]
\fill[pattern color = gray,pattern = horizontal lines] (3,1) -- (4,1) -- (4,2) .. controls (5.3,2.2) and (4.7,3) .. (5.5,3.2) .. controls (6.3,3.5) .. (6.5,4)-- (6.6,4.5) -- (3,4.5) -- (3,1);
\draw (3,1) -- (4,1) -- (4,2) .. controls (5.3,2.2) and (4.7,3) .. (5.5,3.2) .. controls (6.3,3.5) .. (6.5,4);
\draw[smalldash] (6.5,4) -- (6.6,4.5);
\end{scope}
\node[below] at (0.5,-1) {\scriptsize $n\vphantom{1}$};
\node[below] at (1.5,-1) {\scriptsize $n-1$};
\node[below] at (2.5,-1) {\scriptsize $n-2$};
\node at (3.5,-0.7) {\scriptsize $n-3$};
\node at (3.7,-0.2) {\scriptsize $n-4$};
\node[right] at (4,0.5) {\scriptsize $n-5$};
\end{scope}
\end{tikzpicture}
\]
\caption{The Le diagrams for $V$ and $W$ in the setting of Theorem~\ref{thm dim}. The shaded regions indicate where the plusses match up perfectly from $V$ to $W$, the bold (red) plusses exist if and only if $n-2 \in I_{n-2}^V$, and the (blue) plusses in the bottom row of the Le diagrams of $W$ diagrams are the three new plusses coming from the induction step.}\label{fig:le diagrams bijection}
\end{figure}
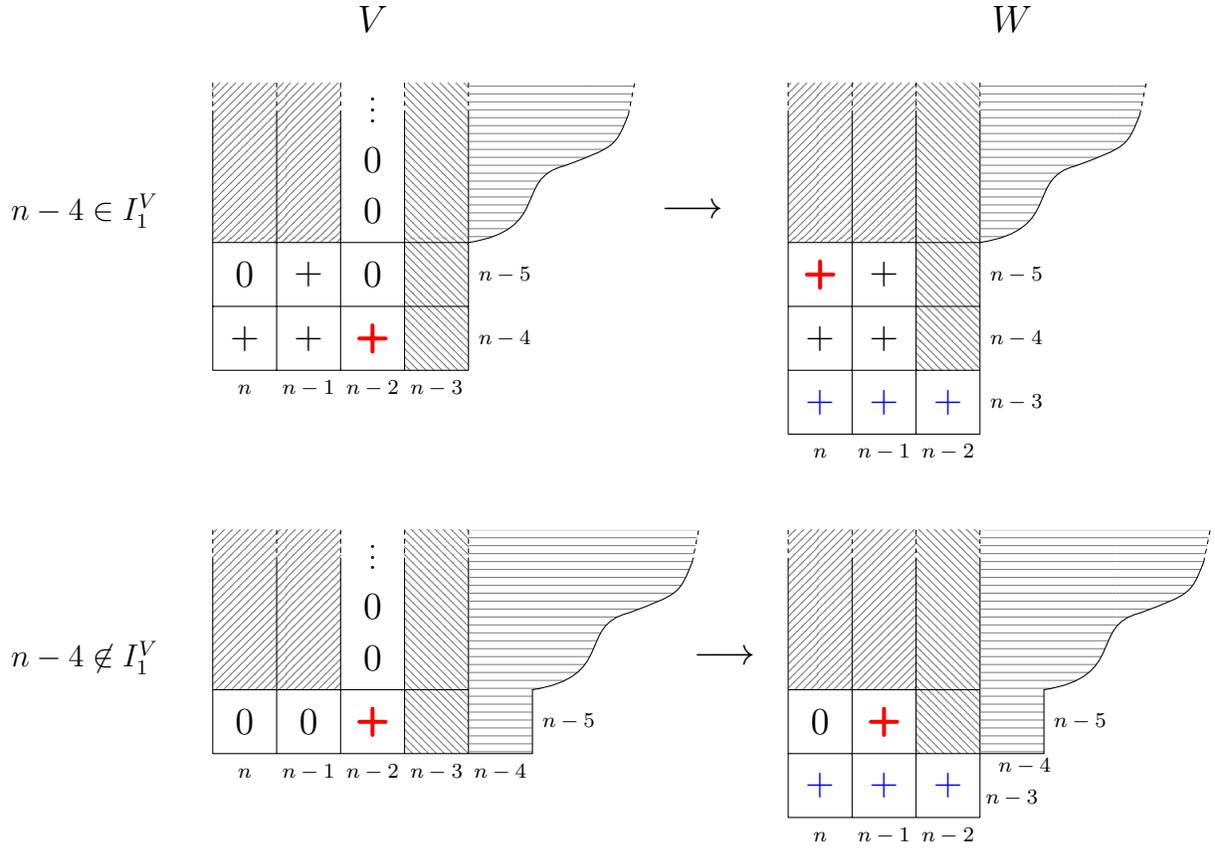

\section{Poles of Wilson loop Integrals}\label{sec poles} 

Recall from Section~\ref{section WLD background} that each Wilson loop diagram defines a functional $I(W)$, which associates a volume to the positroid cell $\Sigma(W)$. In this section we investigate the poles of the integrand of $I(W)$. 

Since SYM $N=4$ is a finite theory {\cite{Alday:2007hr, ParkeTaylorids}}, the scattering amplitudes (and thus the holomorphic Wilson loops $\,\mathfrak{W}_{k,n}$) are finite. That is, for a fixed set of external particles $\cZ$ we must have
\bas \sum_{W \text{ s.t. } |\cP| = k} I(W)(\cZ) <\infty \;.\eas
This holds if all poles of the integrands cancel. As a step towards this, in this section we prove some results about the structure of the denominator $R(W)$ of the integrand, as defined in Definition~\ref{def R(W)}.

The results of Section~\ref{sec GN algorithm} allow us to relate the configuration of propagators in a Wilson loop diagram $W$ to minors of $C(W)$ that contribute to $R(W)$. The main result of this section is Theorem~\ref{thm denom}, which expresses the denominator $R(W)$ in terms of the Grassmann necklace of $W$. This simplifies the computation of $R(W)$ and allows us to directly relate the poles of the integrand to the combinatorics of the diagram.

For the duration of this section, fix an arbitrary ordering on the propagators of $W$. This imposes an ordering on the rows of $C(W)$.

We first give an algorithm which extracts the required minors of $C(W)$ from the Grassmann necklace.

\begin{algorithm}\label{alg WLD to denom via GN}
Let $W = (\cP, [n])$ be a Wilson loop diagram and $C(W)$ be its associated matrix. 
\begin{itemize}
  \item For each $i \in [n]$, we construct a factor $r_i$ as follows:
    \begin{itemize}
      \item Let $S_i = \{p \in \cP \ | \ I_{i-1}(p) \neq I_i(p)\}$. (By convention, set $I_{0} = I_n$.)
      \item Let $r_i$ be the determinant of the $|S_i|\times |S_i|$ minor of $C(W)$ with rows indexed by $S_i$ and columns indexed by $I_i(S_i)$.  If $S_i = \emptyset$, set $r_i = 1$. 
    \end{itemize}
  \item Define $R = \prod_{i=1}^n r_i$.
\end{itemize}
\end{algorithm}

Each $r_i$ is the determinant of a minor determined by certain propagators of $W$. In particular, if a propagator $p$ contributes one value to the $i^{th}$ element of the Grassmann necklace and a different value to the $(i-1)^{th}$ element, then $p$ is one of the propagators contributing to the minor. Specifically, the row associated to $p$ and the column associated to the value that $p$ contributes to the $i^{th}$ Grassmann Necklace element appear in the minor.
Since each $r_i$ is simply the determinant of a particular submatrix of $C(W)$ and an ordering on the propagators has been fixed, the sign of each $r_i$ is well defined.

We also introduce the notation $\Delta_{I_i}$ to represent the determinant of the $k\times k$ minor of $C(W)$ with columns indexed by $I_i$. Note that one summand of this determinant is given by $\prod_{p \in \cP}x_{p, I_i(p)}$; the other summands are (up to sign) given by other bijections between the propagator set $\cP$ and the vertex set $I_i$.

Below, we show that the polynomial $R$ associated to $W$ by Algorithm~\ref{alg WLD to denom via GN} is equal to the denominator $R(W)$ from Definition~\ref{def R(W)}, and that the ideal generated by $R$ is the radical of the ideal generated by the product $\prod_{i=1}^n \Delta_{I_i}$. First, however, we give a worked example of Algorithm~\ref{alg WLD to denom via GN}.

\begin{figure}
\begin{tabular}{rrrrr}
\begin{tikzpicture}[rotate=67.5,baseline=(current bounding box.east)]
  \begin{scope}
  \drawWLD{7}{\diagramscale}
  \drawnumbers
  \drawlabeledprop{1}{-1}{6}{0}{\footnotesize   $p$}
  \drawlabeledprop{1}{0}{5}{0}{\footnotesize   $q$}
  \drawlabeledprop{1}{1}{4}{0}{\footnotesize \; \qquad  $s$}
  \pgfmathsetmacro{\move}{\angle/\n};
  \draw[propassignment,red] (1.5*\angle + -1*\move:\radius) to[bend left = \arrowangle] (\angle*1:\radius); 
  \draw[propassignment,red] (1.5*\angle:\radius) to[bend right = \arrowangle] (\angle*2:\radius); 
  \draw[propassignment,red] (4.5*\angle:\radius) to[bend left = \arrowangle] (\angle*4:\radius); 
  \node at (\angle*1.5:\radius*2) {$I_1$};
  \node[align = center] at (4*\angle:\radius*1.7) {$I_1(p) = 1, \; I_1(q) = 2,$ \\[5pt] $I_1(s) = 4$.};
    \end{scope}
  \end{tikzpicture} 
  & \quad \quad  &
\begin{tikzpicture}[rotate=67.5,baseline=(current bounding box.east)]
  \begin{scope}
  \drawWLD{7}{\diagramscale}
  \drawnumbers
  \drawlabeledprop{1}{-1}{6}{0}{\footnotesize   $p$}
  \drawlabeledprop{1}{0}{5}{0}{\footnotesize   $q$}
  \drawlabeledprop{1}{1}{4}{0}{\footnotesize \; \qquad  $s$}
  \pgfmathsetmacro{\move}{\angle/\n};
  \draw[propassignment,red] (1.5*\angle + -1*\move:\radius) to[bend right = \arrowangle] (\angle*2:\radius); 
  \draw[propassignment,red] (5.5*\angle:\radius) to[bend left = \arrowangle] (\angle*5:\radius); 
  \draw[propassignment,red] (4.5*\angle:\radius) to[bend left = \arrowangle] (\angle*4:\radius); 
  \node at (\angle*1.5:\radius*2) {$I_2$};
  \node[align = center] at (4*\angle:\radius*1.7) {$I_2(p) = 2,  \; I_2(s) = 4,$ \\[5pt] $I_2(q) = 5.$};
  \end{scope} 
\end{tikzpicture}
& \quad & 
\begin{tikzpicture}[rotate=67.5,baseline=(current bounding box.east)]
  \begin{scope}
  \drawWLD{7}{\diagramscale}
  \drawnumbers
  \drawlabeledprop{1}{-1}{6}{0}{\footnotesize   $p$}
  \drawlabeledprop{1}{0}{5}{0}{\footnotesize   $q$}
  \drawlabeledprop{1}{1}{4}{0}{\footnotesize \; \qquad  $s$}
  \pgfmathsetmacro{\move}{\angle/\n};
  \draw[propassignment,red] (6.5*\angle:\radius) to[bend left = \arrowangle] (\angle*6:\radius); 
  \draw[propassignment,red] (5.5*\angle:\radius) to[bend left = \arrowangle] (\angle*5:\radius); 
  \draw[propassignment,red] (4.5*\angle:\radius) to[bend left = \arrowangle] (\angle*4:\radius); 
  \node at (\angle*1.5:\radius*2) {$I_3 = I_4$};
  \node[align = center] at (4*\angle:\radius*1.7) {$I_3(s) = 4,\;  I_3(q) = 5,$ \\[5pt] $I_3(p) = 6.$ };
  \end{scope} 
\end{tikzpicture}  
\\ 
 & & & &  \\
 \begin{tikzpicture}[rotate=67.5,baseline=(current bounding box.east)]
  \begin{scope}
  \drawWLD{7}{\diagramscale}
  \drawnumbers
  \drawlabeledprop{1}{-1}{6}{0}{\footnotesize   $p$}
  \drawlabeledprop{1}{0}{5}{0}{\footnotesize   $q$}
  \drawlabeledprop{1}{1}{4}{0}{\footnotesize \; \qquad  $s$}
  \pgfmathsetmacro{\move}{\angle/\n};
  \draw[propassignment,red] (6.5*\angle:\radius) to[bend right = \arrowangle] (\angle*7:\radius); 
  \draw[propassignment,red] (5.5*\angle:\radius) to[bend right = \arrowangle] (\angle*6:\radius); 
  \draw[propassignment,red] (4.5*\angle:\radius) to[bend right = \arrowangle] (\angle*5:\radius); 
  \node at (\angle*1.5:\radius*2) {$I_5$};
  \node[align = center] at (4*\angle:\radius*1.7) {$I_5(s) = 5,\;  I_5(q) = 6,$  \\[5pt] $I_5(p) = 7.$ };
  \end{scope} 
\end{tikzpicture}
& \quad & 
\begin{tikzpicture}[rotate=67.5,baseline=(current bounding box.east)]
  \begin{scope}
  \drawWLD{7}{\diagramscale}
  \drawnumbers
  \drawlabeledprop{1}{-1}{6}{0}{\footnotesize   $p$}
  \drawlabeledprop{1}{0}{5}{0}{\footnotesize   $q$}
  \drawlabeledprop{1}{1}{4}{0}{\footnotesize \; \qquad  $s$}
  \pgfmathsetmacro{\move}{\angle/\n};
  \draw[propassignment,red] (6.5*\angle:\radius) to[bend right = \arrowangle] (\angle*7:\radius); 
  \draw[propassignment,red] (5.5*\angle:\radius) to[bend right = \arrowangle] (\angle*6:\radius); 
  \draw[propassignment,red] (1.5*\angle + 1*\move:\radius) to[bend left = \arrowangle] (\angle*1:\radius); 
  \node at (\angle*1.5:\radius*2) {$I_6$};
  \node[align = center] at (4*\angle:\radius*1.7) {$I_6(q) = 6, \; I_6(p) = 7, $ \\[5pt] $I_6(s) = 1. $ };  
  \end{scope}
  \end{tikzpicture}  
& \quad &
\begin{tikzpicture}[rotate=67.5,baseline=(current bounding box.east)]
  \begin{scope}
  \drawWLD{7}{\diagramscale}
  \drawnumbers
  \drawlabeledprop{1}{-1}{6}{0}{\footnotesize   $p$}
  \drawlabeledprop{1}{0}{5}{0}{\footnotesize   $q$}
  \drawlabeledprop{1}{1}{4}{0}{\footnotesize \; \qquad  $s$}
  \pgfmathsetmacro{\move}{\angle/\n};
  \draw[propassignment,red] (6.5*\angle:\radius) to[bend right = \arrowangle] (\angle*7:\radius); 
  \draw[propassignment,red] (1.5*\angle:\radius) to[bend left = \arrowangle] (\angle*1:\radius); 
  \draw[propassignment,red] (1.5*\angle + \move:\radius) to[bend right = \arrowangle] (\angle*2:\radius); 
  \node at (\angle*1.5:\radius*2) {$I_7$};
  \node[align = center] at (4*\angle:\radius*1.7) {$I_7(p) = 1,  \; I_7(q) = 1, $ \\[5pt] $I_7(s) = 2.$};
  \end{scope} 
\end{tikzpicture}
\\
\end{tabular}
\caption{Example Wilson loop diagram for illustrating Algorithm~\ref{alg WLD to denom via GN} and bijections between propagators and vertices for each Grassmann necklace element.}\label{fig R eg}
\end{figure}
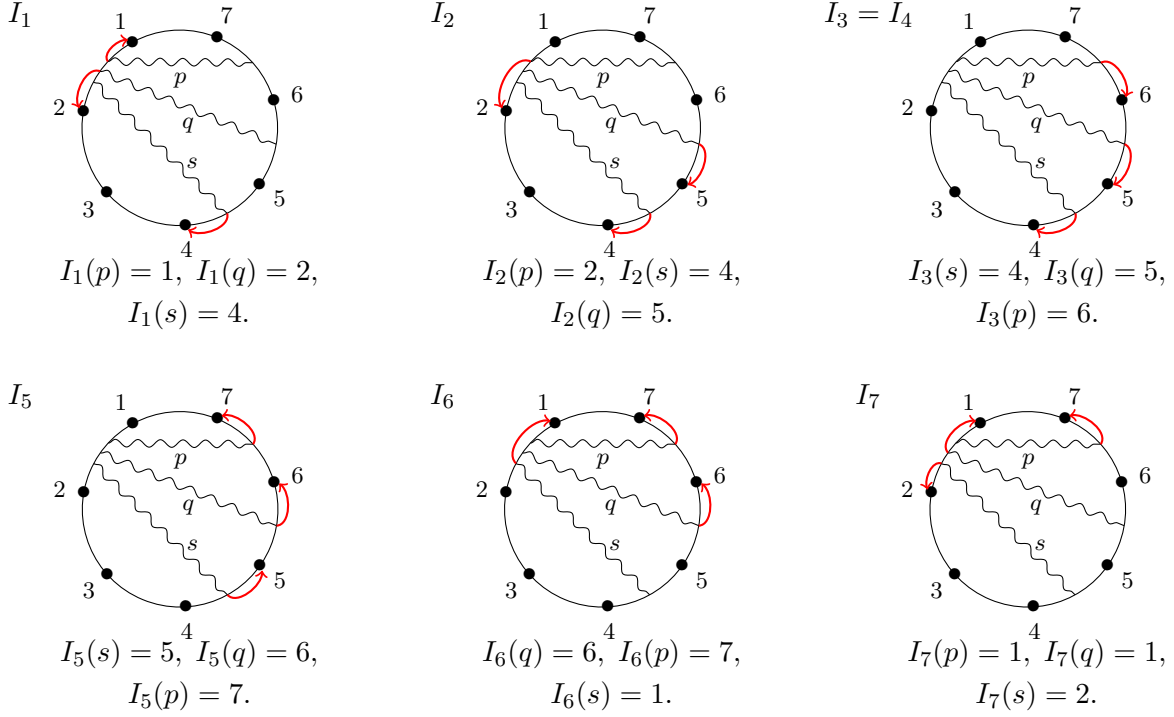

\begin{eg}
Consider the Wilson loop diagram in Figure~\ref{fig R eg}. Assigning propagators $p$, $q$, $r$ to rows $1,2,3$ respectively, we obtain the matrix
\[
C(W) = \begin{bmatrix} a & b & 0 & 0 & 0 & c & d \\ e & f & 0 & 0 & g & h & 0 \\ i & j & 0 & k & l & 0 & 0 \end{bmatrix}.
\]
The Grassmann necklace of this diagram is
\begin{eqnarray*}\Big(I_1 = \{1,2,4\},\ I_2 = \{2,4,5\},\ I_3 = \{4,5,6\},\ I_4=\{4,5,6\},\\ I_5=\{5,6,7\},\ I_6 = \{6,7,1\},\ I_7=\{7,1,2\}\Big). \end{eqnarray*}  
Figure~\ref{fig R eg} indicates which vertex is contributed to $I_i$ by each propagator, for each $i \in [1,7]$.  

From $I_1$ to $I_2$, the propagators $p$ and $q$ change which vertices they contribute to; specifically, we have $I_1(p) = 1$ and $I_1(q)= 2$, while $I_2(p) = 2$ and $I_2(q)= 5$. On the other hand, $r$ contributes  4 to both $I_1$ and $I_2$. We therefore have $S_2 = \{p,q\}$. The associated minors $\Delta_{I_2}$ and $r_2$ are therefore
\[
\Delta_{I_2}=\det\begin{bmatrix} b & 0 & 0 \\ f & 0 & g \\ j & k & l \end{bmatrix} = kgb, \qquad r_2 = \det\begin{bmatrix} b & 0 \\ f & g \end{bmatrix} = gb.
\] 
Continuing likewise, we get $r_3 = c$, $r_4=1$ (since $I_4 = I_3$), $r_5 = lhd$, and $r_6 = i$.

At $I_7$ we see for the first time an irreducible quadratic factor.  We have $S_7 = \{q,s\}$ with ${I_7(S_7) = \{ 1,2\}}$, $\Delta_{I_7} = d(ej-fi)$ and $r_7 = ej-fi$. 

The irreducible quadratic factor corresponds to the fact that $q$ and $s$ share an edge {\em and} contribute both endpoints of that edge to $I_7$; see Proposition~\ref{prop alg gives rad} below.  

Finally, we have $r_1 = (af-be)k$.  Putting the previous calculations together, we obtain
\[
R = (af-be)kgbclhdi(ej-fi) \;,
\]
which is square-free and contains all factors of $\prod_{i=1}^{n}\Delta_{I_i}$. If one were to construct the denominator $R(W)$ associated to this Wilson loop diagram as per Definition~\ref{def R(W)}, we would find that we have $R(W) = R$.
\end{eg}

We are now ready to begin proving the main result of this section. We begin with a proposition stating several facts about $\Delta_{I_i}$ and their relationship to $r_i$.

\begin{prop}\label{prop alg gives rad}
  With notation as in Algorithm~\ref{alg WLD to denom via GN} the following statements hold.
  \begin{enumerate}
    \item Each $\Delta_{I_i}$ is homogeneous, as is each $r_i$.
    \item Each $\Delta_{I_i}$ splits into linear and quadratic factors.  All linear factors of  $\Delta_{I_i}$ are single variables and all irreducible quadratic factors are $2\times 2$ determinants of single variables.
    \item Quadratic factors in $\Delta_{I_i}$ arise precisely when propagators $p$ and $q$ are supported on a common edge $a$ with $I_i(p)=a$ and $I_i(q)=a+1$.
    \item The factor $r_i$ divides $\Delta_{I_i}$.
    \item The ideal generated by $R$ is the radical of the ideal generated by $\prod_{i=1}^{n}\Delta_{I_i}$.
  \end{enumerate}
\end{prop}

\begin{proof}
\begin{enumerate}
    \item The nonzero entries of $C(W)$ are independent indeterminates and so every $i\times i$ minor of $C(W)$ is either $0$ or homogeneous of degree $i$.  Thus each $\Delta_{I_i}$ and each $r_i$ is homogeneous.  
    \item Viewing the determinant as a sum over permutations, and combining this with Corollary~\ref{lem basis as perm}, we see that $\Delta_{I_i}$ is a sum over bijections between $I_i$ and $\mathcal{P}$.  The nonzero terms in this sum are precisely those bijections such that each propagator is associated to one of its supporting vertices in $I_i$, since only those locations in $C(W)$ are nonzero.  Since the nonzero entries of $C(W)$ are independent there can be no cancellation between terms in this expansion. 

Suppose $\Delta_{I_i}$ has an irreducible factor $f$.  Let $\mathcal{P}'$ be the set of propagators which contribute a variable to $f$ and let $J$ be the set of vertices which contribute a variable to $f$.

The first claim is that the minor of $C(W)$ associated to $\mathcal{P}'$ and $J$ is precisely $f$.

{\em Proof of claim}: By the structure of determinants we know that $\Delta_{I_i} = fg$, where $g$ involves only variables associated to propagators not in $\mathcal{P}'$ and associated to vertices not in $J$.

Expanding out $fg$ yields a signed sum of monomials. Since there is no cancellation between terms, this means that the full expansion over permutations of $\Delta_{I_i}$ contains no other nonzero terms and hence no other variables.  Therefore, $\Delta_{I_i}$ is equal to the determinant of the matrix obtained by taking the submatrix of $C(W)$ with columns indexed by $I_i$ and setting any variables not appearing in $\Delta_{I_i}$ to $0$.  This new matrix is, up to permutations of rows and columns, a block matrix with one block for $\mathcal{P}'$ and $J$ and the other block for the complements.  Thus its determinant, and hence also $\Delta_{I_i}$, is the product of the minors for these two blocks.  By considering which variables appear, these two factors must also be $f$ and $g$, and so in particular $f$ is the minor of $C(W)$ associated to $\mathcal{P}'$ and $J$.  This proves the claim.

A consequence of this claim is that every linear factor of $\Delta_{I_i}$ is a $1\times 1$ minor of $C(W)$, hence is a single variable, and every irreducible quadratic factor of $\Delta_{I_i}$ is a $2\times 2$ minor of $C(W)$, hence is a $2\times 2$ determinant of single variables.

All that remains is to prove that $\Delta_{I_i}$ has no irreducible factors of degree 3 or more.  Suppose, by way of contradiction that $f$ is a factor of $\Delta_{I_i}$ of degree $\geq 3$. Without loss of generality\footnote{This is without loss of generality for the following reason.  By removing the propagators which come before those contributing to $f$ (in the order given by $I_i$) and changing $i$ to be the first vertex which contributes to $f$, we obtain an admissible diagram.  In this diagram, $f$ still divides $\Delta_{I_i}$ and we also have $i\in I_i$ and that 
$i$ contributes to $f$.  Thus this diagram can be used in place of $W$.} we may assume $i\in I_i$ and that $i$ contributes to $f$.
Finally, we can suppose that $W$ is minimal with respect to the number of propagators amongst diagrams with an irreducible factor of degree $\geq 3$. 

Let $p$ be the propagator such that $I_i(p) = i$. There are two cases to consider: when $p$ is supported on the edge $i-1$ or the edge $i$.  These are illustrated in Figure~\ref{fig no big factors}.

\begin{figure}
\[
\begin{tikzpicture}[baseline=(current bounding box.east),rotate=-20,scale = 0.95]
\begin{scope}
  \drawWLDfragment[8]{3}{0.4} 
  \newnode[left]{1}{i-1}
  \newnode[below]{2.5}{i}
  \newnode[below]{7}{m}
  \begin{scope}
    \clipcenterarc(0,0)(\startpoint-10:\endpoint+10:\radius)
    \newpropbend{1.9}{7.2}{60}
    \draw[smallpropagator,gray] (\zero+7.8*\step:\radius*1.1) to[bend left = 55] (\zero + 10*\step:\radius*1.1);
    \node at (\zero + 1.5*\step:\radius*0.5) {\footnotesize $p$};
  \end{scope}
  \node[align = center,black] at (\zero + \step*4.5:\radius*0.65) {\em \footnotesize propagators \\ \footnotesize \em inside $p$};
  \draw[propassignment,gray] (\zero + 1.81*\step:\radius) to[bend right = \arrowangle] (\zero + 2.5*\step:\radius); 
  \draw[propassignment,gray] (\zero + 7.9*\step:\radius) to[bend left = \arrowangle] (\zero + 7*\step:\radius); 
  \node[align = center,black] at (\zero + 6.8*\step:1.5*\radius) {\em \footnotesize If $m \in I_i$ then it was \\ \em \footnotesize contributed by a propagator \\ \em \footnotesize outside $p$};
  \node[align = center,black] at (110:\radius*0.5) {\em Case 1: $p = (i-1,m)$};
\end{scope}
\end{tikzpicture}
\qquad  
\begin{tikzpicture}[baseline=(current bounding box.east),rotate=-20,scale = 0.95]
  \begin{scope}
  \drawWLDfragment[8]{3}{0.4}  
  \newnode[left]{1}{i}
  \newnode[below]{2.5}{i+1\ \ \ }
  \newnode[below]{7}{j}
  \centerarc[red,line width = 2pt,line cap = round](0,0)(\zero + 1.61*\step:\zero+1.88*\step:\radius);
  \draw[propassignment,gray] (\zero + 1.55*\step:\radius) to[bend left = \arrowangle] (\zero + 1*\step:\radius); 
  \draw[propassignment,gray] (\zero + 1.93*\step:\radius) to[bend right = \arrowangle] (\zero + 2.5*\step:\radius); 
  \draw[propassignment,gray] (\zero + 7.7*\step:\radius) to[bend left = \arrowangle] (\zero + 7*\step:\radius); 
  \begin{scope}
    \clipcenterarc(0,0)(\startpoint-10:\endpoint+10:\radius)
    \newpropbend{1.9}{7.2}{40} 
    \draw[smallpropagator] (\zero+1.6*\step:\radius*1.1) to[bend right = 40] (\zero  - 3*\step:\radius*1.1); 
    \draw[smallpropagator,gray] (\zero+7.6*\step:\radius*1.1) to[bend left = 55] (\zero + 10*\step:\radius*1.1);
    \node at (\zero + 5*\step:\radius*0.32) {\footnotesize $q$};
    \node at (\zero + 1.1*\step:\radius*0.6) {\footnotesize $p$};
  \end{scope}
  \node[align = center,black] at (\zero + \step*4.5:\radius*0.72) {\em \footnotesize propagators \\ \footnotesize \em inside $q$};
  \node[align = center] at (\zero + 1.8*\step:\radius*1.5) {\footnotesize \em No \\ \footnotesize \em propagators \\ \footnotesize \em between \\ \footnotesize \em $p$ and $q$ \\ \footnotesize \em here};
  \node[align = center,black] at (\zero + 6.8*\step:1.5*\radius) {\em \footnotesize If $j \in I_i$ then it was \\ \em \footnotesize contributed by a propagator \\ \em \footnotesize outside $q$};
    \node[align = center,black] at (110:\radius*0.5) {\em Case 2: $p = (i,m)$};

\end{scope}
\end{tikzpicture}
\]
  \caption{The two cases in the proof that no factors of $\Delta_{I_i}$ have degree 3 or more.}\label{fig no big factors}
\end{figure}

\textbf{Case 1}: Suppose $p = (i-1, m)$. Then $p$ has one end on the $i-1^{th}$ edge, and $I_{i+1}(p) = m$ by Remark~\ref{rmk cyclic}.

Let $T = \cP_{in}(p)$ be the set of propagators inside $p$ in the $<_i$ order.
By Remark \ref{rmk algorithm locally same}, $I_i$ and $I_{i+1}$ can only differ once $p$ contributes to $I_{i+1}$, so $I_i(q) = I_{i+1}(q)$ for each $q \in T \backslash \{p\}$. Thus every propagator inside $p$ contributes before $m$ and if a propagator contributes $m$ in $I_i$ then it must lie outside $p$.

If neither $m$ nor $m+1$ appear in $I_i$ then by Corollary~\ref{no coloops} we have $V(p) \cap I_i = \{i\}$, and so the row of $p$ in the matrix of $\Delta_{I_i}$ has only one nonzero entry; hence $\Delta_{I_i}$ has a linear factor contributed by $p$ and $i$, which is a contradiction to the fact that $i$ contributes to $f$.  So we must have at least one of $m$ and $m+1$ in $I_i$. However, all propagators in $T$ contribute to $I_{i}$ strictly before $m$ by the previous paragraph so, after permuting rows and columns as needed, the matrix giving $\Delta_{I_i}$ has the form
\[
\begin{bmatrix} A & B \\ 0 & C\end{bmatrix}
\]
where $A$ is the $|T|\times |T|$ matrix indexed by the propagators in $T$ and the vertices in $I_i(T)$. No other propagators can be supported on these vertices since all other propagators are outside $p$, and $p$ is the first propagator supported at $i$; this explains the zero block.  Therefore, $\Delta_{I_i} = \pm \det A \det C$, and both factors are nontrivial since at least one of $m$ and $m+1$ appears in $I_i$.  

Since $p\in \cP_{in}(p)$ we have $f|\det A$, and if we remove a propagator outside $p$ that contributes either $m$ or $m+1$, we get a diagram $W'$ with fewer propagators for which $\det A | \Delta_{I^{W'}_i}$ still holds, hence $f|\Delta_{I^{W'}_i}$.  This contradicts the minimality of our choices. 

\textbf{Case 2}: Suppose $p = (i, m)$, i.e. $p$ has an end on the $i^{th}$ edge. If no other propagators are supported on the vertex $i$ then the corresponding column of $C(W)$ has only one nonzero entry in it, and so $\Delta_{I_i}$ has a linear factor contributed by $p$ and $i$; as above, this is a contradiction.
Since $p$ is the first propagator in the ordering imposed by $I_i$, it is the clockwise-most propagator supported on $i$ (by Remark {\ref{clockwise ordering rem}}). Since there are multiple propagators supported on $i$, there must be a $q$ ending on edge $i$, such that $I_i(q) = i+1$, that is second clockwise-most propagators supported by $i$ (see Figure~\ref{fig no big factors}).
Write $q = (i, j)$ for the support of $q$.  The situation for $q$ is very similar to Case 1: in particular, we have $I_{i+1}(q) = j$ by Lemma {\ref{vertex cyclic int lem}} and so if $j\in I_i$ then the propagator which contributes $j$ is outside $q$.

Similarly to Case 1, let $T = \cP_{in}(q) \cup \{p\}$ in the $<_i$ order. All propagators in $T$ contribute to $I_i$ strictly before $j$ and no other propagators are supported on vertices strictly before $j$.  Thus after permuting rows and columns as needed the matrix giving $\Delta_{I_i}$ has the form 
\[
\begin{bmatrix} A & B \\ 0 & C\end{bmatrix} \;,
\]
where $A$ is the submatrix indexed by the propagators in $T$ and the vertices in $I_i(T)$. Again there are two cases to consider.  If some vertex $j$ or larger (with respect to $<_i$) belongs to $I_i$ then $B$ and $C$ are at least one column wide, and so the block form of the matrix gives a nontrivial factorization of $\Delta_{I_i}=\pm\det A\det C$.  This yields a contradiction as in Case 1: $f|\det A$ and removing a propagator not in $T$ gives a smaller diagram with a factor of degree $\geq 3$, which contradicts our minimality assumption.

On the other hand, if no vertex $\geq_i j$ is in $I_i$ then we have $\Delta_{I_i} = \det A$.  Looking in more detail into $A$, note that the only vertices that support both $p$ and $q$ and also belong to $I_i$ are $i$ and $i+1$. Hence
\[
A = \begin{bmatrix} D & 0 \\ E & F\end{bmatrix}\; ,
\]
where $D$ is the $2\times 2$ matrix indexed by the propagators $p$ and $q$ and the vertices $i$ and $i+1$.  Thus $p$ and $i$ contribute to a quadratic factor of $\Delta_{I_i}$, once again contradicting our assumptions.

All cases have now been covered and so $\Delta_{I_i}$ has only irreducible factors of degree $2$ or less.

\item Suppose propagators $p$ and $q$ are supported on a common edge  $a$, with $I_i(p)=a$ and ${I_i(q)=a+1}$.  Let $x_{p,a},x_{p,a+1},x_{q,a},x_{q,a+1}$ be the associated variables in $C(W)$. For any fixed bijection $\sigma$ from $\cP\setminus\{p,q\}$ to $I_i \setminus\{a,a+1\}$ for which each propagator is supported on its image under the bijection, we can extend $\sigma$ to a bijection from propagators to $I_i$ in two ways: either $p\mapsto a$ and $q\mapsto a+1$ or $p\mapsto a+1$ and $q\mapsto a$.  The sum of the contributions of all these bijections to $\Delta_{I_i}$ is therefore the product of $(x_{p,a}x_{q,a+1}-x_{p,a+1}x_{q,a})$  with the minor coming from $\cP\setminus\{p, q\}$ and $I_i \setminus \{a,a+1\}$.  Since there is no cancellation of terms in the expansion of $\Delta_{I_i}$, if any other terms appear then they will cause a factor which is not in the form described at the beginning of the proof of part (2).  Therefore, no such terms exist and $(x_{p,a}x_{q,a+1}-x_{p,a+1}x_{q,a})$ is a factor of $\Delta_{I_i}$. 

Now let $f$ be a quadratic factor of $\Delta_{I_i}$.  By part (2) we know that $f$ is a $2\times 2$ minor coming from two propagators, call them $p$ and $q$, and two vertices, call them $a$ and $b$, and assume that $a <_i b$.  It remains to show that $b = a+1$.  From this we can conclude that $p$ and $q$ each have one end on edge $a$. 

As in the proof of part (2), make a new admissible diagram by removing the propagators which come before those contributing to $f$ in the order imposed by $I_i$.  Without loss of generality we may assume $i=a$. The cases in the proof of part (2) show how $\Delta_{I_i}$ factors. In particular the vertices supporting the other end of $p$ either do not appear in $I_i$, or they contribute to a different factor of $\Delta_{I_i}$ than $p$ and $a$ do.  By assumption $b$ contributes to the same factor as $a$.  Therefore, $(a,b)$ is an edge defining an end of $p$, that is, $b = a+1$.

\item If $S_i = \emptyset$ then $r_i = 1$ and we immediately have $r_i | \Delta_{I_i}$; otherwise, let $p\in S_i$. 

Note that $\Delta_{I_i}$ is homogeneous linear in the variables of the row corresponding to $p$.  By part (2), either exactly one variable in the row corresponding to $p$ appears in $\Delta_{I_i}$ and this variable is a factor of $\Delta_{I_i}$, or exactly two variables from the row corresponding to $p$ appear in $\Delta_{I_i}$ and they appear as part of a quadratic factor.  In the first case let the variable be $x$. Then $x$ is a factor of $\Delta_{I_i}$. Therefore, $x$ must appear in the row of $p$ and column of $I_i(p)$, i.e. $x = x_{p, I_i(p)}$. Thus $x$ also appears in $r_i$. Since the matrix for $r_i$ is a minor of the matrix for $\Delta_{I_i}$ and every term in $\Delta_{I_i}$ involves $x$, we also have that every term in $r_i$ involves $x$ so $x$ is a factor of both $r_i$ and $\Delta_{I_i}$ and is the only variable from this row in either polynomial.

  Now suppose two variables from the row $p$ appear in a quadratic factor $f$ of $\Delta_{I_i}$.  By part (3), there is another propagator $q$ and an edge $a$ such that $f$ is the $2\times 2$ minor coming from $p, q$ and $a, a+1$.  There are two situations which can occur, either $I_i(p)=a$ and $I_i(q)=a+1$, or $I_i(q)=a$ and $I_i(p)=a+1$. 
    We show that in both cases it follows that $q \in S_i$ as well.


In both cases, since we have $I_{i-1}(p)\neq I_{i}(p)$ by assumption it follows from Lemma~\ref{vertex cyclic int lem} that ${I_{i-1}(p) <_{i-1} I_{i}(p)}$ and no other vertex supporting $p$ lies between $I_{i-1}(p)$ and $I_{i}(p)$. If $I_i(q)=a$ then $q$ lies clockwise of $p$ and $q$ contributes before $p$ in $I_i$. This means that $I_{i-1}(p)=a$ and so $I_{i-1}(q)\neq a$. Thus $q\in S_i$ and so $f$ is a factor of $r_i$.

Now consider the case where $I_i(q) = a+1$, and suppose for contradiction that $q \not\in S_i$, i.e. that $I_{i-1}(q) = a+1$. Since $I_{i-1}(p) \neq a$, there must be some other propagator $s$ with $I_{i-1}(s) = a$ (else $q$ contributes $a$ to $I_{i-1}$). This propagator cannot lie on edge $a$, because Remark~\ref{rmk cyclic} would imply that $I_i(s) \in \{a, a+1\}$, contradicting the fact that $I_i(p) = a$ and $I_i(q) = a+1$; thus $s$ has an end on edge $(a-1)$ and is inside $p$ in the $<_k$ order where $p=(k,a)$.

Say $s = (j,a-1)$ with $i-1 \leq_{i-1} k+1 \leq_{i-1} j+1$. But by Lemma~\ref{lem no fourth vertex}, if $I_{i-1}(s) = a$ then $a$ cannot be maximal in the support of $s$ with respect to $<_{i-1}$; thus we must have $i-1 = j+1$, and we are in the situation in Figure~\ref{fig part 4}.

\begin{figure}
\[
\begin{tikzpicture}[baseline=(current bounding box.east),rotate = -40]
	\begin{scope}
	\drawWLDfragment[8]{3}{0.4} 
	\newnode[left]{1}{j}
	\newnode[left]{2}{i-1}
	\newnode[below]{5}{a-1\quad}
	\newnode[below]{6}{a}
	\newnode[below]{7}{a+1}
	\begin{scope}
		\clipcenterarc(0,0)(\startpoint:\endpoint+\step:\radius)
		\newpropbend{1.5}{5.7}{20}
		\newpropbend{1.3}{6.4}{30}
		\newprop{6.6}{-6}{}
	\end{scope}
	\node at (\zero + 3.5*\step:\radius*0.8){\footnotesize $s$};
	\node at (\zero + 4.8*\step:\radius*0.45){\footnotesize $p$};
	\node at (\zero + 8*\step:\radius*0.65){\footnotesize $q$};
\end{scope}
\end{tikzpicture}
\]
  \caption{In order to obtain $I_{i-1}(s) = a$, propagators $s$ and $p$ must each have an end on edge $(i-2)$.}\label{fig part 4}
\end{figure}
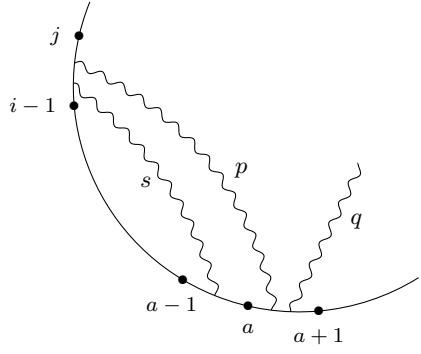

Since $p$ changed its association from $I_{i-1}$ to $I_i$, we have $I_{i-1}(p) = i-1$ by Lemma~\ref{vertex cyclic int lem}. From Figure~\ref{fig part 4} and Remark \ref{rmk algorithm locally same}, $p$ contributes first to  $I_{i-1}$ and proceeds identically to $I_i$ for all vertices inside $p$ in the $<_{i-1}$ ordering, implying that $I_{i-1}(s) = I_i(s)$. Since $I_{i-1}(s) = a$ and $I_i(s) \neq a$, this is a contradiction. 

Thus $q\in S_i$ after all, and so $f$ is a factor of $r_i$ as required.



\item If $W$ has zero propagators then all $I_i=\emptyset$ and both $R$ and $\prod_{i=1}^n \Delta_{I_i}$ are equal to $1$, so the result holds in this case.  Now assume $W$ has at least one propagator.

First we show that every factor of $\prod_{i=1}^n \Delta_{I_i}$ divides $R$.  Take an irreducible factor $f$ of $\prod_{i=1}^n \Delta_{I_i}$. There exists some $i$ such that $f|\Delta_{I_i}$ but $f\!\!\nmid\!\! \Delta_{I_{i-1}}$, since otherwise the variables corresponding to the propagators contributing to $f$ which do not themselves appear in $f$ could never appear, contradicting Lemma~\ref{vertex cyclic int lem}.  If $f$ is a linear factor, say from associating propagator $p$ to vertex $a$, then $I_{i}(p)=a$ and $I_{i-1}(p)\neq a$ so this factor appears in $r_i$.  If $f$ is a quadratic factor, say from associating propagators $p$ and $q$ to vertices $a$ and $a+1$ respectively, then again we cannot have both $I_{i-1}(p) = a$ and $I_{i-1}(q) = a+1$, else $f$ divides $\Delta_{i-1}$. However, by the proof of part (4), if one of $p,q$ belongs to $S_i$ then the other does as well.  Thus $f$ divides $r_i$.

Next we need to show that $R$ is square-free.  Suppose $f^2|R$.  If $f$ is a linear factor, say from associating propagator $p$ to vertex $a$, then there must be two distinct points in the Grassmann necklace algorithm where $p$ changes from not being associated to vertex $a$ to being associated to vertex $a$.  This contradicts Lemma~\ref{vertex cyclic int lem}.  Now suppose $f$ is a quadratic factor, say from propagators $p$ and $q$ supported on edge $a$ with $p$ clockwise to $q$ on the edge.  In this case it is not possible for any $I_i$ to associate $p$ to $a+1$ and $q$ to $a$.  Furthermore, we know by part (4) that $p$ changes from not being associated to $a$ to being associated to $a$ if and only if $q$ changes from not being associated to $a+1$ to being associated to $a+1$.  Thus $f^2|R$ implies that  in the Grassmann necklace $p$ must change twice from not being associated to vertex $a$ to being associated to vertex $a$. This is again a contradiction, and so $R$ is square-free.

{}From the arguments above, we have that $R|\prod_{i=1}^n \Delta_{I_i}$ (from part (4)), $R$ contains all factors of $\prod_{i=1}^n \Delta_{I_i}$, and $R$ is square-free.  Therefore, the ideal generated by $R$ is the radical of the ideal generated by $\prod_{i=1}^n \Delta_{I_i}$.
  \end{enumerate}
\end{proof}

We are now ready to prove the main theorem of this section. We use the facts about $\Delta_{I_i}$ and $R$ above to show that the denominator $R(W)$ of the integrand of $I(W)$ is given by Algorithm \ref{alg WLD to denom via GN}.

\begin{thm}\label{thm denom}
  Given any admissible Wilson loop diagram $W$, let $(I_1, \ldots I_n)$ be the associated Grassmann necklace. Then $R(W)$ as defined in Definition~\ref{def R(W)} is equal to the $R$ constructed by Algorithm~\ref{alg WLD to denom via GN}. Furthermore, the ideal generated by $R(W)$ is the radical of the ideal generated by $\prod_{i=1}^n \Delta_{I_i}$, where $\Delta_{I_i}$ is the determinant of the $k \times k$ minor indicated by $I_i$.  
\end{thm}

\begin{proof}
The equivalence of the ideal generated by $R$ and the radical of the ideal generated by $\prod_{i=1}^n \Delta_{I_i}$ is due to Proposition~\ref{prop alg gives rad}.  It remains to prove that $R(W)$ is the $R$ of Algorithm~\ref{alg WLD to denom via GN}.

To this end, first note that $R(W)$ and $R$ both have total degree $4|\mathcal{P}|$; the degree of $R(W)$ is immediate from the definition while that of $R$ follows from Remark \ref{rmk cyclic}.  By Proposition~\ref{prop alg gives rad} every factor of $R$ is either a quadratic factor coming from two propagators supported on a common edge or a single variable coming from a support vertex of a propagator not appearing in a quadratic factor.  The factors of each $R_e$ making up $R(W)$ in the notation of Definition~\ref{def R(W)} are all of this form and every such factor occurs.  Hence every factor of $R$ divides $R(W)$.
Finally, since $R$ is square-free, this implies that $R(W)$ is a scalar multiple of $R$.

Finally, we need to check the scalar.  By Definition~\ref{def R(W)} each linear factor appears with coefficient $1$ and each $2\times 2$ determinant factor appears with the same sign as the determinant of the corresponding minor in $C(W)$.  Therefore, $R=R(W)$.
\end{proof}

Thus we have shown a deep relationship between the physically derived poles of the integrand $\Omega(W)$ and the geometry and combinatorics associated to the diagrams $W$. This is yet another piece of substantial evidence indicating that the positroid representation of Wilson loop diagrams is the correct geometric object to use when considering the Wilson loop amplitudes geometrically. 


\bibliographystyle{abbrv}
\bibliography{WLDbib}

\end{document}